\documentclass[12pt,reqno,a4paper]{amsart}
\usepackage[dvips]{color}
\usepackage{amsmath}
\usepackage{amsxtra}
\usepackage{amscd}
\usepackage{amsthm}
\usepackage{amsfonts}
\usepackage{amssymb}
\usepackage{nccmath}
\usepackage{eucal}
\usepackage{epsfig}
\usepackage{graphics}
\usepackage{dirtytalk}
\usepackage{ulem}
\usepackage[breakable]{tcolorbox}
\usepackage{ifthen}
\usepackage{varioref}
\usepackage{mathrsfs}

\usepackage{bbold}
\usepackage{mlmodern}

\usepackage{enumerate}
\usepackage{hhline}
\usepackage{multirow}

\usepackage[nosort]{cite}
\usepackage{xspace}
\usepackage{setspace}

\usepackage{titlesec}
\titleformat{\section}[block]{\large\scshape\centering}{§\,\thesection}{1em}{}[\HRule{3pt}] 
\titleformat{\subsection}[block]{\normalsize\bfseries}{\thesubsection}{1em}{}
\titleformat{\subsubsection}[block]{\normalsize\bfseries}{\thesubsubsection}{1em}{}

\usepackage[framemethod=TikZ]{mdframed}

\newcommand{\beq}{\begin{equation}}
	\newcommand{\eeq}{\end{equation}}
\newcommand{\beqa}{\begin{eqnarray}}
	\newcommand{\eeqa}{\end{eqnarray}}

\usepackage[T1]{fontenc}

\newcommand{\bigcdot}{\boldsymbol{\cdot}}
\newcommand*\circled[1]{\tikz[baseline=(char.base)]{
		\node[shape=circle,draw,inner sep=2pt] (char) {#1};}}

\newcommand{\HRule}[1]{\rule{\linewidth}{#1}} 	% Horizontal rule
\newcommand{\comp}{\circ}

\newcommand{\gq}{\mathfrak{q}}

\newcommand{\tens}{\otimes}

\newcommand{\ddsum}{\bigoplus}

\newcommand{\cals}[1]{\mathcal{#1}}
\newcommand{\fraks}[1]{\mathfrak{#1}}

\newcommand*{\defeq}{\stackrel{\text{def}}{=}}
\newcommand{\normord}[1]{:\mathrel{#1}:}
\newcommand{\bb}[1]{\mathbb{#1}}

\numberwithin{equation}{section}
\newtheorem{thm}{Theorem}[section]
\newtheorem{prop}[thm]{Proposition}
\newtheorem{lem}[thm]{Lemma}
\newtheorem{flem}[thm]{Fundamental lemma}
\newtheorem{cor}[thm]{Corollary}
\newtheorem{conj}[thm]{Conjecture}
\newtheorem{dfn}[thm]{Definition}

\newtheorem{exa}[thm]{Example}

\theoremstyle{definition}
\newtheorem{rem}[thm]{Remark}

\usepackage{xpatch}
\xpatchcmd{\proof}{\itshape}{\normalfont\proofnamefont}{}{}
\newcommand{\proofnamefont}{}
\renewcommand{\proofnamefont}{\bfseries}

\begin{document}
	
	\baselineskip = 18pt 
	
	\begin{titlepage}
		
		\bigskip
		\hfill\vbox{\baselineskip12pt
			\hbox{}
		}
		\bigskip
		\begin{center}
			\Large{ \scshape
				\HRule{3pt}
				Elliptic Deformation of the Gaiotto-Rap\v{c}\'{a}k Corner VOA
				and the associated partially symmetric polynoimals 
				%; \\ Intertwiner, $R$ matrix and q-KZ equation 
				\HRule{3pt}
			}
		\end{center}
		\bigskip
		\bigskip
		
		\begin{center}
			\large 
			Panupong Cheewaphutthisakun$^{a}$\footnote{panupong.cheewaphutthisakun@gmail.com},
			Jun'ichi Shiraishi$^{b}$\footnote{shiraish@ms.u-tokyo.ac.jp},
%			and
			Keng Wiboonton$^{a}$\footnote{kwiboonton@gmail.com}
			\\
			\bigskip
			\bigskip
			$^a${\small {\it Department of Mathematics and Computer Science, Faculty of Science, Chulalongkorn University, 
					\\
					Bangkok, 10330, Thailand}}\\
			$^b${\small {\it Graduate School of Mathematical Sciences, University of Tokyo, Komaba, Tokyo 153-8914, Japan}} \\
		\end{center}
		{\vskip 6cm}
		{\small
			\begin{quote}
				\noindent {\textbf{\textit{Abstract}}.}
				We construct the elliptic Miura transformation and use it to obtain the expression of the currents of elliptic corner VOA. 
				We subsequently prove a novel combinatorial formula that is essential for deriving the quadratic relations of the currents.
				In addition, we give a conjecture that relates the correlation function of the currents of elliptic corner VOA to a certain family of partially symmetric polynomials. The elliptic Macdonald polynomials, constructed recently by Awata-Kanno- Mironov-Morozov-Zenkevich, and Fukuda-Ohkubo-Shiraishi, can be obtained as a particular case of this family. 
		\end{quote}}

	\end{titlepage}
	
%	 \tableofcontents
	% \newpage
	
	\section{Introduction}
	\label{intro}
	
%	Keywords : 
%	\begin{itemize}
%		\item Symmetric polynomial
%		\item \sout{Elliptic corner VOA} 
%	\end{itemize}
	
	Quantum deformation of universal enveloping algebra (QUEA) $U_q(\fraks{g})$ was originally proposed by Drinfeld  \cite{DR} and Jimbo \cite{Jim-original} as an algebraic tool for studying quantum integrable systems. The quantization programme is so significant that it has its own place in modern mathematical research. 	
	In the seminal works \cite{AKOS-950} \cite{FF-quantum} \cite{SKAO95} , 
	the quantum deformation program was extended to the family of $W_N$ algebras $(N \geq 2)$ \cite{BPZ}  \cite{BS1993} \cite{Zamo} which is the symmetry algebra of conformal field theories containing higher-spin fields of spin $2, \dots, N$. The quantum $W_N$ algebra was determined by relations called quadratic relations schematically written below:
		\begin{align}
		&f_{r,m}\left(
		\frac{w}{z}
		\right)
		\widetilde{T}_r(z)\widetilde{T}_m(w) 
		- f_{m,r}\left(\frac{z}{w}\right)
		\widetilde{T}_m(w)\widetilde{T}_r(z)
		\notag \\
		&= 
		\sum_{k = 1}^{r}
		\biggl\{
		(\cdots)
		\delta\left(q_3^k\frac{w}{z}\right)
		\widetilde{T}_{r-k}(q_3^{-k}z)
		\widetilde{T}_{m+k}(q_3^kw)
		- 
		(\cdots)
		\delta\left(q_3^{r-m-k}\frac{w}{z}\right)
		\widetilde{T}_{r-k}(z)
		\widetilde{T}_{m+k}(w)
		\biggr\}. 
	\end{align}
	Here $\{\widetilde{T}_i\}$ is the set of generating currents whose expression in terms of the free boson can be deduced from Miura transformation. Note that in the above quadratic relation, we use the convention that for $i > N$, $\widetilde{T}_i = 0$. 
	
%	Motivated from the fact that QUEA possesses a Hopf algebra structure, it is natural to ask if the quantum $W_N$ algebra also has the same structure, or at least, relates to an algebra which possesses this structure. This question was answered in \cite{FHSSY}, in which, 
	
	In \cite{FHSSY},
	it was shown that quantum $W_N$ algebra can be regarded as a truncation of quantum toroidal $\fraks{gl}_1$ algebra \cite{Awata-note} \cite{DI} \cite{Miki} \cite{BS} \cite{Sch} \cite{FHH}. 
	Since it is known from the five-dimensional version \cite{AY} \cite{AY2} \cite{Taki} of AGT correspondence \cite{AGT} \cite{Wyl} that the chiral conformal block of quantum $W_N$ algebra is related to the 
	Nekrasov partition function of five-dimensional $SU(N)$ gauge theories
	(a.k.a. $K$-theoretic Nekrasov partition function) \cite{nekrasov}, the role of quantum toroidal $\fraks{gl}_1$ algebra in gauge theories and string theory has been studied intensively since then \cite{AFS, AKMM16, AKMM16-2, AKMM17, AKMM17-2, AKMM18, BFM16, BFM17, BFM17-2, BJ, B22, CK, fu17, Ma, Noshita, Zen17, Zen18, Zen20, Zen22}. 
	
	It was also shown in \cite{Nieri} \cite{Iqbal} that the similar correspondence can be lifted to six-dimensional gauge theory by considering the elliptic Virasoro algebra (elliptic $W_2$ algebra) which can be deduced by truncating the elliptic toroidal $\fraks{gl}_1$ algebra constructed by Saito \cite{Saito} \cite{Saito2}.
	Here it is important to note that actually there are two elliptic versions of the toroidal $\fraks{gl}_1$ algebra. The first one is the one constructed by Saito as we described previously, while the another one is the one that is obtained through quasi-Hopf twisting \cite{Jimbo:1999zz} \cite{Jimbo:1998bi} (see also \cite{CK2}). However, in this paper, we do not consider the latter one because it is known that the truncation of the quasi-Hopf twisted toroidal $\fraks{gl}_1$ algebra is nothing but the quantum $W_N$ algebra \cite{KO}. 
	
	Recently, a family of vertex operator algebras $Y_{L,M,N}[\Psi] \,\, (L,M,N \in \bb{Z}^{\geq 0}, \Psi \in \bb{C})$, called corner vertex operator algebra (corner VOA), which generalizes the family of $W_N$ algebras was introduced by Gaiotto and Rap\v{c}\'{a}k  by considering the system of $D5, NS5$ and $(-1,-1)$ 5-brane filled out between by three collections of $D3$ branes of numbers $L, M,N$, respectively \cite{GR17} (see also \cite{	EP19, HM ,PR17, PR18	}). More precisely, it was shown that $Y_{0,0,N} = W_N \times \widehat{\fraks{gl}}_1$. Thus, by decoupling $\widehat{\fraks{gl}}_1$, we can reproduce $W_N$ algebra from the corner VOA $Y_{0,0,N}$. 
	
	Motivated from quantization and ellipticization of $W_N$ algebra, it is natural to ask whether it is possible to quantum-mechanically/elliptically deform the corner VOA. The problem of quantum deformation of corner VOA was addressed in \cite{HMNW} in which the trigonometric Miura transformation of quantum corner VOA was proposed. Taking the limit of quantum deformation parameter go to $0$, the Miura transformation of corner VOA was recovered. 
	
	As the quantum/elliptic $W_N$ algebra can be determined by its quadratic relations, it is natural to anticipate that the quantum/elliptic corner VOA can be also determined from its quadratic relations. These quadratic relations was conjectured in \cite{HMNW}, and a proof of a particular case of this conjecture was also demonstrated in the appendix A of the same paper. The authors of \cite{HMNW} have also commented that the fusion relations of quantum corner VOA might be necessary in order to derive the quadratic relation of quantum corner VOA, similar to how the quadratic relation of $W_{q,t}(\fraks{sl}(2|1))$ algebra may be deduced from its fusion relations \cite{Kojima-2021}. The fusion relations mentioned here schematically take the following form: for $r \leq m$, 
	\begin{align}
		&\lim_{w \rightarrow q_3^{-r}z}
		f_{m,r}\left(q_3^{\frac{m-r}{2}}\frac{z}{w} ; p\right)
		\widetilde{T}_{m}(w ; p)\widetilde{T}_{r}(z ; p) 
		\sim 
		\widetilde{T}_{m+r}(z ; p),
	\end{align}
	\begin{align}
		&\lim_{z \rightarrow q_3^{-m}w}
		f_{r,m}\left(q_3^{\frac{r - m}{2}}\frac{w}{z} ; p\right)\widetilde{T}_{r}(z ; p)\widetilde{T}_{m}(w ; p)
		\sim
		\widetilde{T}_{m+r}(w ; p). 
	\end{align}
%	See theorems \ref{thm45-16-1714} and \ref{fusion2-1615} for the precise formulation of the fusion relations. 

	To our knowledge, the fusion relations of the quantum/elliptic corner VOA has not yet been rigorously deduced. 
	In this paper, we first construct the elliptic Miura transformation for elliptic corner VOA (\textbf{Definition \ref{dfellcorner}}), and then derive the fusion relations of the currents obtained from this transformation (\textbf{Theorems \ref{thm45-16-1714}} and \textbf{\ref{fusion2-1615}}). During the derivation, we can see that the combinatorial formula given in \textbf{Lemma \ref{thm44-1334}} is necessary. This lemma is a fundamental result of this paper. By taking the limit $p \rightarrow 0$ of these fusion relations, we obtain the fusion relations of quantum corner VOA. Also, we apply these fusion relations to derive the quadratic relations of elliptic corner VOA (\textbf{Theorem \ref{thm51-1228}}), following the idea of \cite{Kojima-2021}. This is the first goal of this paper. 
	
	The second goal of this paper is to state the relation between the correlation function of the elliptic corner VOA and a family of partially symmetric polynomials (\textbf{Conjecture \ref{conj64-1623}}). We also conjecture that, under a certain appropriate map, the partially symmetric polynomials corresponding to a certain class of elliptic corner VOA will reduce to the elliptic Macdonald polynomials constructed recently in \cite{awata2020-2} \cite{fu20} \cite{Zen21} (\textbf{Conjecture \ref{conj68-1258}}). 
	
	\textit{
		Note add : While preparing this paper, we realize the work  \cite{Kojima-2023}, which attempts to deduce a quadratic relation of quantum corner VOA (the $p \rightarrow 0$ limit of \textbf{Theorem \ref{thm51-1228}}). Nonetheless, there are certain distinctions between the approach used in our paper and \cite{Kojima-2023}. In  \cite{Kojima-2023}, a certain form of fusion relations of the quantum corner VOA were imposed and employed to define the currents $\widetilde{T}_{m}(z)$, while we defined the currents by using Miura transformation, and then derived the fusion relations. In our approach, the combinatorial formula in \textbf{Lemma \ref{thm44-1334}} is essential.   }
	
	\subsubsection*{Organization of material}
	The content of this paper are organized as follows: In section \ref{sec2}, the definition of elliptic toroidal $\fraks{gl}_1$ algebra and its coproduct formula are reviewed. Then we discuss the horizontal Fock representation of the elliptic toroidal $\fraks{gl}_1$ algebra. Subsequently, we utilize this representation to construct the (decoupled) elliptic vertex operators and compute their operator product. We end this section by introducing the elliptic structure function and discuss some of their properties. 
	In section \ref{sec3VOA}, we define the elliptic Miura transformation, and prove its validity by demonstrating that certain properties that hold for the quantum Miura transformation are still valid. 
	After that, we apply this elliptic Miura transformation to obtain the expression for the currents of the decoupled elliptic corner VOA. In section \ref{sec4fusion}, we provide a complete derivation of the fusion relations which, in turn, are crucial for deriving the quadratic relation of decoupled elliptic corner VOA in section \ref{sec5-quad}. 
	In section \ref{sec6-poly}, we state a conjecture which relates the correlation function of the decoupled elliptic corner VOA $\widetilde{Y}_{0,0,N}$ to the elliptic Macdonald polynomials. We also show that for general decoupled elliptic corner VOA $\widetilde{Y}_{L,M,N}$, we can associate a family of partially symmetric polynomials.

	\subsubsection*{Theta function}
	In this paper we frequently deal with theta function $\theta_p(z)$. So, for convenience of the reader, we collect the definition and some basic properties of the theta function here. The theta function $\theta_p(z)$ is defined by
	\begin{align}
	\theta_p(z) := (z;p)_{\infty}(pz^{-1} ; p)_{\infty}
	= 
	\exp\left(
	-\sum_{n \neq 0}
	\frac{z^n}{n(1 - p^n)}
	\right),
	\end{align}
	where $(x;q)_{\infty} := \prod_{k = 0}^{\infty}(1 - q^kx)$. 
	
	One can show that the theta function $\theta_p(z)$ enjoys the following properties:
	\begin{align}
	\theta_p(z) &= (-z)\theta_p(z^{-1}),
	\\
	\theta_p(p^nz) &= (-z)^{-n}p^{-\frac{n}{2}(n-1)}\theta_p(z). 
	\end{align}
	Next we give the definition of elliptic gamma function $\Gamma(z;q,p)$ \cite{rui1997}. The elliptic gamma function is defined by 
	\begin{align}
		\Gamma(z;q,p) 
		:= 
		\frac{
			(qpz^{-1};q,p)_{\infty}
		}{
			(z;q,p)_{\infty}
		}
		= 
		\exp\left(
		\sum_{n \neq 0}
		\frac{z^n}{n(1 - q^n)(1 - p^n)}
		\right), 
	\end{align}
	where $(x ; q,p)_{\infty} := \prod_{k,\ell = 0}^{\infty}(1 - q^kp^\ell x)$. It can be shown that the elliptic gamma function obeys the following relations : 
	\begin{align}
	\Gamma(qz;q,p) 
	&= 
	\theta_p(z)\Gamma(z;q,p),
	\label{eqn18-19j-1249}
	\\
	\lim_{p \rightarrow 0}
	\Gamma(z;q,p) 
	&=
	\frac{1}{(z ; q)_{\infty}}.
	\label{eqn19-19j-1249}
	\end{align}
	
	Finally, we provide definition of the function $\Theta(u;q,p)_n$. This function appears in the elliptic Miura transformation and the definition of elliptic Macdonald polynomials. 
	\begin{align}
		\Theta(z;q,p)_n := \frac{
			\Gamma(q^nz;q,p)
		}{
			\Gamma(z;q,p)
		}.
	\end{align}
	From \eqref{eqn18-19j-1249}, one can show that 
	\begin{align}
	\Theta(z;q,p)_n
	= 
	\begin{cases}
	\prod_{k = 0}^{n-1}
	\theta_p(q^kz)
	\hspace{1.3cm}
	\text{ for } n \geq 0 
	\\
	\prod_{k = 0}^{n-1}
	\theta_p(q^{-k-1}z)^{-1}
	\hspace{0.3cm}
	\text{ for } n < 0 
	\end{cases}
	\end{align}
	Also, from \eqref{eqn19-19j-1249}, one can see that 
	\begin{align}
	\lim_{p \rightarrow 0}\Theta(z;q,p)_n
	= 
	\frac{
	(z ; q)_{\infty}
	}{
	(q^nz ; q)_{\infty}
	}.
	\end{align}

	\subsubsection*{Convention/Notation}
	\begin{enumerate}[(a)]
		\item For any $b < a$, the summation $\sum_{i = a}^{b}$ means the summation over the empty set, hence it is equal to $0$. Similarly, the product $\prod_{i = a}^{b}$ means the product over the empty set. So, its value is equal to $1$. 
		\item The group of permutation of $n$ letters will be denoted by $S_n$. 
	\end{enumerate}

	\subsubsection*{Acknowledgement}
	We would like to thank Hiroaki Kanno for useful discussion. This research project is supported by the Second Century Fund (C2F), Chulalongkorn University.

\section{Elliptic toroidal $\fraks{gl}_1$ algebra}
\label{sec2}
In this section, we provide a quick review of some basic facts on elliptic toroidal algebra that is required for this paper. For more thorough explanation, we refer the readers to the papers \cite{Saito} \cite{Saito2}. 

\subsection{Definition of elliptic toroidal $\fraks{gl}_1$ algebra}

Throughout this paper, we assume that $q_1,q_2,q_3$ are complex numbers which satisfy the genericity conditions described below: 
\begin{enumerate}
	\item $q_1q_2q_3 = 1$
	\item For any $L,M,N \in \bb{Z}$, if $q_1^Lq_2^Mq_3^N = 1$, then $L = M = N$. 
\end{enumerate}

\begin{dfn}[\cite{Saito}, \cite{Saito2}]
Elliptic toroidal $\fraks{gl}_1$ algebra is defined to be a unital associative algebra generated by 
\begin{align*}
	E_n(p), F_n(p), K^\pm_n(p) \,\, (n \in \bb{Z}), C
\end{align*}
which are subjects to the following defining relations:
\begin{gather}
C \text{ is central element}
\\
K^{\pm}(z;p)K^{\pm}(w;p) = 		K^{\pm}(w;p)K^{\pm}(z;p),
\label{2.3main}
\\
K^{+}(z;p)K^-(w;p) = \frac{\cals{G}(w/C z;p)}{\cals{G}(C w/z;p)} K^-(w;p)K^{+}(z;p),
\\
K^+(z;p)E(w;p) = \cals{G}(w/z;p)E(w;p)K^+(z;p),
\\
K^-(\gq z;p)E(w;p) = \cals{G}(w/z;p)E(w;p)K^-(\gq z;p),
\\
K^+(\gq z;p)F(w;p) = \cals{G}(w/z;p)^{-1}F(w;p)K^+(\gq z;p),
\\
K^-(z;p)F(w;p) = \cals{G}(w/z;p)^{-1}F(w;p)K^-(z;p),
\\
\label{EEexchage}
E(z;p)E(w;p) = \cals{G}(w/z;p)E(w;p)E(z;p),
\end{gather}
\begin{gather}
	F(z;p)F(w;p) = \cals{G}(w/z;p)^{-1}F(w;p)F(z;p),
	\\
	[E(z;p),F(w;p)] = \tilde{g}
	\bigg(
	\delta\Big(\frac{Cw}{z}\Big)K^+(z;p) - \delta\Big(\frac{Cz}{w}\Big)K^-(w;p) 
	\bigg),
	\label{2.11main}
\end{gather}
where
\begin{gather}
	\cals{G}(x;p) := \frac{
		\theta_p(q_1^{-1}x)\theta_p(q_2^{-1}x)\theta_p(q_3^{-1}x)
	}{
		\theta_p(q_1x)\theta_p(q_2x)\theta_p(q_3x)
	}, 
\end{gather}	
and
\begin{align}
	\tilde{g} = \frac{1}{(q_1 - 1)(q_2 - 1)(q_3 - 1)}. 
\end{align}
In this paper, we denote elliptic toroidal $\fraks{gl}_1$ algebra by the notation $U^{\text{ell}}_{q_1,q_2,q_3}(\widehat{\widehat{\fraks{gl}}}_1;p) $.
\end{dfn}

\subsection{Hopf algebra structure}

The elliptic toroidal $\fraks{gl}_1$ algebra possesses a Hopf algebra structure. This indicates the existence of its coproduct, counit, and antipode. But all we need for this paper is the coproduct formula. Therefore, we merely write down the coproduct formula and leave out the counit and antipode formulas. 

\begin{prop}[\cite{Saito} \cite{Zhu}]
The map $\Delta : U^{\text{ell}}_{q_1,q_2,q_3}(\widehat{\widehat{\fraks{gl}}}_1;p) \rightarrow U^{\text{ell}}_{q_1,q_2,q_3}(\widehat{\widehat{\fraks{gl}}}_1;p) \tens U^{\text{ell}}_{q_1,q_2,q_3}(\widehat{\widehat{\fraks{gl}}}_1;p)$
defined by 
\begin{align}
	\Delta\big(E(z;p)\big) &= E(z;p) \tens 1 + K^-(C_1z;p) \tens E(C_1z;p),
	\label{copro1}
	\\
	\Delta\big(F(z;p)\big) &= F(C_2z;p) \tens K^+(C_2z;p) + 1 \tens F(z;p),
	\\
	\Delta\big(K^+(z;p)\big) &= K^+(z;p) \tens K^+(C_1^{-1}z;p),
	\\
	\Delta\big(K^-(z;p)\big) &= K^-(C_2^{-1}z;p) \tens K^-(z;p),
	\label{copro4}
\end{align}
where $C_1 := C \tens 1, C_2 := 1 \tens C$, is an algebra homomorphism. 
\end{prop}

For $n \geq 2$, the algebra homomorphism $\Delta^{(n)}$ is defined recursively by 
\begin{align}
\Delta^{(n)} := (\Delta \tens \underbrace{		1 \tens \cdots \tens 1		}_{n-1} ) \comp \Delta^{(n-1)}. 
\end{align}
Before ending this subsection, we record a proposition which will be useful for future calculation. 

\begin{prop}\mbox{}
\begin{align}
	\Delta^{(n-1)}\left(E(z;p)\right)
	= 
	\sum_{k = 1}^{n}
	\underbrace{		K^-(C_1z;p) \tens \cdots \tens K^-(C_1\cdots C_{k-1}z;p) 		}_{k-1}
	\tens E(C_1\cdots C_{k-1}z;p) \tens 
	\underbrace{	1 \tens \cdots \tens 1				}_{n - k}. 
	\label{427-1149}
\end{align}
\end{prop}

\begin{rem}
It is important to note that, due to Hopf algebra structure, tensor product of representations of the elliptic toroidal $\fraks{gl}_1$ algebra remains a representation of the elliptic toroidal $\fraks{gl}_1$ algebra. 
\end{rem}

\subsection{Horizontal Fock representation}
\label{subsec43-1151}

In this subsection, we review the so-called horizontal Fock representation of elliptic toroidal algebra. This representation plays a very important role in this paper. 

\begin{dfn}[\cite{Saito}]
For each $i \in \{1,2,3\}$, the elliptic Heisenberg algebra $\cals{B}^{(i)}$ is defined to be an algebra generated by $\left\{a^{(i)}_{\pm n}~|~ n \in \bb{Z}^{\geq 1}\right\}$ and $\left\{b^{(i)}_{\pm n}~|~ n \in \bb{Z}^{\geq 1}\right\}$ subject to the following defining relations: 
\begin{align}
	[a_n^{(i)},a_m^{(i)}] &= \frac{n}{\kappa_n}(q_i^{n/2} - q_i^{-n/2})
	(1 - p^{|n|})
	\delta_{n+m,0},
	\\
	[b_n^{(i)},b_m^{(i)}] &= \frac{n}{\kappa_n}(q_i^{n/2} - q_i^{-n/2})\frac{1 - p^{|n|}}{p^{|n|}}
	\delta_{n+m,0},
	\\
	[a_n^{(i)},b_m^{(i)}] &= 0,
\end{align}
where $\kappa_n := (q^n_1 - 1)(q^n_2 - 1)(q^n_3 - 1)$. 
\end{dfn}

Define vector space $\cals{H}^{(i)}$ to be 
\begin{align*}
\cals{H}^{(i)} := \operatorname{span}
\left\{
a^{(i)}_{-\lambda_1}\cdots a^{(i)}_{-\lambda_m}
b^{(i)}_{-\mu_1}\cdots b^{(i)}_{-\mu_n}|0\rangle
\;\middle\vert\;
\begin{array}{@{}l@{}}
(1) \,\,  m, n \in \bb{Z}^{\geq 0}\\
(2) \,\, \lambda_1 \geq \cdots \geq \lambda_m \geq 1\\
(3) \,\, \mu_1 \geq \cdots \geq \mu_n \geq 1
\end{array}
\right\}.
\end{align*}
where $|0\rangle$ is annihilated by the positive mode generators. 
It is clear that $\cals{B}^{(i)}$ acts naturally on the vector space $\cals{H}^{(i)}$. In other word, $\cals{H}^{(i)}$ is a left $\cals{B}^{(i)}$-module.

\begin{prop}[\cite{Saito}]
For each $i \in \{1,2,3\}$, the map $\rho_{H,u}^{(i),\text{ell}} : U^{\text{ell}}_{q_1,q_2,q_3}(\widehat{\widehat{\fraks{gl}}}_1;p)
\rightarrow \operatorname{End}(\cals{H}^{(i)})$ defined by 
\begin{align}
	\rho_{H,u}^{(i),\text{ell}}\left(
	K^+(z;p)
	\right)
	&= 
	\exp\left(- \sum_{n = 1}^{\infty}\frac{p^n}{1 - p^n}\frac{\kappa_n}{n}q^{-n/4}_{i}b^{(i)}_{n}z^n\right)
	\exp\left(
	\sum_{n = 1}^{\infty}\frac{1}{1 - p^n}\frac{\kappa_n}{n}q^{n/4}_ia^{(i)}_{n}z^{-n}
	\right),
	\\
	\rho_{H,u}^{(i),\text{ell}}
	\left(
	K^-(z;p)
	\right)
	&= 
	\exp\left(
	\sum_{n = 1}^{\infty}\frac{p^n}{1 - p^n}\frac{\kappa_n}{n}q^{n/4}_i
	b^{(i)}_{-n}z^{-n}
	\right)
	\exp\left(
	- \sum_{n = 1}^{\infty}\frac{1}{1 - p^n}\frac{\kappa_n}{n}q^{-n/4}_{i}
	a^{(i)}_{-n}z^n
	\right),
	\\
	\rho_{H,u}^{(i),\text{ell}}
	\left(
	E(z;p)
	\right)
	&= u\widetilde{d}_1
	\exp\left(
	\sum_{n = 1}^{\infty}
	\frac{p^n}{1 - p^n}\frac{\kappa_n}{n}
	\frac{q^{n/4}_i}{q^{n/2}_{i} - q^{-n/2}_{i}}
	b^{(i)}_{-n}z^{-n}
	\right)
	\notag \\
	&\hspace{0.3cm}\times
	\exp\left(
	\sum_{n = 1}^{\infty}
	\frac{1}{1 - p^n}\frac{\kappa_n}{n}
	\frac{q^{-n/4}_i}{q^{n/2}_{i} - q^{-n/2}_{i}}
	a^{(i)}_{-n}z^n
	\right)
	\notag \\
	&\hspace{0.3cm}\times
	\exp\left(
	\sum_{n = 1}^{\infty}
	\frac{p^n}{1 - p^n}\frac{\kappa_n}{n}
	\frac{q^{-n/4}_i}{	q^{-n}_{i} - 1			}
	b^{(i)}_{n}z^{n}
	\right)
	\exp\left(
	-\sum_{n = 1}^{\infty}
	\frac{1}{1 - p^n}\frac{\kappa_n}{n}
	\frac{q^{n/4}_i}{q^{n}_{i}  - 1}
	a^{(i)}_{n}z^{-n}
	\right),
	\\
	\rho_{H,u}^{(i),\text{ell}}
	\left(
	F(z;p)
	\right)
	&= 
	u^{-1}\widetilde{d}_2
	\exp\left(
	-\sum_{n = 1}^{\infty}
	\frac{p^n}{1 - p^n}
	\frac{\kappa_n}{n}
	\frac{q^{-n/4}_i}{q^{n/2}_{i} - q^{-n/2}_{i}}
	b^{(i)}_{-n}z^{-n}
	\right)
	\notag \\
	&\hspace{0.3cm}\times
	\exp\left(
	-\sum_{n = 1}^{\infty}\frac{1}{1 - p^n}
	\frac{\kappa_n}{n}
	\frac{q^{n/4}_i}{q^{n/2}_{i} - q^{-n/2}_{i}}
	a^{(i)}_{-n}z^{n}
	\right)
	\notag \\
	&\hspace{0.3cm}\times
	\exp\left(
	\sum_{n = 1}^{\infty}
	\frac{p^n}{1 - p^n}
	\frac{\kappa_n}{n}
	\frac{q^{-n/4}_i}{q^{n/2}_{i} - q^{-n/2}_{i}}
	b^{(i)}_{n}z^{n}
	\right)
	\exp\left(
	\sum_{n = 1}^{\infty}\frac{1}{1 - p^n}
	\frac{\kappa_n}{n}
	\frac{q^{n/4}_i}{q^{n/2}_{i} - q^{-n/2}_{i}}
	a^{(i)}_{n}z^{-n}
	\right),
	\\
	\rho_{H,u}^{(i),\text{ell}}
	\left(
	C
	\right)
	&= q^{1/2}_{i},
\end{align}
where $\widetilde{d}_1, \widetilde{d}_2$ are
\begin{align}
	\widetilde{d}_1 &= 
	\frac{
		(p;p)_{\infty}
	}{
		\theta_p(q_{i^\prime})
		\theta_p(q_{i^{\prime\prime}})
	}, 
	\label{eqn315-2135}
	\\
	\widetilde{d}_2 &=
	\frac{
		(p;p)_{\infty}\theta_p(q^{-1}_i)
	}{
		(q_1 - 1)(q_2 - 1)(q_3 - 1)
	},
	\label{eqn316-2136}
\end{align}
is an algebra homomorphism. Here $i^\prime, i^{\prime\prime}$ stand for the numbers in the set $\{1,2,3\}$ that are different from $i$. This representation is called horizontal Fock representation. 
\label{prp24-1540}
\end{prop}

\begin{rem}
When $i = 3$, the horizontal Fock representation in \textbf{Proposition \ref{prp24-1540}} coincides with those in \cite{Saito}. 
\end{rem}

\begin{prop}
Given $c_1, \dots, c_n \in \{1,2,3\}$ and $u_1,\dots,u_n \in \bb{C}$, the map 
\begin{align}
	\rho^{\vec{c},\text{ell}}_{H,\vec{u}}
	: 
	U^{\text{ell}}_{q_1,q_2,q_3}(\widehat{\widehat{\fraks{gl}}}_1;p)
	\rightarrow \operatorname{End}(\cals{H}^{(c_1)} \tens \cdots \tens \cals{H}^{(c_n)}), 
\end{align}
defined by the following composition 
\begin{align}
	\rho^{\vec{c},\text{ell}}_{H,\vec{u}}
	:= 
	\left(
	\rho_{H,u_1}^{(c_1),\text{ell}} \tens \rho_{H,u_2}^{(c_2),\text{ell}} \tens \cdots \tens \rho_{H,u_n}^{(c_n),\text{ell}}	
	\right) \comp 
	\Delta^{(n-1)}, 
\end{align}
is an algebra homomorphism. 
\end{prop}

\subsection{Elliptic vertex operataors}
Vertex operators of quantum toroidal $\fraks{gl}_1$ algebra are defined via the tensor product of horizontal Fock representations of the quantum toroidal $\fraks{gl}_1$ algebra \cite{FHSSY} \cite{HMNW}. Following this idea, the elliptic vertex operators are defined in a similar way via the tensor product of horizontal Fock representations of the elliptic toroidal algebra.

\begin{dfn}[Elliptic vertex operator]
For each $\vec{c} = (c_1,\dots,c_n)$, the elliptic vertex operators $\Lambda^{\vec{c}}_i(z;p) \,\, (i = 1,\dots,n)$ are defined by the following relation: 
\begin{align}
	\rho^{\vec{c},\text{ell}}_{H,\vec{u}^\prime}\left(E(z;p)\right)
	:=
	\sum_{j = 1}^{n}
	y_j(p)
	\Lambda^{\vec{c}}_j(z;p),
	\label{427-1353}
\end{align}
where 
\begin{align}
\vec{u}^\prime := \left(u^\prime_k\right)_{k = 1}^n = \left(
q^{-1/2}_{c_k}q^{1/2}_{3} 
\frac{
\theta_p(q_1)
\theta_p(q_2)
}{
(p;p)_{\infty}
}
u_k
\right)_{k = 1}^{n}, 
\end{align}
and 
\begin{align}
	y_j(p) 
	:=
	\left(-q^{1/2}_{c_j}q^{1/2}_{3}\right)
	\frac{
		\theta_p\left(q^{-1}_{c_j}\right)
	}{
		\theta_p(q_3)
	}.
\label{232-1124}
\end{align}
\end{dfn}

From the above definition, it can be shown that the elliptic vertex operator $\Lambda^{\vec{c}}_j(z;p)$ takes the following form:
\begin{align}
&\Lambda^{\vec{c}}_j(z;p) 
= 
u_j
\varphi^{- \, (c_1)}(q_{c_1}^{1/2}z;p) \tens \cdots \tens \varphi^{- \, (c_{j-1})}(q_{c_1}^{1/2} \cdots q_{c_{j-1}}^{1/2}z;p) 	
\tens \eta^{(c_j)}(q_{c_1}^{1/2} \cdots q_{c_{j-1}}^{1/2}z;p) \tens 
\underbrace{	1 \tens \cdots \tens 1				}_{n - j}, 
\end{align}
where 
\begin{align}
\varphi^{- \, (i)}(z;p) &= 
\exp\left(
\sum_{n = 1}^{\infty}\frac{p^n}{1 - p^n}\frac{\kappa_n}{n}q^{n/4}_{i}b^{(i)}_{-n}z^{-n}
\right)
\exp\left(
- \sum_{n = 1}^{\infty}\frac{1}{1 - p^n}\frac{\kappa_n}{n}q^{-n/4}_{i}
a^{(i)}_{-n}z^n
\right),
\\
\eta^{(i)}(z;p) &= 
\exp\left(
\sum_{n = 1}^{\infty}
\frac{p^n}{1 - p^n}\frac{\kappa_n}{n}
\frac{q^{n/4}_{i}}{q^{n/2}_{i} - q^{-n/2}_{i}}
b^{(i)}_{-n}z^{-n}
\right)
\exp\left(
\sum_{n = 1}^{\infty}
\frac{1}{1 - p^n}\frac{\kappa_n}{n}
\frac{q^{-n/4}_{i}}{q^{n/2}_{i} - q^{-n/2}_{i}}
a^{(i)}_{-n}z^n
\right)
\notag \\
&\hspace{0.3cm}\times
\exp\left(
\sum_{n = 1}^{\infty}
\frac{p^n}{1 - p^n}\frac{\kappa_n}{n}
\frac{q^{-n/4}_{i}}{	q^{-n}_{i} - 1			}
b^{(i)}_{n}z^{n}
\right)
\exp\left(
-\sum_{n = 1}^{\infty}
\frac{1}{1 - p^n}\frac{\kappa_n}{n}
\frac{q^{n/4}_{i}}{q^{n}_{i}  - 1}
a^{(i)}_{n}z^{-n}
\right). 
\end{align}

\begin{prop}
The elliptic vertex operators satisfy the following relation: 
\mbox{}
\vspace{0.2cm}
\begin{align}
	\Lambda^{\vec{c}}_i(z;p)\Lambda^{\vec{c}}_j(w;p) = 
	\begin{cases}
		\displaystyle
		\frac{
		\theta_p\left(q_1^{-1}\frac{w}{z}\right)
		\theta_p\left(q_2^{-1}\frac{w}{z}\right)
		\theta_p\left(q_3^{-1}\frac{w}{z}\right)
		}{
		\theta_p\left(q_1\frac{w}{z}\right)
		\theta_p\left(q_2\frac{w}{z}\right)
		\theta_p\left(q_3\frac{w}{z}\right)
		}
		\normord{
			\Lambda^{\vec{c}}_i(z;p)\Lambda^{\vec{c}}_j(w;p)
		}
		\text{ for } i < j 
		\\
		\displaystyle
		\frac{
			\theta_p\left(	q_{c^\prime_i}q_{c^{\prime\prime}_i}		\frac{w}{z}	\right)
			\theta_p\left(	\frac{w}{z}	\right)
		}{
			\theta_p\left(	q_{c^\prime_i}	\frac{w}{z}		\right)
			\theta_p\left(	q_{c^{\prime\prime}_i}		\frac{w}{z}		\right)
		}
		\normord{
			\Lambda^{\vec{c}}_i(z;p)\Lambda^{\vec{c}}_i(w;p)
		}
		\text{ for } i = j 
		\\
		\displaystyle 
		\normord{
			\Lambda^{\vec{c}}_i(z;p)\Lambda^{\vec{c}}_j(w;p)
		}
		\text{ for } i > j 
	\end{cases}
\end{align}
\label{prp210-1230-18}
\end{prop}

\begin{cor}
\label{cor211-14}
\mbox{}
For $i < j$, we have 
\begin{align}
\Lambda^{\vec{c}}_i(q_3^{-1}z;p)\Lambda^{\vec{c}}_j(z;p) = 0. 
\end{align}
\end{cor}

\subsection{Decoupled elliptic vertex operator}

In this subsection, we introduce the decoupled elliptic vertex operators. 
In section \ref{sec3VOA} below, we will demonstrate how the currents $T^{\vec{c}}_m(z;p)$ appearing in the fusion relations can be expressed in terms of these decoupled elliptic vertex operators. 

\begin{dfn}[\cite{Nieri}]
Define 
\begin{align}
	t(z;p) := \alpha(z;p)E(z;p)\beta(z;p),
\end{align}
where $\alpha(z;p)$ and $\beta(z;p)$ are
\begin{align}
	\alpha(z;p) &= 
	\exp\left(
	-\sum_{r = 1}^{\infty}\frac{1}{C^r - C^{-r}}\widetilde{b}_{-r}z^r
	\right)
	\exp\left(
	-\sum_{r = 1}^{\infty}\frac{1}{C^r - C^{-r}}\widetilde{b}^\prime_{-r}z^{-r}
	\right),
	\\
	\beta(z;p) &=
	\exp\left(
	\sum_{r = 1}^{\infty}\frac{1}{C^r - C^{-r}}\widetilde{b}_rz^{-r}
	\right)
	\exp\left(
	\sum_{r = 1}^{\infty}\frac{1}{C^r - C^{-r}}\widetilde{b}^{\prime}_{r}z^r
	\right),
\end{align}
provided that $\widetilde{b}_{\pm r}, \widetilde{b}^{\prime}_{\pm r}$ are generators of the elliptic toroidal $\fraks{gl}_1$ algebra defined by 
\begin{align}
	K^+(z;p) &= 
	\exp\left(
	\sum_{r = 1}^{\infty}\widetilde{b}_rC^{r}z^{-r}
	\right)\exp\left(
	-\sum_{r = 1}^{\infty}\widetilde{b}^{\prime}_{r}C^{-r}z^r
	\right),
	\\
	K^-(z;p) &=
	\exp\left(
	-\sum_{r = 1}^{\infty}\widetilde{b}_{-r}z^r
	\right)\exp\left(
	\sum_{r = 1}^{\infty}\widetilde{b}^{\prime}_{-r}z^{-r}
	\right).
\end{align}
\end{dfn}

\begin{dfn}[Decoupled elliptic vertex operator]\mbox{}
For each $\vec{c} = (c_1,\dots,c_n)$, the decoupled elliptic vertex operators $\widetilde{\Lambda}^{\vec{c}}_i(z;p) \,\, (i = 1,\dots,n)$ are defined by the following relation: 
\begin{align}
	\rho^{\vec{c},\text{ell}}_{H,\vec{u}}\left(t(z;p)\right) = \sum_{i = 1}^{n}y_i(p)\widetilde{\Lambda}^{\vec{c}}_i(z;p). 
\end{align}
where $y_i(p) \,\, (i = 1,\dots,n)$ were defined as in \eqref{232-1124}. 
\end{dfn}

Next, in analogy to \textbf{Proposition \ref{prp210-1230-18}}, we establish the relationship between the product of the decoupled elliptic vertex operators and its normal ordering. To express these relations, the factors $f_{r,m}(z;p)$, introduced in the subsequent definition, are useful.

\begin{dfn}
For $r \leq m$ and for $\vec{c} = (c_1,\dots,c_n)$ such that $q_{\vec{c}} := \prod_{k = 1}^{n}q_{c_k} \neq 1$, 
\begin{align}
	f^{\vec{c}}_{r,m}(z;p)
	&:=
	\exp\bigg[
	\sum_{k = 1}^{\infty}\frac{z^k}{k}
	\frac{1}{1 - p^k}
	(q_3^{\frac{r}{2}k} - q_3^{-\frac{r}{2}k})(q_{\vec{c}}^{\frac{k}{2}}q^{-\frac{m}{2}k}_{3}
	- q_{\vec{c}}^{-\frac{k}{2}}q^{\frac{m}{2}k}_{3})
	\frac{
		(q^{\frac{k}{2}}_1 - q^{-\frac{k}{2}}_1)(q^{\frac{k}{2}}_2 - q^{-\frac{k}{2}}_2)
	}{
		(q^{\frac{k}{2}}_{\vec{c}} - q^{-\frac{k}{2}}_{\vec{c}})(q^{\frac{k}{2}}_3 - q^{-\frac{k}{2}}_3)
	}
	\bigg]
	\notag \\
	&\hspace{0.4cm}
	\times
	\exp\bigg[
	\sum_{k = 1}^{\infty}\frac{z^{-k}}{k}
	\frac{p^k}{1 - p^k}
	(q_3^{\frac{r}{2}k} - q_3^{-\frac{r}{2}k})(q_{\vec{c}}^{\frac{k}{2}}q^{-\frac{m}{2}k}_{3}
	- q_{\vec{c}}^{-\frac{k}{2}}q^{\frac{m}{2}k}_{3})
	\frac{
		(q^{\frac{k}{2}}_1 - q^{-\frac{k}{2}}_1)(q^{\frac{k}{2}}_2 - q^{-\frac{k}{2}}_2)
	}{
		(q^{\frac{k}{2}}_{\vec{c}} - q^{-\frac{k}{2}}_{\vec{c}})(q^{\frac{k}{2}}_3 - q^{-\frac{k}{2}}_3)
	}
	\bigg],
\end{align}
and
\begin{align}
	f^{\vec{c}}_{m,r}(z;p)
	&:=
	f^{\vec{c}}_{r,m}(z;p). 
\end{align}
\label{df213-1157-18}
\end{dfn}

\begin{rem}
The factor $f^{\vec{c}}_{r,m}(z;p)$ defined in \textbf{Definition \ref{df213-1157-18}} can be regarded as an elliptic analogue of the factor $f^{\vec{c}}_{r,m}(z)$ defined in \cite{HMNW}. More precisely, if we take the limit $p \rightarrow 0$ of $f^{\vec{c}}_{r,m}(z;p)$, we then obtain $f^{\vec{c}}_{r,m}(z)$. 
\end{rem}

\begin{prop}
The decoupled elliptic vertex operators obey the following relations: 
\begin{align}
	\widetilde{\Lambda}^{\vec{c}}_i(z;p)			\widetilde{\Lambda}^{\vec{c}}_j(w;p)		
	= 
	\begin{cases}
		\displaystyle
		f^{\vec{c}}_{11}\left(\frac{w}{z}	 ; p 	\right)^{-1}
		\Delta\left(	q_3^{\frac{1}{2}}\frac{w}{z}	; p 		\right)
		\normord{		\widetilde{\Lambda}^{\vec{c}}_i(z;p)	 	\widetilde{\Lambda}^{\vec{c}}_j(w;p)						}
		\text{ for } i < j 
		\\
		\displaystyle
		f^{\vec{c}}_{11}\left(\frac{w}{z} ; p \right)^{-1}
		\gamma_{c_i}\left(\frac{w}{z} ; p\right)
		\normord{		\widetilde{\Lambda}^{\vec{c}}_i(z;p)		\widetilde{\Lambda}^{\vec{c}}_i(w;p)						}
		\text{ for } i = j
		\\
		\displaystyle
		f^{\vec{c}}_{11}\left(\frac{w}{z} ; p\right)^{-1}
		\Delta\left(	q_3^{-\frac{1}{2}}\frac{w}{z} ; p			\right)
		\normord{		\widetilde{\Lambda}^{\vec{c}}_i(z ; p)		\widetilde{\Lambda}^{\vec{c}}_j(w ; p)				}
		\text{ for } i > j 
	\end{cases}
\end{align}
where
\begin{align}
	\Delta(z;p) 
	&:= 
	\frac{
		\theta_p(q_1q_3^{\frac{1}{2}}z)
		\theta_p(q^{-1}_1q_3^{-\frac{1}{2}}z)
	}{
		\theta_p(q_3^{\frac{1}{2}}z)
		\theta_p(q_3^{-\frac{1}{2}}z)
	},
	\\
	\gamma_{c_i}(z;p) 
	&:= 
	\frac{
		\theta_p(q_{c_i}z)
		\theta_p(q^{-1}_{c_i}z)
	}{
		\theta_p(q_{3}z)
		\theta_p(q^{-1}_{3}z)
	}. 
\end{align} 
\label{prp216-2149}
\end{prop}

Before ending this section, we present a proposition which give various properties of the factor $f^{\vec{c}}_{r,m}(z ; p)$. These properties will be useful in subsequent sections.

\begin{prop}\mbox{}
\begin{enumerate}[(1)]
	\item \label{217-1}
	For $r \leq m $ we have
	\begin{align}
		f^{\vec{c}}_{r,m}(z ; p) = f^{\vec{c}}_{m,r}(z ; p)  = \prod_{a = 0}^{r - 1}f^{\vec{c}}_{1,m}(q_3^{\frac{r - 1 - 2a}{2}}z ; p).
	\end{align}
	\item \label{217-2} For $m \in \bb{Z}^{\geq 1}$ we obtain that 
	\begin{align}
		f^{\vec{c}}_{1,m}(z;p)
		=
		\frac{
			\prod_{k = 1}^{m}
			f^{\vec{c}}_{1,1}
			\left(	q_3^{\frac{1}{2}(m+1-2k)}z ; p			\right)
		}{
			\prod_{k = 1}^{m-1}
			\Delta\left(	q_3^{\frac{1}{2}(m - 2k)}z	; p 	\right)
		}. 
	\end{align}
	\item For $r \leq s$, we have 
	\begin{gather}
		f^{\vec{c}}_{1,s}(z ;p)f^{\vec{c}}_{r,s}(q_3^{-\frac{r+1}{2}}z ; p) = f^{\vec{c}}_{r+1,s}(q_3^{-\frac{r}{2}}z ; p),
		\label{248-1624-16}
		\\
		f^{\vec{c}}_{1,s}(z ; p)f^{\vec{c}}_{r,s}(q_3^{\frac{r+1}{2}}z ; p) = f^{\vec{c}}_{r+1,s}(q_3^{\frac{r}{2}}z ; p),
		\\
		f^{\vec{c}}_{1,r}(z ; p)\frac{f^{\vec{c}}_{s,r}(q_3^{-\frac{s+1}{2}}z ; p)}{
			f^{\vec{c}}_{s+1,r}(q_3^{-\frac{s}{2}}z ; p)
		}
		= 
		\frac{
			\theta_p( q_1^{\frac{1}{2}}q_2^{-\frac{1}{2}}q_3^{-\frac{r}{2}}z)
			\theta_p( q_1^{-\frac{1}{2}}q_2^{\frac{1}{2}}q_3^{-\frac{r}{2}}z)
		}{
			\theta_p( q_1^{\frac{1}{2}}q_2^{\frac{1}{2}}q_3^{-\frac{r}{2}}z)
			\theta_p( q_1^{-\frac{1}{2}}q_2^{-\frac{1}{2}}q_3^{-\frac{r}{2}}z)
		},
		\label{250-1624-16}
		\\
		f^{\vec{c}}_{1,r}(z ; p)\frac{f^{\vec{c}}_{s,r}(q_3^{\frac{s+1}{2}}z ; p)}{
			f^{\vec{c}}_{s+1,r}(q_3^{\frac{s}{2}}z ; p)
		}
		= 
		\frac{
			\theta_p( q_1^{\frac{1}{2}}q_2^{-\frac{1}{2}}q_3^{\frac{r}{2}}z)
			\theta_p( q_1^{-\frac{1}{2}}q_2^{\frac{1}{2}}q_3^{\frac{r}{2}}z)
		}{
			\theta_p( q_1^{\frac{1}{2}}q_2^{\frac{1}{2}}q_3^{\frac{r}{2}}z)
			\theta_p( q_1^{-\frac{1}{2}}q_2^{-\frac{1}{2}}q_3^{\frac{r}{2}}z)
		},
		\label{251-1624-16}
		\\
		f^{\vec{c}}_{1,r}(z ; p)
		\frac{
			f^{\vec{c}}_{1,s}(q_3^{\frac{r+s}{2}}z ; p)
		}{
			f^{\vec{c}}_{1,r+s}(q_3^{\frac{s}{2}}z ; p)
		}
		= 
		\frac{
			\theta_p( q_1^{\frac{1}{2}}q_2^{-\frac{1}{2}}q_3^{\frac{r}{2}}z)
			\theta_p( q_1^{-\frac{1}{2}}q_2^{\frac{1}{2}}q_3^{\frac{r}{2}}z)
		}{
			\theta_p( q_1^{\frac{1}{2}}q_2^{\frac{1}{2}}q_3^{\frac{r}{2}}z)
			\theta_p( q_1^{-\frac{1}{2}}q_2^{-\frac{1}{2}}q_3^{\frac{r}{2}}z)
		},
		\\
		f^{\vec{c}}_{1,r}(z ; p)\frac{
			f^{\vec{c}}_{1,s}(q_3^{-\frac{r+s}{2}}z ; p)
		}{
			f^{\vec{c}}_{1,r+s}(q_3^{-\frac{s}{2}}z ; p)
		}
		= 
		\frac{
			\theta_p( q_1^{\frac{1}{2}}q_2^{-\frac{1}{2}}q_3^{-\frac{r}{2}}z)
			\theta_p( q_1^{-\frac{1}{2}}q_2^{\frac{1}{2}}q_3^{-\frac{r}{2}}z)
		}{
			\theta_p( q_1^{\frac{1}{2}}q_2^{\frac{1}{2}}q_3^{-\frac{r}{2}}z)
			\theta_p( q_1^{-\frac{1}{2}}q_2^{-\frac{1}{2}}q_3^{-\frac{r}{2}}z)
		}.
	\end{gather}
	\item 
	For $i, j \in \bb{Z}^{\geq 1}$ and for $n \in \bb{Z}^{\geq 0}$ such that $i - n, j+n \in \bb{Z}^{\geq 1}$, we have 
	\begin{align}
		f^{\vec{c}}_{1,i}(z ; p)f^{\vec{c}}_{1,j}(q_3^{\pm \frac{i - j - 2n}{2}}z ; p)
		= 
		f^{\vec{c}}_{1,i-n}(q_3^{\mp \frac{n}{2}}z ; p)
		f^{\vec{c}}_{1,j+n}(q_3^{\pm \frac{i - j - n}{2}}z ; p). 
	\end{align}
\end{enumerate}
\label{217}
\end{prop}

\section{Elliptic corner VOA} 
\label{sec3VOA}

In this section, we give the definition of elliptic Miura transformation and then prove its validity.

\begin{dfn}
For given $L, M, N \in \bb{Z}^{\geq 0}$, define $\vec{c}$ to be 
\begin{align}
	\vec{c} = (c_1,\dots,c_n) = (\underbrace{		1,\dots,1			}_{L},\underbrace{	2,\dots,2			}_{M},\underbrace{	3,\dots,3		}_{N})
	= 
	(1^L2^M3^N). 
\end{align}
The elliptic corner VOA, denoted by $qY^{\text{ell}}_{L,M,N}[\Psi;p]$, is determined by the currents 
$\left\{T^{\vec{c}}_m(z;p)\right\}_{m \in  \bb{Z}^{\geq 0}}$ whose expression can be obtained from the elliptic Miura transformation below: 
\begin{align}
	\normord{
		R^{(c_1)}_{1}(z;p)R^{(c_2)}_{2}(z;p)\cdots R^{(c_n)}_{n}(z;p)
	}
	\,\, 
	= \sum_{m = 0}^{\infty}(-1)^mT^{\vec{c}}_m(z;p)q_3^{-mD_z},
	\label{431-1247}
\end{align}
provided that 
\begin{align}
	R^{(c)}_{i}(z;p) 
	&:=
	\sum_{k = 0}^{\infty}
	(-1)^kq_3^{\frac{1}{2}k^2}
	q_c^{-\frac{k}{2}}
	\frac{
		\Theta(  q_3^{1-k}q_c		; q_3, p)_k
	}{
		\Theta(  q_3 ; q_3, p)_k
	}
	\normord{
		\Lambda_i^{\vec{c}}(z;p)
		\Lambda_i^{\vec{c}}(q_3^{-1}z;p)
		\cdots
		\Lambda_i^{\vec{c}}(q_3^{-(k-1)}z;p)
	}
	q_3^{-kD_z}, 
	\label{232-eqn}
\end{align}
where $\Lambda^{\vec{c}}_i(z;p)$ is the elliptic vertex operator defined in \textbf{Definition \ref{232-1124}}. The function $\Theta(u;q,p)_n$ is defined as in the introduction, while the shift operator $q_3^{D_z}$ acts as follows: for any function $g(z)$, we have $q_3^{D_z}g(z) = g(q_3z)$.
\label{dfellcorner}
\end{dfn}

Taking the limit $p \rightarrow 0$, we obtain the trigonometric Miura transformation for quantum corner VOA $qY_{L,M,N}[\Psi]$ constructed in \cite{HMNW}. From the \textbf{Definition \ref{dfellcorner}}, we see the connection between the horizontal Fock representation of the elliptic toroidal $\fraks{gl}_1$ algebra and the elliptic corner VOA. That is, we use the elliptic vertex operator 
$\Lambda^{\vec{c}}_i(z;p)$ defined by tensor product of horizontal Fock representations in the construction of the elliptic corner VOA. 

Thus, it is natural to ask whether it is possible to obtain the expression of currents $T^{\vec{c}}_m(z;p)$ directly from the tensor product of horizontal Fock representations. In \cite{HMNW}, it was shown that 
\begin{align}
	\rho^{\vec{c}}_{H,\vec{u}}\left(E_m(z)\right) = T^{\vec{c}}_{m}(z).
	\label{eqn34}
\end{align}
where $E_m(z) := E(q_3^{-m+1}z)E(q_3^{-m+2}z)\cdots E(z)$. 
Therefore, if the \textbf{Definition \ref{dfellcorner}} is correct, we should obtain the elliptic version of \eqref{eqn34}.
We show in \textbf{Proposition \ref{prp32-1may}} below that the elliptic version of \eqref{eqn34} holds, and this justify the \textbf{Definition \ref{dfellcorner}}. 

\begin{prop}
\label{prp32-1may}
For each $m \in \bb{Z}^{\geq 1}$, define $E_m(z;p) := E(q_3^{-m+1}z;p)E(q_3^{-m+2}z;p)\cdots E(z;p)$. 
Then, we get that 
\begin{align}
\rho^{\vec{c},\text{ell}}_{H,\vec{u}}\left(E_m(z;p)\right) = T^{\vec{c}}_{m}(z;p). 
\end{align}
\end{prop}
\begin{proof}
\begin{align}
	\rho^{\vec{c},\text{ell}}_{H,\vec{u}}\left(E_m(z;p)\right) 
	&= 
	\rho^{\vec{c},\text{ell}}_{H,\vec{u}}\left(E(q_3^{-m+1}z;p)\right)
	\rho^{\vec{c},\text{ell}}_{H,\vec{u}}\left(E(q_3^{-m+2}z;p)\right)
	\cdots
	\rho^{\vec{c},\text{ell}}_{H,\vec{u}}\left(E(z;p)\right)
	\notag \\
	&= 
	\sum_{i_{m} = 1}^{n}
	\cdots
	\sum_{i_1 = 1}^{n}
	y_{i_{m}}(p) \cdots y_{i_1}(p)
	\Lambda^{\vec{c}}_{i_{m}}(q_3^{-m+1}z;p) \cdots \Lambda^{\vec{c}}_{i_1}(z;p)
	\notag \\
	&= 
	\underbrace{	\sum_{i_{m} = 1}^{n}	\cdots
		\sum_{i_1 = 1}^{n}			}_{
		i_m \geq \cdots \geq i_2 \geq i_1
	}
	y_{i_{m}}(p) \cdots y_{i_1}(p)
	\Lambda^{\vec{c}}_{i_{m}}(q_3^{-m+1}z;p) \cdots \Lambda^{\vec{c}}_{i_1}(z;p).
	\label{434-1253-5}
\end{align}
In the last equality, we apply corollary \ref{cor211-14}. We can rewrite \eqref{434-1253-5} as 
\begin{align}
	\rho^{\vec{c},\text{ell}}_{H,\vec{u}}\left(E_m(z;p)\right) = 
	\underbrace{			\sum_{m_1,\dots,m_n \in \bb{Z}^{\geq 0}}		}_{m_1 + \cdots + m_n = m}
	&\left(
	\prod_{i_n = 1}^{m_n}y_n(p)\Lambda^{\vec{c}}_n(q_3^{-m+i_n}z;p)
	\right)
	\times
	\left(
	\prod_{i_{n-1} = 1}^{m_{n-1}}y_{n-1}(p)\Lambda^{\vec{c}}_{n-1}(q_3^{-m+m_n+i_{n-1}}z;p)
	\right)
	\\ \notag 
	&\times
	\cdots
	\left(
	\prod_{i_1 = 1}^{m_1}y_1(p)\Lambda^{\vec{c}}_1(q_3^{-m+m_n+\cdots+m_2+i_1}z;p)
	\right).
\end{align}
According to \textbf{Proposition \ref{prp210-1230-18}}, we can show that 
\begin{align}
	\prod_{i_n = 1}^{m_n}y_n(p)\Lambda^{\vec{c}}_n(q_3^{-m+i_n}z;p)
	= 
	\prod_{\ell = 1}^{m_n}
	\left[
	\left(-q^{1/2}_{c_n}q^{1/2}_{3}\right)
	\frac{
		\theta_p\left(q_3^{\ell - 1}q_{c_n}^{-1}\right)
	}{
		\theta_p\left(q_3^{\ell}\right)
	}
	\right]
	\normord{
		\prod_{i_n = 1}^{m_n}\Lambda^{\vec{c}}_n(q_3^{-m+i_n}z;p) 
	}. 
\end{align}
Hence, 
\begin{align}
	&\rho^{\vec{c},\text{ell}}_{H,\vec{u}}\left(E_m(z;p)\right) 
	\notag \\
	&= 
	\underbrace{			\sum_{m_1,\dots,m_n \in \bb{Z}^{\geq 0}}		}_{m_1 + \cdots + m_n = m}
	\left(
	\prod_{k = 1}^{n}
	\prod_{\ell_k = 1}^{m_k}
	\left(-q^{1/2}_{c_k}q^{1/2}_{3}\right)
	\frac{
		\theta_p\left(q_3^{\ell_k - 1}q_{c_k}^{-1}\right)
	}{
		\theta_p\left(q_3^{\ell_k}\right)
	}
	\right)
	\normord{
		\prod_{k = 1}^{n}
		\prod_{i_k = 1}^{m_k}
		\Lambda^{\vec{c}}_k(q_3^{- \sum_{j = 1}^{k}m_j+ i_k}z;p)
	}. 
\end{align}
In the above equation, we used the fact that the product 
\begin{align}
\normord{
	\prod_{i_n = 1}^{m_n}\Lambda^{\vec{c}}_n(q_3^{-m+i_n}z;p) 
}
\times
\cdots
\times
\normord{
	\prod_{i_1 = 1}^{m_1}\Lambda^{\vec{c}}_1(q_3^{-m+m_n+\cdots+m_2+i_1}z;p)
}
\end{align}
is already in normal-ordered form. Thus, 
\begin{align}
	&\normord{
		\prod_{i_n = 1}^{m_n}\Lambda^{\vec{c}}_n(q_3^{-m+i_n}z;p) 
	}
	\times
	\cdots
	\times
	\normord{
		\prod_{i_1 = 1}^{m_1}\Lambda^{\vec{c}}_1(q_3^{-m+m_n+\cdots+m_2+i_1}z;p)
	}
	\notag	\\
	&= 
	\normord{
		\prod_{i_n = 1}^{m_n}\Lambda^{\vec{c}}_n(q_3^{-m+i_n}z;p) 
		\times
		\cdots
		\times
		\prod_{i_1 = 1}^{m_1}\Lambda^{\vec{c}}_1(q_3^{-m+m_n+\cdots+m_2+i_1}z;p)
	}. 
\end{align}
Since 
\begin{align}
	\normord{
		\prod_{i_k = 1}^{m_k}
		\Lambda^{\vec{c}}_k(q_3^{- \sum_{j = 1}^{k}m_j+ i_k}z;p)
	}
	&=
	\normord{
		\prod_{i_k = 1}^{m_k}
		\Lambda^{\vec{c}}_k(q_3^{- \sum_{j = 1}^{k-1}m_j - (m_k - i_k)}z;p)
	}
	\notag	\\
	&=
	\normord{
		\prod_{i_k = 1}^{m_k}
		\Lambda^{\vec{c}}_k(q_3^{- \sum_{j = 1}^{k-1}m_j -  i_k + 1}z;p)
	},
\end{align}
we obtain that 
\begin{align}
	&\rho^{\vec{c}}_{H,\vec{u}}\left(E_m(z;p)\right) 
	= 
	\underbrace{			\sum_{m_1,\dots,m_n \in \bb{Z}^{\geq 0}}		}_{m_1 + \cdots + m_n = m}
	\left(
	\prod_{k = 1}^{n}
	\prod_{\ell_k = 1}^{m_k}
	\left(-q^{1/2}_{c_k}q^{1/2}_{3}\right)
	\frac{
		\theta_p\left(q_3^{\ell_k - 1}q_{c_k}^{-1}\right)
	}{
		\theta_p\left(q_3^{\ell_k}\right)
	}
	\right)
	\normord{
		\prod_{k = 1}^{n}
		\prod_{i_k = 1}^{m_k}
		\Lambda^{\vec{c}}_k(q_3^{- \sum_{j = 1}^{k-1}m_j -  i_k + 1}z;p)
	}.
\end{align}
On the other side, according to \eqref{431-1247}, we can show that 
\begin{align}
	T^{\vec{c}}_m(z;p)
	&=
	\underbrace{			\sum_{k_1, \dots,k_n \in \bb{Z}^{\geq 0}}			}_{k_1 + \cdots + k_n = m}
	\left[
	\prod_{i = 1}^{n}
	\prod_{j_i = 1}^{k_i}
	\left(
	-(q^{1/2}_{c_i}q^{1/2}_{3})
	\frac{\theta_p(q_3^{j_i-1}q_{c_i}^{-1})}{\theta_p(q_3^{j_i})}
	\right)
	\right]
	\normord{
		\prod_{i = 1}^{n}\prod_{j_i = 1}^{k_i}
		\Lambda^{\vec{c}}_i(q_3^{- \sum_{\ell = 1}^{i-1}k_\ell -j_i + 1}z;p)
	}.
\end{align}
So we have proved the proposition. 
\end{proof}

Similar to the \textbf{Definition \ref{dfellcorner}}, the decoupled currents $\left\{\widetilde{T}^{\vec{c}}_m(z;p)\right\}_{m \in  \bb{Z}^{\geq 0}}$ can be defined via the relation \eqref{431-1247}, but this time we have to replace the elliptic vertex operator $\Lambda^{\vec{c}}_i(z;p)$ in $R^{(c)}_{i}(z;p)$ by the decoupled elliptic vertex operator $\widetilde{\Lambda}^{\vec{c}}_i(z;p)$. Clearly, 
\begin{align}
	\widetilde{T}^{\vec{c}}_m(z;p)
	&=
	\underbrace{			\sum_{k_1, \dots,k_n \in \bb{Z}^{\geq 0}}			}_{k_1 + \cdots + k_n = m}
	\left[
	\prod_{i = 1}^{n}
	\prod_{j_i = 1}^{k_i}
	\left(
	-(q^{1/2}_{c_i}q^{1/2}_{3})
	\frac{\theta_p(q_3^{j_i-1}q_{c_i}^{-1})}{\theta_p(q_3^{j_i})}
	\right)
	\right]
	\normord{
		\prod_{i = 1}^{n}\prod_{j_i = 1}^{k_i}
		\widetilde{		\Lambda		}^{\vec{c}}_i(q_3^{- \sum_{\ell = 1}^{i-1}k_\ell -j_i + 1}z;p)
	}. 
\label{eqn314-1625}
\end{align}

\begin{dfn}
The algebra determined by the decoupled currents $\left\{\widetilde{T}^{\vec{c}}_m(z;p)\right\}_{m \in  \bb{Z}^{\geq 0}}$, subjected to operator product relations,is called the decoupled elliptic corner VOA, and is denoted by $q\widetilde{Y}^{\text{ell}}_{L,M,N}[\Psi;p]$. 
\end{dfn}

Since our goal is to determine the fusion relations of the decoupled elliptic corner VOA, from now on, we focus primarily to the decoupled elliptic corner VOA $q\widetilde{Y}^{\text{ell}}_{L,M,N}[\Psi;p]$.

\section{Fusion relations of decoupled elliptic corner VOA}
\label{sec4fusion}

In this section, we deduce the two relations which are known as the fusion rules of the elliptic corner VOA. 
Throughout this section we fix a $\vec{c} = (c_1,\dots,c_n)$ such that $q_{\vec{c}} \neq 1$. 

\subsection{Fusion relation (I)}

\begin{lem}\mbox{}
\begin{align}
	&f^{\vec{c}}_{m,r}\left(q_3^{\frac{m-r}{2}}\frac{z}{w} ; p\right)\widetilde{T}^{\vec{c}}_{m}(w ; p)\widetilde{T}^{\vec{c}}_{r}(z ; p) 
	\notag \\
	&= 
	\sum_{j = 1}^{r}
	\underbrace{		\sum_{p_1,\dots,p_j \in \bb{Z}^{\geq 1}}			}_{1 \leq p_1 < p_2 < \cdots < p_j \leq n}
	\underbrace{	\sum_{m_{p_1},\dots,m_{p_j} \in \bb{Z}^{\geq 1}}				}_{
		\substack{
m_{p_1} + \cdots + m_{p_j} = r
		}
	}
	\underbrace{		\sum_{m^\prime_1, \dots, m^\prime_n \in \bb{Z}^{\geq 0}}			}_{m^\prime_1 + \cdots + m^\prime_n = m}
	\biggl\{
	A(m_{p_1},c_{p_1} ; p)\cdots A(m_{p_j},c_{p_j} ; p)
	\times
	\prod_{k = 1}^{n}A(m^\prime_k,c_k ; p)
	\notag \\
	&\hspace{0.3cm}\times \prod_{\ell = 1}^{j}
	\prod_{u_{p_\ell} = 1}^{m_{p_\ell}}
	\textbf{O}^{\prime \,\, (m^\prime_1,\dots,m^\prime_n)}_{p_\ell}
	\left(
	(q_3^{-1})^{\sum_{k = 1}^{\ell - 1}m_{p_k} + u_{p_\ell} - 1}
	\frac{z}{w} ; p
	\right)
	\notag \\
	&\hspace{0.3cm}\times
	\normord{
		\prod_{\ell = 1}^{n}
		\prod_{u_\ell = 1}^{m^\prime_\ell}
		\widetilde{\Lambda}_\ell^{\vec{c}}
		\left(			
		(q_3^{-1})^{\sum_{k = 1}^{\ell - 1}m^\prime_k + u_\ell - 1}w ; p
		\right)
		\prod_{\ell = 1}^{j}
		\prod_{u_{p_\ell} = 1}^{m_{p_\ell}}
		\widetilde{\Lambda}^{\vec{c}}_{p_\ell}\left((q_3^{-1})^{\sum_{k = 1}^{\ell - 1}m_{p_k} + u_{p_\ell} - 1			}z ; p\right)
	}\biggr\},
	\label{prp14-13mar}
\end{align}
provided that 
\begin{align}
	&\textbf{O}^{\prime \,\, (m^\prime_1,\dots,m^\prime_n)}_{p_\ell}(z/w ; p)
	\notag \\
	&= \frac{1}{
		\prod_{k = 1}^{m-1}
		\Delta\left(	q_3^{m - k}q_3^{-\frac{1}{2}}		\frac{z}{w}	; p \right)
	}
	\notag \\
	&\times
	\prod_{\ell = 1}^{\sum_{k = 1}^{p_\ell-1}m^\prime_k}
	\Delta\left(	q_3^{-\frac{1}{2}}q_3^\ell	\frac{z}{w}	; p\right)
	\times
	\prod_{\ell = \sum_{k = 1}^{p_\ell}m^\prime_k}^{m - 1}
	\Delta\left(	q_3^{-\frac{1}{2}}q_3^\ell	\frac{z}{w} ; p	\right)
	\times
	\prod_{u = 1}^{m^\prime_{p_\ell}}\gamma_{c_{p_\ell}}\left(
	q^{\sum_{k = 1}^{p_\ell - 1}m^\prime_k + u - 1}_{3}
	\frac{z}{w} ; p
	\right).
	\label{eqn43-1301-18-4}
\end{align}
\label{lemm41-0920}
\end{lem}
\begin{proof}
According to \eqref{eqn314-1625}, we know that 
\begin{align}
	\widetilde{T}^{\vec{c}}_{m}(w;p) &= \underbrace{\sum_{m^\prime_{1},\dots,m^\prime_n \in \bb{Z}^{\geq 0}}}_{
		m^\prime_1 + \cdots + m^\prime_n = m
	}
	\left(
	\prod_{k = 1}^{n}A(m^\prime_k,c_k ; p)
	\right)
	\normord{
		\prod_{\ell = 1}^{n}
		\prod_{u_\ell = 1}^{m^\prime_\ell}
		\widetilde{\Lambda}_\ell^{\vec{c}}
		\left(
		(q_3^{-1})^{\sum_{k = 1}^{\ell - 1}m^\prime_k + u_\ell - 1}w ; p
		\right)
	},
	\\
	\widetilde{T}^{\vec{c}}_{r}(z ; p)
	&= 
	\sum_{j = 1}^{r}
	\underbrace{		\sum_{p_1,\dots,p_j \in \bb{Z}^{\geq 1}}			}_{1 \leq p_1 < p_2 < \cdots < p_j \leq n}
	\underbrace{	\sum_{m_{p_1},\dots,m_{p_j} \in \bb{Z}^{\geq 1}}				}_{
		\substack{
m_{p_1} + \cdots + m_{p_j} = r
		}
	}
	\bigg(
	\prod_{\ell = 1}^{j}A(m_{p_\ell},c_{p_\ell} ; p)
	\bigg)
	\normord{
		\prod_{\ell = 1}^{j}
		\prod_{u_{p_\ell} = 1}^{m_{p_\ell}}
		\widetilde{\Lambda}^{\vec{c}}_{p_\ell}\left((q_3^{-1})^{\sum_{k = 1}^{\ell - 1}m_{p_k} + u_{p_\ell} - 1			}z ; p\right)
	},
\end{align}
where 
\begin{align}
	A(m_k,c_k;p)  = 
	\prod_{\ell = 1}^{m_k}
	\left(
	-(q^{1/2}_{c_k}q^{1/2}_{3})
	\frac{\theta_p(q_3^{\ell-1}q_{c_k}^{-1})}{\theta_p(q_3^{\ell})}
	\right).
\end{align}
By employing \textbf{Proposition \ref{prp216-2149}}, we can show that 
\begin{align}
	&\widetilde{T}^{\vec{c}}_{m}(w ; p)\widetilde{T}^{\vec{c}}_{r}(z ; p)
	\notag \\
	&= 
	\sum_{j = 1}^{r}
	\underbrace{		\sum_{p_1,\dots,p_j \in \bb{Z}^{\geq 1}}			}_{1 \leq p_1 < p_2 < \cdots < p_j \leq n}
	\underbrace{	\sum_{m_{p_1},\dots,m_{p_j} \in \bb{Z}^{\geq 1}}				}_{
		\substack{
m_{p_1} + \cdots + m_{p_j} = r
		}
	}
	\underbrace{		\sum_{m^\prime_1, \dots, m^\prime_n \in \bb{Z}^{\geq 0}}			}_{m^\prime_1 + \cdots + m^\prime_n = m}
	\biggl\{
	A(m_{p_1},c_{p_1} ; p)\cdots A(m_{p_j},c_{p_j} ; p)
	\times
	\prod_{k = 1}^{n}A(m^\prime_k,c_k ; p)
	\notag \\
	&\hspace{0.3cm}
	\times
	\prod_{\ell = 1}^{j}
	\prod_{u_{p_\ell} = 1}^{m_{p_\ell}}
	\textbf{Factor}^{\prime \,\, (m^\prime_1,\dots,m^\prime_n)}_{p_\ell}
	\left(
	(q_3^{-1})^{\sum_{k = 1}^{\ell - 1}m_{p_k} + u_{p_\ell} - 1}
	\frac{z}{w} ; p
	\right)
	\notag \\
	&\hspace{0.3cm}
	\times
	\normord{
		\prod_{\ell = 1}^{n}
		\prod_{u_\ell = 1}^{m^\prime_\ell}
		\widetilde{\Lambda}^{\vec{c}}_\ell
		\left(			
		(q_3^{-1})^{\sum_{k = 1}^{\ell - 1}m^\prime_k + u_\ell - 1}w ; p
		\right)
		\prod_{\ell = 1}^{j}
		\prod_{u_{p_\ell} = 1}^{m_{p_\ell}}
		\widetilde{\Lambda}^{\vec{c}}_{p_\ell}\left((q_3^{-1})^{\sum_{k = 1}^{\ell - 1}m_{p_k} + u_{p_\ell} - 1			}z ; p\right)
	}
	\biggr\},
\end{align}
where 
\begin{align}
	&\text{\textbf{Factor}}^{\prime \,\, (m^\prime_1,\dots,m^\prime_n)}_{p_\ell}(z/w ; p)
	\notag \\
	&= \bigg[
	\prod_{\ell = 0}^{\sum_{i = 1}^{n}m^\prime_i - 1}
	f^{\vec{c}}_{11}\left(q^{\ell}_3\frac{z}{w} ; p\right)^{-1}
	\bigg] 
	\times
	\bigg[
	\prod_{\ell = 1}^{p_\ell - 1}\prod_{u_\ell = 1}^{m^\prime_\ell}
	\Delta\left(q^{1/2}_{3}q^{\sum_{k = 1}^{\ell - 1}m^\prime_k + u_\ell - 1}_{3}\frac{z}{w} ; p\right)
	\bigg]
	\notag \\
	&\hspace{0.3cm} \times 
	\bigg[
	\prod_{u = 1}^{m^\prime_{p_\ell}}\gamma_{c_{p_\ell}}\left(
	q^{\sum_{k = 1}^{p_\ell - 1}m^\prime_k + u - 1}_{3}
	\frac{z}{w} ; p
	\right)
	\bigg]
	\times
	\bigg[
	\prod_{\ell = p_\ell+1}^{n}\prod_{u_\ell = 1}^{m^\prime_\ell}
	\Delta\left(q^{-1/2}_{3}q^{\sum_{k = 1}^{\ell - 1}m^\prime_k + u_\ell - 1}_{3}\frac{z}{w} ; p\right)
	\bigg].
\label{eqn44}
\end{align}
According to \textbf{Proposition \ref{217} \eqref{217-1}}, we can show that 
\begin{align}
&f^{\vec{c}}_{m,r}\left(q_3^{\frac{m-r}{2}}\frac{z}{w} ; p\right)
\times
\prod_{\ell = 1}^{j}
\prod_{u_{p_\ell} = 1}^{m_{p_\ell}}
\textbf{Factor}^{\prime \,\, (m^\prime_1,\dots,m^\prime_n)}_{p_\ell}
\left(
(q_3^{-1})^{\sum_{k = 1}^{\ell - 1}m_{p_k} + u_{p_\ell} - 1}
\frac{z}{w} ; p
\right)
\notag \\
&=
\prod_{\ell = 1}^{j}
\prod_{u_{p_\ell} = 1}^{m_{p_\ell}}
f^{\vec{c}}_{1,m}\left(q_3^{\frac{m - 1}{2}}(q_3^{-1})^{\sum_{k = 1}^{\ell - 1}m_{p_k} + u_{p_\ell} - 1}
\frac{z}{w}  ; p\right)
\textbf{Factor}^{\prime \,\, (m^\prime_1,\dots,m^\prime_n)}_{p_\ell}
\left(
(q_3^{-1})^{\sum_{k = 1}^{\ell - 1}m_{p_k} + u_{p_\ell} - 1}
\frac{z}{w} ; p
\right).
\end{align}
Note that here we use the fact that $m_{p_1} + \cdots + m_{p_j} = r$. 
From \textbf{Proposition \ref{217} \eqref{217-2}}, equation \eqref{eqn44} and the fact that $m^\prime_1 + \cdots + m^\prime_n = m$, it is not hard to see that 
\begin{align}
&f^{\vec{c}}_{1,m}\left(q_3^{\frac{m - 1}{2}}\frac{z}{w} ; p\right)
\textbf{Factor}^{\prime \,\, (m^\prime_1,\dots,m^\prime_n)}_{p_\ell}(z/w ; p)
\notag \\
&= \frac{1}{
	\prod_{k = 1}^{m-1}
	\Delta\left(	q_3^{m - k}q_3^{-\frac{1}{2}}		\frac{z}{w}	; p \right)
}
\notag \\
&\hspace{0.3cm}
\times
\prod_{\ell = 1}^{\sum_{k = 1}^{p_\ell-1}m^\prime_k}
\Delta\left(	q_3^{-\frac{1}{2}}q_3^\ell	\frac{z}{w}	; p\right)
\times
\prod_{\ell = \sum_{k = 1}^{p_\ell}m^\prime_k}^{m - 1}
\Delta\left(	q_3^{-\frac{1}{2}}q_3^\ell	\frac{z}{w} ; p	\right)
\times
\prod_{u = 1}^{m^\prime_{p_\ell}}\gamma_{c_{p_\ell}}\left(
q^{\sum_{k = 1}^{p_\ell - 1}m^\prime_k + u - 1}_{3}
\frac{z}{w} ; p
\right).
\end{align}
So we have proved \textbf{Lemma \ref{lemm41-0920}}. 
\end{proof}

\begin{lem}\mbox{}
\begin{align}
	\textbf{O}^{\prime \,\, (m^\prime_1,\dots,m^\prime_n)}_{p_\ell}(z/w ; p)
	= 
	\begin{cases}
		\displaystyle
		\frac{
			\theta_p\left(q_{c_{p_\ell}}q_3^{\sum_{i = 1}^{p_\ell-1}m^\prime_i}\frac{z}{w}\right)
			\theta_p\left(q_{c^\prime_{p_\ell}}q_3^{\sum_{i = 1}^{p_\ell}m^\prime_i}\frac{z}{w}\right)
		}{
			\theta_p\left(q_3^{-1}q_3^{\sum_{i = 1}^{p_\ell-1}m^\prime_i}\frac{z}{w}\right)
			\theta_p\left(q_3^{\sum_{i = 1}^{p_\ell}m^\prime_i}\frac{z}{w}\right)
		}
		\text{ if } c_{p_\ell} = 1, 2
		\\
		\displaystyle
		\frac{
			\theta_p\left(q_1q_3^{\sum_{k = 1}^{p_\ell-1} m^\prime_k}\frac{z}{w}\right)
			\theta_p\left( q_2q_3^{\sum_{k = 1}^{p_\ell-1} m^\prime_k}\frac{z}{w}\right)
		}{
			\theta_p\left( q_3^{\sum_{k = 1}^{p_\ell-1} m^\prime_k}\frac{z}{w}\right)
			\theta_p\left( q_3^{-1}q_3^{\sum_{k = 1}^{p_\ell-1} m^\prime_k}\frac{z}{w}\right)
		}
		\text{ if } c_{p_\ell} = 3 \text{ and } m^\prime_{p_\ell} = 0
		\\
		\displaystyle
		\prod_{k = \sum_{i = 1}^{p_\ell-1}m^\prime_i + 1}^{\sum_{i = 1}^{p_\ell}m^\prime_i - 1}
		\frac{
			\theta_p\left( q_3^k\frac{z}{w}\right)
			\theta_p\left( q_3^{k-1}\frac{z}{w}\right)
		}{
			\theta_p\left( q_1q_3^k\frac{z}{w}\right)
			\theta_p\left( q_2q_3^k\frac{z}{w}\right)
		}
		\text{ if } c_{p_\ell} = 3 \text{ and } m^\prime_{p_\ell} \neq 0
	\end{cases}
\end{align}
\label{lemm42-0950}
\end{lem}
\begin{proof}
The proof follows directly from equation \eqref{eqn43-1301-18-4}.
\end{proof}

Next, we compute the quantity 
\begin{align}
	&\lim_{w \rightarrow q_3^{-r}z}
	\theta_p\left( q_3^{-r}\frac{z}{w}\right)
	f^{\vec{c}}_{m,r}\left(q_3^{\frac{m-r}{2}}\frac{z}{w} ; p\right)\widetilde{T}^{\vec{c}}_{m}(w ; p)\widetilde{T}^{\vec{c}}_{r}(z ; p). 
\end{align}

\begin{prop}
	\mbox{}
	\begin{align}
		&\lim_{w \rightarrow q_3^{-r}z}
		\theta_p\left( q_3^{-r}\frac{z}{w}\right)
		f^{\vec{c}}_{m,r}\left(q_3^{\frac{m-r}{2}}\frac{z}{w} ; p\right)\widetilde{T}^{\vec{c}}_{m}(w ; p)\widetilde{T}^{\vec{c}}_{r}(z ; p) 
		\notag \\
		&=
		\sum_{j = 1}^{r}
		\underbrace{		\sum_{p_1,\dots,p_j \in \bb{Z}^{\geq 1}}			}_{
			\substack{
		1 \leq p_1 < p_2 < \cdots < p_j \leq n
			}
		}
		\underbrace{	\sum_{m_{p_1},\dots,m_{p_j} \in \bb{Z}^{\geq 1}}				}_{
			\substack{
m_{p_1} + \cdots + m_{p_j} = r
			}
		}
		\underbrace{		\sum_{m^\prime_1, \dots, m^\prime_n \in \bb{Z}^{\geq 0}}			}_{
			\substack{
				(1) \,\, m^\prime_1 + \cdots + m^\prime_n = m
				\\
				(2) \sum_{i = 1}^{p_j - 1}m^\prime_i = 0
			}
		}
		\biggl\{
		\prod_{k = 1}^{j-1}	A(m_{p_k},c_{p_k} ; p)		
		\times
		A(m^\prime_{p_j} + m_{p_j},c_{p_j} ; p)
		\notag	\\
		&\hspace{0.4cm}\times
		\underbrace{		\prod_{k = 1}^{n}			}_{k \neq p_1,\dots,p_j}
		A(m^\prime_k,c_k ; p)
		\times
		\prod_{\ell = 1}^{r}
		\frac{
			\theta_p( q_1q_3^\ell)
			\theta_p( q_2q_3^\ell)
		}{
			\theta_p( q_3^\ell)
		}
		\times
		\prod_{\ell = 2}^{r}
		\frac{
			1
		}{
			\theta_p( q_3^{-1}q_3^\ell)
		}
		\notag	\\
		&\hspace{0.4cm}\times
		\normord{
			\prod_{\ell = 1}^{n}
			\prod_{u_\ell = 1}^{m^\prime_\ell}
			\widetilde{\Lambda}_\ell^{\vec{c}}
			\left(
			(q_3^{-1})^{\sum_{k = 1}^{\ell - 1}m^\prime_k + u_\ell - 1}q_3^{-r}z ; p
			\right)
			\prod_{\ell = 1}^{j}
			\prod_{u_{p_\ell} = 1}^{m_{p_\ell}}
			\widetilde{\Lambda}^{\vec{c}}_{p_\ell}\left((q_3^{-1})^{\sum_{k = 1}^{\ell - 1}m_{p_k} + u_{p_\ell} - 1			}z ; p\right)
		}\biggr\}.
	\label{eqn4.6}
	\end{align}
	\label{prp221-1004}
\end{prop}
\begin{proof}
The operation $\lim_{w \rightarrow q_3^{-r}z}
\theta_p\left( q_3^{-r}\frac{z}{w}\right)$ annihilates all terms in $f^{\vec{c}}_{m,r}\left(q_3^{\frac{m-r}{2}}\frac{z}{w} ; p\right)\widetilde{T}^{\vec{c}}_{m}(w ; p)\widetilde{T}^{\vec{c}}_{r}(z ; p) $ except those that contains $\theta_p\left( q_3^{-r}\frac{z}{w}\right)$ in denominator. From \textbf{Lemma \ref{lemm41-0920}} and \textbf{\ref{lemm42-0950}}, we can see that 
\begin{align}
\textbf{\textit{O}}^{\prime \,\, (m^\prime_1,\dots,m^\prime_n)}_{p_\ell}
\left(
(q_3^{-1})^{\sum_{k = 1}^{\ell - 1}m_{p_k} + u_{p_\ell} - 1}
\frac{z}{w} ; p
\right)
\end{align}
will contain $\theta_p\left( q_3^{-r}\frac{z}{w}\right)$ in denominator if and only if $\ell = j, u_{p_j} = m_{p_j}$ and one of the following conditions holds : 
\begin{enumerate}[(a)]
	\item \label{cona} $m^\prime_{p_j} \geq 1$, $c_{p_j} \in \left\{1,2\right\}$, and $\sum_{i = 1}^{p_j - 1}m^\prime_i = 0$,
	\item \label{conb} $m^\prime_{p_j} = 0$, and $\sum_{i = 1}^{p_j - 1}m^\prime_i = 0$.
\end{enumerate}

Motivated from this, we decompose the summation as follows: 
\begin{align}
	&f^{\vec{c}}_{m,r}\left(q_3^{\frac{m-r}{2}}\frac{z}{w} ; p\right)
	\widetilde{T}^{\vec{c}}_{m}(w ; p)\widetilde{T}^{\vec{c}}_{r}(z ; p) 
	\notag \\
	&= 
	\Bigg[
	\sum_{j = 1}^{r}
	\underbrace{		\sum_{p_1,\dots,p_j \in \bb{Z}^{\geq 1}}			}_{
		\substack{
			(1) \,\, 1 \leq p_1 < p_2 < \cdots < p_j \leq n
			\\
			(2) \,\, c_{p_j} \neq 3
		}
	}
	\underbrace{	\sum_{m_{p_1},\dots,m_{p_j} \in \bb{Z}^{\geq 1}}				}_{
		\substack{
m_{p_1} + \cdots + m_{p_j} = r
		}
	}
	\underbrace{		\sum_{m^\prime_1, \dots, m^\prime_n \in \bb{Z}^{\geq 0}}			}_{m^\prime_1 + \cdots + m^\prime_n = m}
	+
	\sum_{j = 1}^{r}
	\underbrace{		\sum_{p_1,\dots,p_j \in \bb{Z}^{\geq 1}}			}_{
		\substack{
			(1) \,\, 1 \leq p_1 < p_2 < \cdots < p_j \leq n
			\\
			(2) \,\, c_{p_j} = 3
		}
	}
	\underbrace{	\sum_{m_{p_1},\dots,m_{p_j} \in \bb{Z}^{\geq 1}}				}_{
		\substack{
m_{p_1} + \cdots + m_{p_j} = r
		}
	}
	\underbrace{		\sum_{m^\prime_1, \dots, m^\prime_n \in \bb{Z}^{\geq 0}}			}_{
		\substack{
			(1) \,\, m^\prime_1 + \cdots + m^\prime_n = m
			\\
			(2) \,\, m^\prime_{p_j} = 0
		}
	}
	\notag \\
	&+ \sum_{j = 1}^{r}
	\underbrace{		\sum_{p_1,\dots,p_j \in \bb{Z}^{\geq 1}}			}_{
		\substack{
			(1) \,\, 1 \leq p_1 < p_2 < \cdots < p_j \leq n
			\\
			(2) \,\, c_{p_j} = 3
		}
	}
	\underbrace{	\sum_{m_{p_1},\dots,m_{p_j} \in \bb{Z}^{\geq 1}}				}_{
		\substack{
m_{p_1} + \cdots + m_{p_j} = r
		}
	}
	\underbrace{		\sum_{m^\prime_1, \dots, m^\prime_n \in \bb{Z}^{\geq 0}}			}_{
		\substack{
			(1) \,\, m^\prime_1 + \cdots + m^\prime_n = m
			\\
			(2) \,\, m^\prime_{p_j} = 1
		}
	}
	\Bigg]
	\biggl\{
	A(m_{p_1},c_{p_1} ; p)\cdots A(m_{p_j},c_{p_j} ; p)
	\notag \\
	&\hspace{0.3cm} 
	\times 
	\prod_{k = 1}^{n}A(m^\prime_k,c_k ; p)
%	\notag	\\
%	&\hspace{0.3cm}
	\times
	\prod_{\ell = 1}^{j}
	\prod_{u_{p_\ell} = 1}^{m_{p_\ell}}
	\textbf{\textit{O}}^{\prime \,\, (m^\prime_1,\dots,m^\prime_n)}_{p_\ell}
	\left(
	(q_3^{-1})^{\sum_{k = 1}^{\ell - 1}m_{p_k} + u_{p_\ell} - 1}
	\frac{z}{w} ; p
	\right)
	\notag	\\
	&\hspace{0.3cm}
	\times
	\normord{
		\prod_{\ell = 1}^{n}
		\prod_{u_\ell = 1}^{m^\prime_\ell}
		\widetilde{\Lambda}_\ell^{\vec{c}}
		\left(
		(q_3^{-1})^{\sum_{k = 1}^{\ell - 1}m^\prime_k + u_\ell - 1}w ; p
		\right)
		\prod_{\ell = 1}^{j}
		\prod_{u_{p_\ell} = 1}^{m_{p_\ell}}
		\widetilde{\Lambda}^{\vec{c}}_{p_\ell}\left((q_3^{-1})^{\sum_{k = 1}^{\ell - 1}m_{p_k} + u_{p_\ell} - 1			}z ; p\right)
	}
	\biggr\}.
\end{align}
Note that there is no need to consider the case $m^\prime_{p_j} \geq 2$. This is because when $m^\prime_{p_j} \geq 2$, we get that $A(m^\prime_{p_j},3 ; p) = 0$. For convenience, we introduce the following notations: 
\begin{align*}
	\circled{1} &= 
	\sum_{j = 1}^{r}
	\underbrace{		\sum_{p_1,\dots,p_j \in \bb{Z}^{\geq 1}}			}_{
		\substack{
			(1) \,\, 1 \leq p_1 < p_2 < \cdots < p_j \leq n
			\\
			(2) \,\, c_{p_j} \neq 3
		}
	}
	\underbrace{	\sum_{m_{p_1},\dots,m_{p_j} \in \bb{Z}^{\geq 1}}				}_{
		\substack{
m_{p_1} + \cdots + m_{p_j} = r
		}
	}
	\underbrace{		\sum_{m^\prime_1, \dots, m^\prime_n \in \bb{Z}^{\geq 0}}			}_{m^\prime_1 + \cdots + m^\prime_n = m}
	\cdots,
	\\
	\circled{2} &= 
	\sum_{j = 1}^{r}
	\underbrace{		\sum_{p_1,\dots,p_j \in \bb{Z}^{\geq 1}}			}_{
		\substack{
			(1) \,\, 1 \leq p_1 < p_2 < \cdots < p_j \leq n
			\\
			(2) \,\, c_{p_j} = 3
		}
	}
	\underbrace{	\sum_{m_{p_1},\dots,m_{p_j} \in \bb{Z}^{\geq 1}}				}_{
		\substack{
m_{p_1} + \cdots + m_{p_j} = r
		}
	}
	\underbrace{		\sum_{m^\prime_1, \dots, m^\prime_n \in \bb{Z}^{\geq 0}}			}_{
		\substack{
			(1) \,\, m^\prime_1 + \cdots + m^\prime_n = m
			\\
			(2) \,\, m^\prime_{p_j} = 0
		}
	}
	\cdots,
	\\
	\circled{3} &= 
	\sum_{j = 1}^{r}
	\underbrace{		\sum_{p_1,\dots,p_j \in \bb{Z}^{\geq 1}}			}_{
		\substack{
			(1) \,\, 1 \leq p_1 < p_2 < \cdots < p_j \leq n
			\\
			(2) \,\, c_{p_j} = 3
		}
	}
	\underbrace{	\sum_{m_{p_1},\dots,m_{p_j} \in \bb{Z}^{\geq 1}}				}_{
		\substack{
m_{p_1} + \cdots + m_{p_j} = r
		}
	}
	\underbrace{		\sum_{m^\prime_1, \dots, m^\prime_n \in \bb{Z}^{\geq 0}}			}_{
		\substack{
			(1) \,\, m^\prime_1 + \cdots + m^\prime_n = m
			\\
			(2) \,\, m^\prime_{p_j} = 1
		}
	}
	\cdots.
\end{align*}
According to the conditions \eqref{cona} \eqref{conb} explained above, it is clear that 
\begin{align}
	\lim_{w \rightarrow q_3^{-r}z}
	\left[
	\theta_p\left( q_3^{-r}\frac{z}{w}\right)\circled{3} 
	\right]
	= 0. 
\end{align}
Also, from these conditions \eqref{cona} \eqref{conb}, we obtain that 
\begin{align}
	&\lim_{w \rightarrow q_3^{-r}z}
	\left[
	\theta_p\left( q_3^{-r}\frac{z}{w}\right)\circled{1} 
	\right]
	\notag \\
	&=
	\sum_{j = 1}^{r}
	\underbrace{		\sum_{p_1,\dots,p_j \in \bb{Z}^{\geq 1}}			}_{
		\substack{
			(1) \,\, 1 \leq p_1 < p_2 < \cdots < p_j \leq n
			\\
			(2) \,\, c_{p_j} \neq 3
		}
	}
	\underbrace{	\sum_{m_{p_1},\dots,m_{p_j} \in \bb{Z}^{\geq 1}}				}_{
		\substack{
m_{p_1} + \cdots + m_{p_j} = r
		}
	}
	\underbrace{		\sum_{m^\prime_1, \dots, m^\prime_n \in \bb{Z}^{\geq 0}}			}_{
		\substack{			
			(1) \,\, m^\prime_1 + \cdots + m^\prime_n = m
			\\
			(2) \,\, \sum_{i = 1}^{p_j - 1} m^\prime_i = 0
		}
	}
	\biggl\{
	A(m_{p_1},c_{p_1} ; p)\cdots A(m_{p_j},c_{p_j} ; p)
	\notag \\
	&\hspace{0.3cm}\times
	\prod_{k = 1}^{n}A(m^\prime_k,c_k ; p)
	\times
	\prod_{\ell = 1}^{m_{p_j}}
	\frac{
		\theta_p\left( q_{c_{p_j}}q_3^\ell\right)
		\theta_p\left( q_{c^\prime_{p_j}}q_3^{m^\prime_{p_j}}q_3^\ell\right)
	}{
		\theta_p\left( q_3^{m^\prime_{p_j}}q_3^\ell\right)
	}
	\times
	\prod_{u_{p_j} = 1}^{m_{p_j} - 1}
	\frac{
		1
	}{
		\theta_p\left( q_3^{-1}		(q_3^{-1})^{\sum_{k = 1}^{j - 1}m_{p_k} + u_{p_j} - 1}
		q_3^r			\right)
	}
	\notag	\\
	&\hspace{0.3cm}\times
	\prod_{\ell = 1}^{j - 1}
	\prod_{u_{p_\ell} = 1}^{m_{p_\ell}}
	\frac{
		\theta_p\left( q_1	(q_3^{-1})^{\sum_{k = 1}^{\ell - 1}m_{p_k} + u_{p_\ell} - 1}
		q_3^r				\right)
		\theta_p\left( q_2		(q_3^{-1})^{\sum_{k = 1}^{\ell - 1}m_{p_k} + u_{p_\ell} - 1}
		q_3^r			\right)
	}{
		\theta_p\left( 	(q_3^{-1})^{\sum_{k = 1}^{\ell - 1}m_{p_k} + u_{p_\ell} - 1}
		q_3^r				\right)
		\theta_p\left( q_3^{-1}		(q_3^{-1})^{\sum_{k = 1}^{\ell - 1}m_{p_k} + u_{p_\ell} - 1}
		q_3^r			\right)
	}
	\notag	\\
	&\hspace{0.3cm}
	\times
	\normord{
		\prod_{\ell = 1}^{n}
		\prod_{u_\ell = 1}^{m^\prime_\ell}
		\widetilde{\Lambda}_\ell^{\vec{c}}
		\left(				
		(q_3^{-1})^{\sum_{k = 1}^{\ell - 1}m^\prime_k + u_\ell - 1}q_3^{-r}z ; p
		\right)
		\prod_{\ell = 1}^{j}
		\prod_{u_{p_\ell} = 1}^{m_{p_\ell}}
		\widetilde{\Lambda}^{\vec{c}}_{p_\ell}\left((q_3^{-1})^{\sum_{k = 1}^{\ell - 1}m_{p_k} + u_{p_\ell} - 1			}z ; p\right)
	}
	\biggr\}.
\end{align}
Since $\sum_{i = 1}^{p_j - 1} m^\prime_i = 0$, we get that 
\begin{align}
\prod_{k = 1}^{n}A(m^\prime_k,c_k ; p) = A(m^{\prime}_{p_j},c_{p_j} ; p) 
\times
\underbrace{		\prod_{d = 1}^{n}		}_{d \neq p_1, \dots, p_j}
A(m^\prime_d,c_d ; p). 
\label{eqn410}
\end{align}
Because of $c_{p_j} \neq 3$, we see that 
\begin{align}
	A(m_{p_j},c_{p_j})  
	A(m^{\prime}_{p_j},c_{p_j} ; p)
	\times
	\prod_{\ell = 1}^{m_{p_j}}
	\frac{
		\theta_p\left( q_{c^\prime_{p_j}}q_3^{m^\prime_{p_j}}q_3^\ell\right)
	}{
		\theta_p\left( q_3^{m^\prime_{p_j}}q_3^\ell\right)
	}
	= 
	A(m^{\prime}_{p_j} + m_{p_j} ,c_{p_j} ; p)
	\times
	\prod_{\ell = 1}^{m_{p_j}}
	\frac{
		\theta_p\left(		 q_{c^\prime_{p_j}}q_3^\ell		\right)
	}{
		\theta_p\left(			 q_3^\ell		\right)
	}.
	\label{eqn411}
\end{align}
Here $c^\prime_{p_j}$ denotes the element of $\{1,2,3\}$ that is different from $c_{p_j}$ and $3$. 
From \eqref{eqn410} and  \eqref{eqn411}, we obtain that 
\begin{align}
	&\lim_{w \rightarrow q_3^{-r}z}
	\left[
	\theta_p\left( q_3^{-r}\frac{z}{w}\right)\circled{1} 
	\right]
	\notag \\
	&=
	\sum_{j = 1}^{r}
	\underbrace{		\sum_{p_1,\dots,p_j \in \bb{Z}^{\geq 1}}			}_{
		\substack{
			(1) \,\, 1 \leq p_1 < p_2 < \cdots < p_j \leq n
			\\
			(2) \,\, c_{p_j} \neq 3
		}
	}
	\underbrace{	\sum_{m_{p_1},\dots,m_{p_j} \in \bb{Z}^{\geq 1}}				}_{
		\substack{
m_{p_1} + \cdots + m_{p_j} = r
		}
	}
	\underbrace{		\sum_{m^\prime_1, \dots, m^\prime_n \in \bb{Z}^{\geq 0}}			}_{
		\substack{			
			(1) \,\, m^\prime_1 + \cdots + m^\prime_n = m
			\\
			(2) \,\, \sum_{i = 1}^{p_j - 1} m^\prime_i = 0
		}
	}
	\biggl\{
	\prod_{\ell = 1}^{j-1}A(m_{p_\ell},c_{p_\ell} ; p)
	\times
	A(m^{\prime}_{p_j} + m_{p_j} ,c_{p_j} ; p)
	\notag	\\
	&\hspace{0.3cm}\times			
	\underbrace{		\prod_{d = 1}^{n}		}_{d \neq p_1, \dots, p_j}
	A(m^\prime_d,c_d ; p)
	\times
	\prod_{\ell = 2}^{r}
	\frac{
		1
	}{
		\theta_p\left( q_3^{-1}		q_3^\ell			\right)
	}
	\times
	\prod_{\ell = 1}^{r}
	\frac{
		\theta_p\left( q_{1}q_3^\ell\right)
		\theta_p( q_{2}q_3^\ell)}{\theta_p( q_3^\ell)}
	\notag	\\
	&\hspace{0.3cm} \times
	\normord{
		\prod_{\ell = 1}^{n}
		\prod_{u_\ell = 1}^{m^\prime_\ell}
		\widetilde{\Lambda}_\ell^{\vec{c}}
		\left(			
		(q_3^{-1})^{\sum_{k = 1}^{\ell - 1}m^\prime_k + u_\ell - 1}q_3^{-r}z ; p
		\right)
		\prod_{\ell = 1}^{j}
		\prod_{u_{p_\ell} = 1}^{m_{p_\ell}}
		\widetilde{\Lambda}^{\vec{c}}_{p_\ell}\left((q_3^{-1})^{\sum_{k = 1}^{\ell - 1}m_{p_k} + u_{p_\ell} - 1			}z ; p\right)
	}\biggr\}.
\label{eqn412}
\end{align} 
Following the same line of argument, we can show that 
\begin{align}
	&\lim_{w \rightarrow q_3^{-r}z}
	\left[
	\theta_p\left( q_3^{-r}\frac{z}{w}\right)\circled{2} 
	\right]
	\notag	\\
	&=
	\sum_{j = 1}^{r}
	\underbrace{		\sum_{p_1,\dots,p_j \in \bb{Z}^{\geq 1}}			}_{
		\substack{
			(1) \,\, 1 \leq p_1 < p_2 < \cdots < p_j \leq n
			\\
			(2) \,\, c_{p_j} = 3
		}
	}
	\underbrace{	\sum_{m_{p_1},\dots,m_{p_j} \in \bb{Z}^{\geq 1}}				}_{
		\substack{
m_{p_1} + \cdots + m_{p_j} = r
		}
	}
	\underbrace{		\sum_{m^\prime_1, \dots, m^\prime_n \in \bb{Z}^{\geq 0}}			}_{
		\substack{
			(1) \,\, m^\prime_1 + \cdots + m^\prime_n = m
			\\
			(2) \sum_{i = 1}^{p_j}m^\prime_i = 0
		}
	}
	\biggl\{
	\prod_{\ell = 1}^{j - 1}A(m_{p_\ell},c_{p_\ell} ; p)					
	\times
	A(m^\prime_{p_j} + m_{p_j},c_{p_j} ; p)
	\notag		\\
	&\hspace{0.3cm}\times
	\underbrace{		\prod_{k = 1}^{n}			}_{k \neq p_1,\dots,p_j}
	A(m^\prime_k,c_k ; p)
	\times
	\prod_{\ell = 1}^{r}
	\frac{
		\theta_p( q_1q_3^\ell)
		\theta_p( q_2q_3^\ell)
	}{
		\theta_p( q_3^\ell)
	}
	\times
	\prod_{\ell = 2}^{r}
	\frac{
		1
	}{
		\theta_p( q_3^{-1}q_3^\ell)
	}
	\notag		\\
	&\hspace{0.3cm}
	\times
	\normord{
		\prod_{\ell = 1}^{n}
		\prod_{u_\ell = 1}^{m^\prime_\ell}
		\widetilde{\Lambda}_\ell^{\vec{c}}
		\left(
		(q_3^{-1})^{\sum_{k = 1}^{\ell - 1}m^\prime_k + u_\ell - 1}q_3^{-r}z ; p
		\right)
		\prod_{\ell = 1}^{j}
		\prod_{u_{p_\ell} = 1}^{m_{p_\ell}}
		\widetilde{\Lambda}^{\vec{c}}_{p_\ell}\left((q_3^{-1})^{\sum_{k = 1}^{\ell - 1}m_{p_k} + u_{p_\ell} - 1			}z ; p\right)
	}
	\biggr\}.
\label{eqn413}	
\end{align}
Summing \eqref{eqn412} and  \eqref{eqn413}, we obtain the RHS of \eqref{eqn4.6}. So we have proved the \textbf{Proposition \ref{prp221-1004}}. 
\end{proof}

Next, we would like to rewrite the RHS of \eqref{eqn4.6} in terms of the current $\widetilde{T}^{\vec{c}}_{m+r}(z ; p)$. To do this, the following lemma is necessary. 

\begin{flem}
Define sets
\begin{align}
	\vec{P}^{(n)}_{m}
	&:=
	\left\{
	(m_1,\dots,m_n) \in (\bb{Z}^{\geq 0})^n
	\;\middle\vert\;
	\begin{array}{@{}l@{}}
		m_1 + \cdots + m_n = m
	\end{array}
	\right\},
	\\
	\bb{I}^{(n)}_{m,p}
	&:=
	\left\{
	(m_1,\dots,m_n) \in (\bb{Z}^{\geq 0})^n
	\;\middle\vert\;
	\begin{array}{@{}l@{}}
		(1) \,\, m_1 + \cdots + m_n = m
		\\
		(2) \,\,	m_1 = \dots = m_{p-1} = 0
		\\
		(3) \,\,	m_p \neq 0
	\end{array}
	\right\},
\end{align}
and for each $a_{j_1}, a_{j_2}, \dots, a_{j_k}	\in \bb{Z}^{\geq 0}$, 
define a map 
$\varphi^{(	a_{j_1}, a_{j_2}, \dots, a_{j_k}		)}_{j_1,j_2,\dots,j_k} : (\bb{Z}^{\geq 0})^n \rightarrow (\bb{Z}^{\geq 0})^n$ 
by
\begin{align}
	\varphi^{(	a_{j_1}, a_{j_2}, \dots, a_{j_k}		)}_{j_1,j_2,\dots,j_k}(m_1,\dots,m_{j_1},\dots,m_{j_2},\dots,m_n) := 
	(m_1,\dots,m_{j_1} + a_{j_1},\dots,m_{j_2} + a_{j_2},\dots,m_n).
\end{align}
Then, 
	\begin{align}
		\vec{P}^{(n)}_{m+r} = 
		\bigsqcup_{k = 1}^{r}
		\bigsqcup_{p = 1}^{n}
		\underbrace{		\bigsqcup_{j_1,j_2,\dots,j_k \in \bb{Z}^{\geq 1}	}		}_{
		\substack{	
			1 \leq j_1 < j_2 < \cdots < j_k \leq p
		}
		}
		\underbrace{	\bigsqcup_{a_{j_1}, a_{j_2}, \dots, a_{j_k} \in \bb{Z}^{\geq 1}}					}_{
			\substack{	
a_{j_1} + a_{j_2} + \cdots + a_{j_k} = r
		}
		}
		\varphi^{(	a_{j_1}, a_{j_2}, \dots, a_{j_k}		)}_{j_1,j_2,\dots,j_k}
		(\bb{I}^{(n)}_{m,p}). 
		\label{eqn48-2057}
	\end{align}
Here the notation $\bigsqcup$ denotes the disjoint union. 
\label{thm44-1334}
\end{flem}
\begin{proof}
The proof of this theorem is relegated to Appendix \ref{appA-1323}. 
\end{proof}

From \textbf{Lemma \ref{thm44-1334}}, we then see that\footnote{
To avoid confusion with elliptic deformation parameter $p$, we rename the index $p$ in equation \eqref{eqn48-2057} to be $b$. 
}
\begin{align}
	&\widetilde{T}^{\vec{c}}_{m+r}(z ; p) 
	\notag \\
	&=
	\sum_{k = 1}^{r}
	\sum_{b = 1}^{n}
	\underbrace{		\sum_{	j_1,j_2,\dots,j_k	\in \bb{Z}^{\geq 1}		}					}_{
	\substack{	
		1 \leq j_1 < j_2 < \cdots < j_k \leq b
	}
	}
	\underbrace{	\sum_{a_{j_1}, a_{j_2}, \dots, a_{j_k} \in \bb{Z}^{\geq 1}}						}_{
	\substack{	
		\\
a_{j_1} + a_{j_2} + \cdots + a_{j_k} = r
	}
	}
	\sum_{
		m_{1},\dots,m_n \in
		\varphi^{(	a_{j_1}, a_{j_2}, \dots, a_{j_k}		)}_{j_1,j_2,\dots,j_k}
		(\bb{I}^{(n)}_{m,b})
	}
	\biggl\{
	\prod_{d = 1}^{n}A(m_d,c_d ; p)
	\notag \\
	&\hspace{0.3cm}
	\times
	\normord{
		\prod_{\ell = 1}^{n}
		\prod_{u_\ell = 1}^{m_\ell}
		\widetilde{\Lambda}^{\vec{c}}_\ell\big(
		(q_3^{-1})^{\sum_{k = 1}^{\ell - 1}m_k + u_\ell - 1}z ; p
		\big)
	}
	\biggr\}. 
\end{align}
For convenience in calculation, we introduce the notation 
\begin{align}
	\textbf{N}^{\vec{c}}(m_1,\dots,m_n;z ; p) 
	:=
	\normord{
		\prod_{\ell = 1}^{n}
		\prod_{u_\ell = 1}^{m_\ell}
		\widetilde{\Lambda}^{\vec{c}}_\ell\big(
		(q_3^{-1})^{\sum_{k = 1}^{\ell - 1}m_k + u_\ell - 1}z ; p
		\big)
	}.
	\label{eqn16-3feb}
\end{align}
Consequently, 
\begin{align}
	&\widetilde{T}^{\vec{c}}_{m+r}(z ; p) 
	\notag \\
	&=
	\sum_{k = 1}^{r}
	\sum_{b = 1}^{n}
	\underbrace{	\sum_{j_1,j_2,\dots,j_k	\in \bb{Z}^{\geq 1}}				}_{
	\substack{	
		1 \leq j_1 < j_2 < \cdots < j_k \leq b
	}
	}
	\underbrace{	\sum_{		a_{j_1}, a_{j_2}, \dots, a_{j_k} \in \bb{Z}^{\geq 1}				}					}_{
	\substack{	
a_{j_1} + a_{j_2} + \cdots + a_{j_k} = r
	}
	}
	\sum_{
		m_{1},\dots,m_n \in \bb{I}^{(n)}_{m,b}
	}
	\biggl\{
	\underbrace{	\prod_{d = 1}^{n}			}_{d \neq j_1,j_2,\dots,j_k} A(m_d,c_d ; p)
	\times
	\prod_{\ell = 1}^{k}A(m_{j_\ell} + a_{j_\ell},c_{j_\ell} ; p)
	\notag	\\
	&\hspace{0.3cm}
	\times
	\textbf{N}^{\vec{c}}(m_1,\dots,m_{j_1} + a_{j_1},\dots,m_{j_k} + a_{j_k}, \dots, m_n;z ; p) 
	\biggr\}.
\end{align}
Since $(m_1,\dots,m_n) \in \bb{I}^{(n)}_{m,b}$, we get that  $m_{j_1} = m_{j_2} = \dots = m_{j_{k-1}} = 0$. Thus,
\begin{align}
	&\widetilde{T}^{\vec{c}}_{m+r}(z ; p) 
	\notag	\\
	&=
	\sum_{k = 1}^{r}
	\sum_{b = 1}^{n}
	\underbrace{	\sum_{		j_1,j_2,\dots,j_k	\in \bb{Z}^{\geq 1}			}					}_{
	\substack{	
		1 \leq j_1 < j_2 < \cdots < j_k \leq b
	}
	}
	\underbrace{		\sum_{		a_{j_1}, a_{j_2}, \dots, a_{j_k} \in \bb{Z}^{\geq 1}					}					}_{
	\substack{	
a_{j_1} + a_{j_2} + \cdots + a_{j_k} = r
	}
	}
	\sum_{
		m_{1},\dots,m_n \in \bb{I}^{(n)}_{m,b}
	}
	\biggl\{
	\underbrace{	\prod_{d = 1}^{n}			}_{d \neq j_1,j_2,\dots,j_k} A(m_d,c_d ; p)
	\times
	\prod_{\ell = 1}^{k-1}A(a_{j_\ell},c_{j_\ell} ; p)
	\notag	\\
	&\hspace{0.3cm}
	\times 
	A(m_{j_k} + a_{j_k},c_{j_k} ; p)
	\times
	\textbf{N}^{\vec{c}}(m_1,\dots,a_{j_1},\dots, a_{j_{k-1}}, \dots, m_{j_k} + a_{j_k}, \dots, m_n;z;p) 
	\biggr\}.
\end{align}
By using the identity 
\begin{align}
	\sum_{b = 1}^{n}
	\underbrace{	\sum_{	j_1,j_2,\dots,j_k	\in \bb{Z}^{\geq 1}			}					}_{
	\substack{	
		\\
		1 \leq j_1 < j_2 < \cdots < j_k \leq b
	}
	}
	= 
	\underbrace{		\sum_{	j_1,j_2,\dots,j_k	\in \bb{Z}^{\geq 1}						}					}_{
	\substack{	
		1 \leq j_1 < j_2 < \cdots < j_k \leq n
	}
	}
	\sum_{b = j_k}^{n}, 
\end{align}
and the fact that 
\begin{align}
	\sum_{b = j_k}^{n}
	\sum_{
		m_{1},\dots,m_n \in \bb{I}^{(n)}_{m,b}
	}
	= 
	\underbrace{		\sum_{m_1,\dots,m_n \in \bb{Z}^{\geq 0 }}			}_{
		\substack{
			(1) \,\, m_1 + \cdots + m_n = m
			\\
			(2) \,\, m_1 + \cdots + m_{j_k - 1} = 0
		}
	},
\end{align}
we obtain that 
\begin{align}
	&\widetilde{T}^{\vec{c}}_{m+r}(z ; p) 
	\notag	\\
	&=
	\sum_{k = 1}^{r}
	\underbrace{	\sum_{j_1,j_2,\dots,j_k	\in \bb{Z}^{\geq 1}}				}_{
		\substack{	
		1 \leq j_1 < j_2 < \cdots < j_k \leq n
	}
	}
	\underbrace{		\sum_{		a_{j_1}, a_{j_2}, \dots, a_{j_k} \in \bb{Z}^{\geq 1}			}						}_{
	\substack{	
		a_{j_1} + a_{j_2} + \cdots + a_{j_k} = r
	}
	}
	\underbrace{		\sum_{m_1,\dots,m_n \in \bb{Z}^{\geq 0 }}			}_{
		\substack{
			(1) \,\, m_1 + \cdots + m_n = m
			\\
			(2) \,\, m_1 + \cdots + m_{j_k - 1} = 0
		}
	}
	\biggl\{
	\underbrace{	\prod_{d = 1}^{n}			}_{d \neq j_1,j_2,\dots,j_k} A(m_d,c_d ; p)
	\times
	\prod_{\ell = 1}^{k-1}A(a_{j_\ell},c_{j_\ell} ; p) 
	\notag	\\
	&\hspace{0.3cm}
	\times 
	A(m_{j_k} + a_{j_k},c_{j_k} ; p)
	\times 
	\textbf{N}^{\vec{c}}(0,\dots,a_{j_1},\dots, a_{j_{k-1}}, \dots, m_{j_k} + a_{j_k}, \dots, m_n;z;p) 
	\biggr\}
	\notag	\\
	&= \sum_{j = 1}^{r}
	\underbrace{	\sum_{p_1,p_2,\dots,p_j \in \bb{Z}^{\geq 1}}					}_{
	\substack{	
		1 \leq p_1 < p_2 < \cdots < p_j \leq n
	}
	}
	\underbrace{	\sum_{		m_{p_1}, m_{p_2}, \dots, m_{p_j} \in \bb{Z}^{\geq 1}			}						}_{
	\substack{	
		m_{p_1} + \cdots + m_{p_j}  = r
	}
	}
	\underbrace{		\sum_{m^\prime_1, \dots, m^\prime_n \in \bb{Z}^{\geq 0}}			}_{
		\substack{
			(1) \,\, m^\prime_1 + \cdots + m^\prime_n = m
			\\
			(2) \sum_{i = 1}^{p_j - 1}m^\prime_i = 0
		}	
	}
	\biggl\{
	\underbrace{	\prod_{d = 1}^{n}			}_{d \neq p_1,p_2,\dots,p_j} A(m^\prime_d,c_d ; p)
	\times
	\prod_{\ell = 1}^{j-1}A(m_{p_\ell},c_{p_\ell} ; p)
	\notag \\
	&\hspace{0.3cm}
	\times 
	A(m^\prime_{p_j} + m_{p_j},c_{p_j} ; p)
	\times
	\textbf{N}^{\vec{c}}(0, \dots ,m_{p_1},\dots, m_{p_{j-1}}, \dots, m^\prime_{p_j} + m_{p_j}, \dots, m^\prime_n;z;p) 
	\biggr\}.
	\label{eqn416-1720-18-4}
\end{align}
In the last equality, we merely rename the summation indices. Comparing \eqref{eqn416-1720-18-4} to \textbf{Proposition \ref{prp221-1004}}, we immediately obtain the following theorem. 

\begin{thm}[Fusion rule (I)]
	For $r \leq m$, 
	\begin{align}
		&\lim_{w \rightarrow q_3^{-r}z}
		\theta_p\left( q_3^{-r}\frac{z}{w}\right)
		f^{\vec{c}}_{m,r}\left(q_3^{\frac{m-r}{2}}\frac{z}{w} ; p\right)
		\widetilde{T}^{\vec{c}}_{m}(w ; p)\widetilde{T}^{\vec{c}}_{r}(z ; p) 
		\notag	\\
		&
		= 
		\frac{\theta_p(q_1^{-1})\theta_p(q_2^{-1})}{\theta_p(q_3)}
		\times
		\prod_{\ell = 1}^{r-1}
		\frac{
			\theta_p( q_2q_3^{-\ell})
			\theta_p( q_1q_3^{-\ell})
		}{
			\theta_p( q_3^{-\ell-1})
			\theta_p( q_3^{-\ell})
		}
		\times
		\widetilde{T}^{\vec{c}}_{m+r}(z ; p). 
	\end{align}
\label{thm45-16-1714}
\end{thm}

\subsection{Fusion relation (II)}

In this subsection, we will omit several proofs since the idea of proofs for the propositions and lemmas in this section is almost the same as those in the previous section. 

\begin{lem}\mbox{}
\begin{align}
	&f^{\vec{c}}_{r,m}\left(q_3^{\frac{r - m}{2}}\frac{w}{z} ; p\right)\widetilde{T}^{\vec{c}}_{r}(z; p )\widetilde{T}^{\vec{c}}_{m}(w ; p)
	\notag \\
	&= 
	\sum_{j = 1}^{r}
	\underbrace{		\sum_{p_1,\dots,p_j \in \bb{Z}^{\geq 1}}			}_{1 \leq p_1 < p_2 < \cdots < p_j \leq n}
	\underbrace{	\sum_{m_{p_1},\dots,m_{p_j} \in \bb{Z}^{\geq 1}}				}_{
		\substack{
m_{p_1} + \cdots + m_{p_j} = r
		}
	}
	\underbrace{		\sum_{m^\prime_1, \dots, m^\prime_n \in \bb{Z}^{\geq 0}}			}_{m^\prime_1 + \cdots + m^\prime_n = m}
	\biggl\{
	A(m_{p_1},c_{p_1} ; p)\cdots A(m_{p_j},c_{p_j} ; p)
	\notag \\
	&\hspace{0.3cm}
	\times
	\prod_{k = 1}^{n}A(m^\prime_k,c_k ; p)
	\times
	\prod_{\ell = 1}^{j}
	\prod_{u_{p_\ell} = 1}^{m_{p_\ell}}
	\widetilde{\textbf{O}}^{(m^\prime_1,\dots,m^\prime_n)}_{p_\ell}
	\left(
	q_3^{\sum_{k = 1}^{\ell - 1}m_{p_k} + u_{p_\ell} - 1}\frac{w}{z} ; p
	\right)
	\notag \\
	&\hspace{0.3cm} 
	\times
	\normord{
		\prod_{\ell = 1}^{j}
		\prod_{u_{p_\ell} = 1}^{m_{p_\ell}}
		\widetilde{\Lambda}^{\vec{c}}_{p_\ell}\left((q_3^{-1})^{\sum_{k = 1}^{\ell - 1}m_{p_k} + u_{p_\ell} - 1			}z ; p\right)
		\prod_{\ell = 1}^{n}
		\prod_{u_\ell = 1}^{m^\prime_\ell}
		\widetilde{\Lambda}_\ell^{\vec{c}}
		\left(
		(q_3^{-1})^{\sum_{k = 1}^{\ell - 1}m^\prime_k + u_\ell - 1}w ; p
		\right)
	}
	\biggr\}.
	\label{336-0958}
\end{align}
where 
\begin{align}
	&\widetilde{\textbf{O}}^{(m^\prime_1,\dots,m^\prime_n)}_{p_\ell}(w/z ; p)
	\notag	\\
	&= 
	\frac{
		1
	}{
		\prod_{k = 1}^{m - 1}\Delta\left(	q_3^{\frac{1}{2}}q_3^{-k}		\frac{w}{z} ; p\right)
	}
	\times
	\prod_{\ell = 1}^{	\sum_{k = 1}^{p_\ell-1}m^\prime_k				}
	\Delta\left(
	q^{1/2}_{3}(q^{-1}_{3})^{\ell}\frac{w}{z} ; p
	\right)
	\times
	\prod_{\ell = \sum_{k = 1}^{p_\ell}m^\prime_k}^{m - 1}
	\Delta\left(
	q^{1/2}_{3}(q^{-1}_{3})^{\ell}\frac{w}{z} ; p
	\right)
	\notag \\
	&\hspace{0.3cm}
	\times
	\prod_{u = 1}^{m^\prime_{p_\ell}}
	\gamma_{c_{p_\ell}}\left((q_3^{-1})^{\sum_{k = 1}^{p_\ell-1}m^\prime_k + u - 1}\frac{w}{z} ; p		\right).
\end{align}
\label{lemm46-1317}
\end{lem}

\begin{lem}
\begin{align}
	\widetilde{\textbf{O}}^{(m^\prime_1,\dots,m^\prime_n)}_{p_\ell}(w/z ; p)
	= 
	\begin{cases}
		\displaystyle
		\frac{
		\theta_p\left( q^{-1}_{c_{p_\ell}}q^{-\sum_{i = 1}^{p_\ell-1}m^\prime_i}_{3}\frac{w}{z}\right)
		\theta_p\left( q_{c^\prime_{p_\ell}}^{-1}q^{-\sum_{i = 1}^{p_\ell}m^\prime_i}_{3}\frac{w}{z}\right)
		}{
		\theta_p\left( q_{3}q^{-\sum_{i = 1}^{p_\ell-1}m^\prime_i}_{3}\frac{w}{z}\right)
		\theta_p\left( q^{-\sum_{i = 1}^{p_\ell}m^\prime_i}_{3}\frac{w}{z}\right)
		}
		\text{ if } c_{p_\ell} = 1, 2
		\\
		\displaystyle
		\frac{
		\theta_p\left( q_1^{-1}q_3^{-\sum_{i = 1}^{p_\ell-1} m^\prime_i}\frac{w}{z}\right)
		\theta_p\left( q_2^{-1}q_3^{-\sum_{i = 1}^{p_\ell-1} m^\prime_i}\frac{w}{z}\right)
		}{
		\theta_p\left( q_3q_3^{-\sum_{i = 1}^{p_\ell-1} m^\prime_i}\frac{w}{z}\right)
		\theta_p\left( q_3^{-\sum_{i = 1}^{p_\ell-1} m^\prime_i}\frac{w}{z}\right)
		}
		\text{ if } c_{p_\ell} = 3 \text{ and } m^\prime_{p_\ell} = 0
		\\
		\displaystyle
		\prod_{
			k = \sum_{i = 1}^{p_\ell-1}m^\prime_i + 1
		}^{\sum_{i = 1}^{p_\ell}m^\prime_i - 1}
		\frac{
		\theta_p\left(	 q_3q_3^{-k}\frac{w}{z}		\right)
		\theta_p\left(	 q_3^{-k}\frac{w}{z}			\right)
		}{
		\theta_p\left(	 q_1^{-1}q_3^{-k}\frac{w}{z}		\right)
		\theta_p\left(	 q_2^{-1}q_3^{-k}\frac{w}{z}		\right)
		}
		\text{ if } c_{p_\ell} = 3 \text{ and } m^\prime_{p_\ell} \neq 0
	\end{cases}
	\label{338-1014-2}
\end{align}
\end{lem}

\begin{prop}\mbox{}
\begin{align}
	&\lim_{z \rightarrow q_3^{-m}w}	\theta_p\left( q_3^{-m}\frac{w}{z}\right)
	f^{\vec{c}}_{r,m}\left(q_3^{\frac{r - m}{2}}\frac{w}{z} ; p\right)\widetilde{T}^{\vec{c}}_{r}(z ; p)\widetilde{T}^{\vec{c}}_{m}(w ; p)
	\notag \\
	&=
	\sum_{j = 1}^{r}
	\underbrace{		\sum_{p_1,\dots,p_j \in \bb{Z}^{\geq 1}}			}_{1 \leq p_1 < p_2 < \cdots < p_j \leq n}
	\underbrace{	\sum_{m_{p_1},\dots,m_{p_j} \in \bb{Z}^{\geq 1}}				}_{
		\substack{
m_{p_1} + \cdots + m_{p_j} = r
		}
	}
	\underbrace{		\sum_{m^\prime_1, \dots, m^\prime_n \in \bb{Z}^{\geq 0}}			}_{
		\substack{
			(1) \,\, m^\prime_1 + \cdots + m^\prime_n = m
			\\
			(2) \,\, m^\prime_1 + \cdots + m^\prime_{p_1} = m
		}	
	}
	\biggl\{
	\prod_{k = 2}^{j}
	A(m_{p_k},c_{p_k} ; p)
	\times
	A(m^\prime_{p_1} + m_{p_1},c_{p_1} ; p)
	\notag \\
	&\times
	\underbrace{		\prod_{d =1}^{n}		}_{d \neq p_1,\dots,p_j}
	A(m^\prime_d,c_d ; p)
	\times
	\prod_{\ell = 1}^{
		r
	} 
	\frac{
	\theta_p\left( q_2				q_3^{\ell}				\right)
	\theta_p\left( q_1		q_3^{\ell}		\right)
	}{
	\theta_p\left( 			q_3^{\ell}					\right)
	}
	\times
	\prod_{\ell = 1}^{
		r- 1
	} 
	\frac{
		1
	}{
	\theta_p\left( 	q_3^\ell  \right)
	}
	\notag \\
	&\times
	\normord{
		\prod_{\ell = 1}^{j}
		\prod_{u_{p_\ell} = 1}^{m_{p_\ell}}
		\widetilde{\Lambda}^{\vec{c}}_{p_\ell}\left((q_3^{-1})^{\sum_{k = 1}^{\ell - 1}m_{p_k} + u_{p_\ell} - 1			}	q_3^{-m}w	; p	\right)
		\prod_{\ell = 1}^{n}
		\prod_{u_\ell = 1}^{m^\prime_\ell}
		\widetilde{\Lambda}^{\vec{c}}_\ell\big(
		(q_3^{-1})^{\sum_{k = 1}^{\ell - 1}m^\prime_k + u_\ell - 1}w ; p
		\big)
	} 
	\biggr\}.
	\label{349-2-1232e1}
\end{align}
\end{prop}

Next, we write the RHS of \eqref{349-2-1232e1} in terms of $T^{(n)}_{m+r}(w ; p)$. According to \textbf{Lemma \ref{thm44-1334}}, we know that 
\begin{align}
\vec{P}^{(n)}_{m+r} = 
\bigsqcup_{k = 1}^{r}
\bigsqcup_{p = 1}^{n}
\underbrace{		\bigsqcup_{j_1,j_2,\dots,j_k \in \bb{Z}^{\geq 1}	}		}_{
	\substack{	
		1 \leq j_1 < j_2 < \cdots < j_k \leq p
	}
}
\underbrace{	\bigsqcup_{a_{j_1}, a_{j_2}, \dots, a_{j_k} \in \bb{Z}^{\geq 1}}					}_{
	\substack{	
a_{j_1} + a_{j_2} + \cdots + a_{j_k} = r
	}
}
\varphi^{(	a_{j_1}, a_{j_2}, \dots, a_{j_k}		)}_{j_1,j_2,\dots,j_k}
(\bb{I}^{(n)}_{m,p}). 
\end{align}
If we define the map $\cals{R} : \left(\bb{Z}^{\geq 0}\right)^n \rightarrow  \left(\bb{Z}^{\geq 0}\right)^n$ by 
$\cals{R}(b_1,\dots,b_n) = (b_n,\dots,b_1)$, we then get that 
\begin{align}
	\vec{P}^{(n)}_{m+r} = 
	\bigsqcup_{k = 1}^{r}
	\bigsqcup_{p = 1}^{n}
	\underbrace{	\bigsqcup_{j_1,j_2,\dots,j_k	\in \bb{Z}^{\geq 1}}				}_{
		\substack{	
			1 \leq j_1 < j_2 < \cdots < j_k \leq p
		}
	}
	\underbrace{	\bigsqcup_{a_{j_1}, a_{j_2}, \dots, a_{j_k} \in \bb{Z}^{\geq 1}		}					}_{
		\substack{
			a_{j_1} + a_{j_2} + \cdots + a_{j_k} = r
		}
	}
	\cals{R}
	\left(
	\varphi^{(	a_{j_1}, a_{j_2}, \dots, a_{j_k}		)}_{j_1,j_2,\dots,j_k}
	(\bb{I}^{(n)}_{m,p})
	\right). 
\end{align}
Therefore, 
\begin{align}
	&\widetilde{T}^{\vec{c}}_{m+r}(w ; p)
	\notag \\
	&=
	\sum_{k = 1}^{r}
	\sum_{b = 1}^{n}
	\underbrace{		\sum_{	j_1, j_2,\dots,j_k	\in \bb{Z}^{\geq 1}	}			}_{
		\substack{
			1 \leq j_1 < j_2 < \dots < j_k \leq b 
		}
	}
	\underbrace{		\sum_{	a_{j_1}, a_{j_2}, \dots, a_{j_k} \in \bb{Z}^{\geq 1}	}					}_{
		\substack{
			a_{j_1} + a_{j_2} + \dots + a_{j_k} = r 
		}
	}
	\sum_{
		(m_1,\dots,m_n) \in \cals{R}
		\left(
		\varphi^{(	a_{j_1}, a_{j_2}, \dots, a_{j_k}		)}_{j_1,j_2,\dots,j_k}
		(\bb{I}^{(n)}_{m,b})
		\right)
	}
	\biggl\{
	\prod_{k = 1}^{n}A(m_k,c_k ; p)
	\notag \\
	&\hspace{0.3cm}\times
	\textbf{N}^{\vec{c}}(m_1,\dots,m_n;w ; p) 
	\biggr\}
%	\notag	\\
%	&=
%	\sum_{k = 1}^{r}
%	\sum_{p = 1}^{n}
%	\underbrace{	\sum_{j_1, j_2,\dots,j_k \in \bb{Z}^{\geq 0}}					}_{
%		\substack{
%			n+1 - p \leq j_k < j_{k-1} < \dots < j_1 \leq n
%		}
%	}
%	\underbrace{		\sum_{	a_{j_1}, a_{j_2}, \dots, a_{j_k} \in \bb{Z}^{\geq 0}		}				}_{
%		\substack{
%			a_{j_1} + a_{j_2} + \dots + a_{j_k} = r 
%		}
%	}
%	\sum_{
%		(m_1,\dots,m_n) \in \widetilde{\bb{I}}^{(m)}_{n+1-p}
%	}
%	\biggl\{
%	\bigg(
%	\underbrace{	\prod_{d = 1}^{n}			}_{d \neq j_k, \dots, j_1}			A(m_d,c_d)
%	\bigg)
%	\bigg[
%	\prod_{\ell = 1}^{k} A(m_{j_\ell} + a_{j_\ell},c_{j_\ell})
%	\bigg]
%	\notag \\
%	&\hspace{0.3cm}
%	\times
%	\textbf{N}(m_1,\dots, m_{j_k} + a_{j_k}, \dots, m_{j_1} + a_{j_1}, \dots m_n;w) 
%	\biggr\}
	\notag \\
	&=
	\sum_{k = 1}^{r}
	\sum_{b = 1}^{n}
	\underbrace{		\sum_{j_1, j_2,\dots,j_k \in \bb{Z}^{\geq 1} }			}_{
		\substack{
			b \leq j_k < j_{k-1} < \dots < j_1 \leq n
		}
	}
	\underbrace{	\sum_{a_{j_1}, a_{j_2}, \dots, a_{j_k} \in \bb{Z}^{\geq 1}		}			}_{
		\substack{
			a_{j_1} + a_{j_2} + \dots + a_{j_k} = r 
		}
	}
	\sum_{
		(m_1,\dots,m_n) \in \widetilde{\bb{I}}^{(n)}_{m,b}
	}
	\biggl\{
	\underbrace{	\prod_{d = 1}^{n}			}_{d \neq j_k, \dots, j_1}			A(m_d,c_d ; p)
	\notag \\
	&\hspace{0.3cm}
	\times
	\prod_{\ell = 1}^{k} A(m_{j_\ell} + a_{j_\ell},c_{j_\ell} ; p)
	\times
	\textbf{N}^{\vec{c}}(m_1,\dots, m_{j_k} + a_{j_k}, \dots, m_{j_1} + a_{j_1}, \dots m_n;w ; p) 
	\biggr\}.
\end{align}
where 
\begin{align}
	\widetilde{\bb{I}}^{(n)}_{m,b}
	&\defeq 
	\left\{
	(m_1,\dots,m_n) \in (\bb{Z}^{\geq 0})^n
	\;\middle\vert\;
	\begin{array}{@{}l@{}}
		(1) \,\, m_1 + \cdots + m_n = m
		\\
		(2) \,\,	m_{b+1} = \cdots = m_n = 0
		\\
		(3) \,\,	m_b \neq 0
	\end{array}
	\right\}.
\end{align}
Using the relations
\begin{align}
	\sum_{b = 1}^{n}
	\underbrace{	\sum_{	j_1, j_2,\dots,j_k \in \bb{Z}^{\geq 1}		}			}_{
		\substack{
			b \leq j_k < j_{k-1} < \dots < j_1 \leq n
		}
	}
	= 
	\underbrace{		\sum_{	j_1, j_2,\dots,j_k \in \bb{Z}^{\geq 1}	}				}_{
		\substack{
			1 \leq j_k < j_{k-1} < \dots < j_1 \leq n
		}
	}
	\sum_{b = 1}^{j_k},
\end{align}
we can show that 
\begin{align}
	&\widetilde{T}^{\vec{c}}_{m+r}(w ; p)
%	\notag \\
%	&=
%	\sum_{k = 1}^{r}
%	\underbrace{		\sum_{	j_1, j_2,\dots,j_k \in \bb{Z}^{\geq 0}	}				}_{
%		\substack{
%			1 \leq j_k < j_{k-1} < \dots < j_1 \leq n
%		}
%	}
%	\underbrace{	\sum_{	a_{j_1}, a_{j_2}, \dots, a_{j_k} \in \bb{Z}^{\geq 0}	}					}_{
%		\substack{
%			a_{j_1} + a_{j_2} + \dots + a_{j_k} = r 
%		}
%	}
%	\underbrace{	\sum_{m_1,\dots,m_n \in \bb{Z}^{\geq 0}		}				}_{
%		\substack{
%			(1) \,\, m_1 + \cdots + m_n = m
%			\\
%			(2) \,\, m_{j_k + 1} = \cdots = m_n = 0
%		}
%	}
%	\biggl\{
%	\bigg(
%	\underbrace{	\prod_{d = 1}^{n}			}_{d \neq j_k, \dots, j_1}			A(m_d,c_d)
%	\bigg)
%	\bigg[
%	\prod_{\ell = 1}^{k} A(m_{j_\ell} + a_{j_\ell},c_{j_\ell})
%	\bigg]
%	\notag \\
%	&\hspace{0.1cm}
%	\textbf{N}(m_1,\dots, m_{j_k} + a_{j_k}, \dots, m_{j_1} + a_{j_1}, \dots m_n;w) 
%	\biggr\}
	\notag \\
	&=
	\sum_{j = 1}^{r}
	\underbrace{	\sum_{p_1,\dots,p_j \in \bb{Z}^{\geq 1}}				}_{
		\substack{
			1 \leq p_1 < p_2 < \dots < p_j \leq n 
		}
	}
	\underbrace{		\sum_{m_{p_1}, \dots, m_{p_j} \in \bb{Z}^{\geq 1}}			}_{
		\substack{
m_{p_1} + \cdots + m_{p_j} = r
		}
	}
	\underbrace{		\sum_{m^\prime_1, \dots, m^\prime_n \in \bb{Z}^{\geq 0}}			}_{
		\substack{
			(1) \,\, m^\prime_1 + \cdots + m^\prime_n = m
			\\
			(2) \,\, m^{\prime}_{p_1 + 1} = \cdots = m^\prime_n = 0 
		}
	}
	\biggl\{
	\underbrace{		\prod_{d =1}^{n}		}_{d \neq p_1,\dots,p_j}
	A(m^\prime_d,c_d ; p)
	\times
	A(m^\prime_{p_1} + m_{p_1},c_{p_1} ; p)
	\notag \\
	&\hspace{0.3cm}
	\times
	\prod_{\ell = 2}^{j}
	A( m_{p_\ell},c_{p_\ell} ; p)
	\times
	\textbf{N}^{\vec{c}}(m^\prime_1,\dots, m^\prime_{p_1} + m_{p_1}, \dots, m_{p_j}, \dots m^\prime_n;w ; p)
	\biggr\}.
	\label{358-1402}
\end{align}
Here we renamed the summation indices. Comparing \eqref{358-1402} to the RHS of \eqref{349-2-1232e1}, we obtain the following theorem. 

\begin{thm}[Fusion rule (II)]
\label{fusion2-1615}
\begin{align}
&\lim_{z \rightarrow q_3^{-m}w}	\theta_p\left( q_3^{-m}\frac{w}{z}\right)
f^{\vec{c}}_{r,m}\left(q_3^{\frac{r - m}{2}}\frac{w}{z} ; p\right)\widetilde{T}^{\vec{c}}_{r}(z ; p)\widetilde{T}^{\vec{c}}_{m}(w ; p)
\notag \\
&
= 
\frac{\theta_p(q_1^{-1})\theta_p(q_2^{-1})}{\theta_p(q_3)}
\times
\prod_{\ell = 1}^{r-1}
\frac{
	\theta_p( q_2q_3^{-\ell})
	\theta_p( q_1q_3^{-\ell})
}{
	\theta_p( q_3^{-\ell-1})
	\theta_p( q_3^{-\ell})
}
\times
\widetilde{T}^{\vec{c}}_{m+r}(w ; p). 
\end{align}
\end{thm}

\section{Quadratic relations of elliptic corner VOA}
\label{sec5-quad}
In  \cite{Kojima-2021}, Kojima used the fusion rules of $W_{q,t}(\fraks{sl}(2|1))$ algebra to deduce the quadratic relations of $W_{q,t}(\fraks{sl}(2|1))$. Following his work, we apply the fusion rules of the elliptic corner VOA that were deduced in the previous section to obtain the quadratic relation of the elliptic corner VOA. This quadratic relation is written in the following theorem:

\begin{thm}
\label{thm51-1228}
\mbox{}
		For $\vec{c} = (c_1,\dots,c_n)$ such that $q_{\vec{c}} \neq 1$, and for any $r, m \in \bb{Z}^{\geq 0} \,\, (r \leq m)$, 
		\begin{align}
			&f^{\vec{c}}_{r,m}\left(
			q^{\frac{r - m}{2}}_{3}\frac{w}{z} ; p
			\right)
			\widetilde{	T	}^{\vec{c}}_r(z ; p)\widetilde{	T		}^{\vec{c}}_m(w ; p) 
			- f^{\vec{c}}_{m,r}\left(
			q^{\frac{m - r}{2}}_{3}\frac{z}{w} ; p
			\right)
			\widetilde{T}^{\vec{c}}_m(w ; p)\widetilde{T}^{\vec{c}}_r(z ; p)
			\notag \\
			&= 
			\frac{
				\theta_p(q_1)\theta_p(q_2)
			}{
				(p;p)^2_{\infty}\theta_p(q_3^{-1})
			}
			\sum_{k = 1}^{r}
			\left(
			\prod_{\ell = 1}^{k-1}
			\frac{\theta_p( q_1q_3^{-\ell})\theta_p( q_2q_3^{-\ell})}{\theta_p( q_3^{-\ell - 1})\theta_p( q_3^{-\ell})}
			\right)
			\biggl\{
			\delta\left(q_3^k\frac{w}{z}\right)f^{\vec{c}}_{r-k,m+k}(q_3^{\frac{r-m}{2}} ; p)\widetilde{T}^{\vec{c}}_{r-k}(q_3^{-k}z ; p)\widetilde{T}^{\vec{c}}_{m+k}(q_3^kw ; p)
			\notag \\ 
			&\hspace{0.3cm}- 
			\delta\left(q_3^{r-m-k}\frac{w}{z}\right)f^{\vec{c}}_{r-k,m+k}(q_3^{\frac{m-r}{2}} ; p)\widetilde{T}^{\vec{c}}_{r-k}(z ; p)\widetilde{T}^{\vec{c}}_{m+k}(w ; p)
			\biggr\}. 
		\end{align}
This relation is called quadratic relations for elliptic corner VOA. 
\label{thm51-1432}
\end{thm}

We prove this theorem by mathematical induction on $r$. 
\subsection{Basis step}
A detailed explanation of the proof of basis step for the quantum corner VOA is provided in appendix A of paper \cite{HMNW}. 
Since the proof of basis step for the elliptic corner VOA is nearly identical to those of quantum corner VOA, here, we merely present a sketch of the proof. For readers interested in comprehensive proof, we recommend consulting appendix A of paper \cite{HMNW}. 

According to \textbf{Lemmas \ref{lemm41-0920}} and \textbf{\ref{lemm46-1317}}, we see that 
\begin{align}
	&f^{\vec{c}}_{1,m}\left(q_3^{\frac{1 - m}{2}}\frac{w}{z} ; p \right)
	\widetilde{T}^{\vec{c}}_{1}(z;p)
	\widetilde{T}^{\vec{c}}_{m}(w;p)
	- 
	f^{\vec{c}}_{m,1}\left(q_3^{\frac{m - 1}{2}}\frac{z}{w} ; p\right)
	\widetilde{T}^{\vec{c}}_{m}(w;p)
	\widetilde{T}^{\vec{c}}_{1}(z;p)
	\notag	\\
	&=
	\underbrace{			\sum_{p_1 \in \bb{Z}^{\geq 0}}					}_{1 \leq p_1 \leq n}
	\underbrace{			\sum_{	m^\prime_1, \dots, m^\prime_n \in \bb{Z}^{\geq 0}		}							}_{
		m^\prime_1 + \cdots + m^\prime_n = m
	}
	\biggl\{
	A(1,c_{p_1} ; p)
	\bigg(	\prod_{k = 1}^{n}A(m^\prime_k,c_k;p)					\bigg)
	\times
	\bigg[
	\widetilde{\textbf{\textit{O}}}^{(m^\prime_1,\dots,m^\prime_n)}_{p_1}(w/z ; p) - \text{\textbf{\textit{O}}}^{\prime \,\, (m^\prime_1,\dots,m^\prime_n)}_{p_1}(z/w ; p)
	\bigg]
	\notag	\\
	&\hspace{4cm}
	\times
	\normord{
		\widetilde{		\Lambda		}^{\vec{c}}_{p_1}(z;p)
		\bigcdot
		\prod_{\ell = 1}^{n}
		\prod_{u_\ell = 1}^{m^\prime_\ell}
		\widetilde{\Lambda}^{\vec{c}}_\ell\big(
		(q_3^{-1})^{\sum_{k = 1}^{\ell - 1}m^\prime_k + u_\ell - 1}w ; p
		\big)
	}
	\biggr\}.
\label{eqn52-}
\end{align}
By using the identity, 
\begin{align}
	&\frac{
		\theta_p(t^{\frac{1}{2}}q^{-1}z)
		\theta_p(t^{-\frac{1}{2}}qz)
	}{
		\theta_p(t^{\frac{1}{2}}z)
		\theta_p(t^{-\frac{1}{2}}z)
	}
	- 
	\frac{
		\theta_p(t^{\frac{1}{2}}q^{-1}z^{-1})
		\theta_p(t^{-\frac{1}{2}}qz^{-1})
	}{
		\theta_p(t^{\frac{1}{2}}z^{-1})
		\theta_p(t^{-\frac{1}{2}}z^{-1})
	}
	= - 
	\frac{
		\theta_p(q)\theta_p(tq^{-1})
	}{
		(p;p)^2_{\infty}\theta_p(t)
	}
	\bigg[
	\delta(t^{\frac{1}{2}}z) - \delta(t^{-\frac{1}{2}}z)
	\bigg], 
\label{iden53}
\end{align} 
one can show that 
\begin{align}
&\widetilde{\textbf{\textit{O}}}^{(m^\prime_1,\dots,m^\prime_n)}_{p_1}(w/z ; p) 
- \text{\textbf{\textit{O}}}^{\prime \,\, (m^\prime_1,\dots,m^\prime_n)}_{p_1}(z/w ; p)
\notag \\
&= 
\begin{cases}
\displaystyle
- 
\frac{
	\theta_p\left(		q_3^{-m^\prime_{p_1}}q_{c_{p_1}}					\right)
	\theta_p\left(		q_{c^\prime_{p_1}}			\right)
}{
	(p;p)^2_{\infty}
	\theta_p\left(		(q_3^{-1})^{m^\prime_{p_1} + 1}		\right)
}
\bigg[
\delta\left(		q_3^{-\sum_{i = 1}^{p_1}m^\prime_i}\frac{w}{z}		\right)
- 
\delta\left(	q_3q_3^{-\sum_{i = 1}^{p_1 - 1}m^\prime_i}\frac{w}{z}				\right)
\bigg]
&\text{ if } c_{p_1} = 1,2 
\\
\displaystyle
- 
\frac{
	\theta_p(q_1)
	\theta_p(q_2)
}{
	(p;p)^2_{\infty}\theta_p(q_3^{-1})
}
\bigg[
\delta\left(		q_3^{-\sum_{i = 1}^{p_1 - 1}m^\prime_i}\frac{w}{z}		\right)
- 
\delta\left(	q_3q_3^{-\sum_{i = 1}^{p_1 - 1}m^\prime_i}\frac{w}{z}				\right)
\bigg]
&\text{ if $c_{p_1} = 3$ and $m^\prime_{p_1} = 0$} 
\\
0
&\text{ if $c_{p_1} = 3$ and $m^\prime_{p_1} = 1$} 
\end{cases}
\end{align}
For the proof of the identity \eqref{iden53}, we encourage the reader to read \cite{Nieri}. From the fact that $A(m^\prime_{p_1},3 ; p) = 0$ for $m^\prime_{p_1} \geq 2$, we can rewrite equation \eqref{eqn52-} as 
\begin{align}
	&f^{\vec{c}}_{1,m}\left(q_3^{\frac{1 - m}{2}}\frac{w}{z} ; p \right)
	\widetilde{T}^{\vec{c}}_{1}(z;p)
	\widetilde{T}^{\vec{c}}_{m}(w;p)
	- 
	f^{\vec{c}}_{m,1}\left(q_3^{\frac{m - 1}{2}}\frac{z}{w} ; p\right)
	\widetilde{T}^{\vec{c}}_{m}(w;p)
	\widetilde{T}^{\vec{c}}_{1}(z;p)
	\notag	\\
	&=
	\underbrace{			\sum_{p_1 \in \bb{Z}^{\geq 0}}					}_{
		\substack{				
			(1) \,\, 1 \leq p_1 \leq n
			\\
			(2) \,\, c_{p_1} \neq 3
		}
	}
	\underbrace{			\sum_{	m^\prime_1, \dots, m^\prime_n \in \bb{Z}^{\geq 0}		}							}_{
		m^\prime_1 + \cdots + m^\prime_n = m
	}
	\biggl\{
	A(1,c_{p_1} ; p)
	\times
	\bigg(	\prod_{k = 1}^{n}A(m^\prime_k,c_k;p)					\bigg)
	\times
	- 
	\frac{
		\theta_p\left(		q_3^{-m^\prime_{p_1}}q_{c_{p_1}}					\right)
		\theta_p\left(		q_{c^\prime_{p_1}}			\right)
	}{
		(p;p)^2_{\infty}
		\theta_p\left(		(q_3^{-1})^{m^\prime_{p_1} + 1}		\right)
	}
	\notag	\\
	&\hspace{0.3cm}
	\times
	\bigg[
	\delta\left(		q_3^{-\sum_{i = 1}^{p_1}m^\prime_i}\frac{w}{z}		\right)
	- 
	\delta\left(	q_3q_3^{-\sum_{i = 1}^{p_1 - 1}m^\prime_i}\frac{w}{z}				\right)
	\bigg]
	\times
	\normord{
		\widetilde{		\Lambda		}^{\vec{c}}_{p_1}(z;p)
		\bigcdot
		\prod_{\ell = 1}^{n}
		\prod_{u_\ell = 1}^{m^\prime_\ell}
		\widetilde{\Lambda}^{\vec{c}}_\ell\big(
		(q_3^{-1})^{\sum_{k = 1}^{\ell - 1}m^\prime_k + u_\ell - 1}w ; p
		\big)
	}
	\biggr\}
	\notag	\\
	&+
	\underbrace{			\sum_{p_1 \in \bb{Z}^{\geq 0}}					}_{
		\substack{				
			(1) \,\, 1 \leq p_1 \leq n
			\\
			(2) \,\, c_{p_1} = 3
		}
	}
	\underbrace{			\sum_{	m^\prime_1, \dots, m^\prime_n \in \bb{Z}^{\geq 0}		}							}_{
		\substack{				
			(1) \,\, m^\prime_1 + \cdots + m^\prime_n = m 
			\\
			(2) \,\, m^\prime_{p_1} = 0
		}
	}	
	\biggl\{
	A(1,c_{p_1} ; p)
	\times
	\bigg(	\prod_{k = 1}^{n}A(m^\prime_k,c_k;p)					\bigg)
	\times
	- 
	\frac{
		\theta_p(q_1)
		\theta_p(q_2)
	}{
		(p;p)^2_{\infty}\theta_p(q_3^{-1})
	}
	\notag	\\
	&\hspace{0.3cm}
	\times
	\bigg[
	\delta\left(		q_3^{-\sum_{i = 1}^{p_1 - 1}m^\prime_i}\frac{w}{z}		\right)
	- 
	\delta\left(	q_3q_3^{-\sum_{i = 1}^{p_1 - 1}m^\prime_i}\frac{w}{z}				\right)
	\bigg]
	\times
	\normord{
		\widetilde{		\Lambda		}^{\vec{c}}_{p_1}(z;p)
		\bigcdot
		\prod_{\ell = 1}^{n}
		\prod_{u_\ell = 1}^{m^\prime_\ell}
		\widetilde{\Lambda}^{\vec{c}}_\ell\big(
		(q_3^{-1})^{\sum_{k = 1}^{\ell - 1}m^\prime_k + u_\ell - 1}w ; p
		\big)
	}
	\biggr\}. 
	\label{217-0957}
\end{align}
The formal delta functions appearing in the above equation are $\delta\left(q_3^{-m}\frac{w}{z}\right), \delta\left(q_3\frac{w}{z}\right)$ and $\delta\left(q_3^{-k}\frac{w}{z}\right) \,\, (1 \leq k \leq m - 1)$. It can be shown that for each $1 \leq k \leq m - 1$, the coefficient in front of $\delta\left(q_3^{-k}\frac{w}{z}\right)$ vanishes \cite{HMNW}. Moreover, 
it is evident from equations \eqref{eqn416-1720-18-4} \eqref{358-1402} that the parts 
$\delta\left(q_3^{-m}\frac{w}{z}\right)$ and $\delta\left(q_3\frac{w}{z}\right)$ are equal to 
\begin{align}
- 
\frac{
	\theta_p(q_1)
	\theta_p(q_2)
}{
	(p;p)^2_{\infty}\theta_p(q_3^{-1})
}
\times
\delta\left(q_3^{-m}\frac{w}{z}\right)
\times
\widetilde{T}^{\vec{c}}_{m+1}(w;p), 
\end{align}
and 
\begin{align}
	\frac{
		\theta_p(q_1)
		\theta_p(q_2)
	}{
		(p;p)^2_{\infty}\theta_p(q_3^{-1})
	}
	\times
	\delta\left(q_3\frac{w}{z}\right)
	\times
	\widetilde{T}^{\vec{c}}_{m+1}(z;p), 
\end{align}
respectively. Therefore, we obtain that 
\begin{align}
	&f^{\vec{c}}_{1,m}\left(q_3^{\frac{1 - m}{2}}\frac{w}{z} ; p\right)
	\widetilde{T}^{\vec{c}}_{1}(z;p)
	\widetilde{T}^{\vec{c}}_{m}(w;p)
	- 
	f^{\vec{c}}_{m,1}\left(q_3^{\frac{m-1}{2}}\frac{z}{w} ; p\right)
	\widetilde{T}^{\vec{c}}_{m}(w;p)
	\widetilde{T}^{\vec{c}}_{1}(z;p)
	\notag	\\
	&=
	\frac{
		\theta_p(q_1)\theta_p(q_2)
	}{
		(p;p)^2_{\infty}\theta_p(q_3^{-1})
	}
	\delta\left(q_3\frac{w}{z}\right)
	\widetilde{T}^{\vec{c}}_{m+1}(z;p)
	- 
	\frac{
		\theta_p(q_1)\theta_p(q_2)
	}{
		(p;p)^2_{\infty}\theta_p(q_3^{-1})
	}
	\delta\left(q_3^{-m}\frac{w}{z}\right)
	\widetilde{	T		}^{\vec{c}}_{m+1}(w ; p).
\label{bassss-1612}
\end{align}
So the basis step of induction is done. 

\subsection{Inductive step}
As an induction hypothesis, we assume that the relation below holds for any $i = 1, \dots, r-1 \,\, (r \geq 2)$ and for $j \geq i$ : 

\begin{align}
		&f^{\vec{c}}_{i,j}\left(
		q^{\frac{i - j}{2}}_{3}\frac{w}{z} ; p
		\right)
		\widetilde{	T	}^{\vec{c}}_i(z ; p)\widetilde{	T		}^{\vec{c}}_j(w ; p) 
		- f^{\vec{c}}_{j,i}\left(
		q^{\frac{j - i}{2}}_{3}\frac{z}{w} ; p
		\right)
		\widetilde{T}^{\vec{c}}_j(w ; p)\widetilde{T}^{\vec{c}}_i(z ; p)
		\notag \\
		&= 
		\frac{
			\theta_p(q_1)\theta_p(q_2)
		}{
			(p;p)^2_{\infty}\theta_p(q_3^{-1})
		}
		\sum_{k = 1}^{i}
		\left(
		\prod_{\ell = 1}^{k-1}
		\frac{\theta_p( q_1q_3^{-\ell})\theta_p( q_2q_3^{-\ell})}{\theta_p( q_3^{-\ell - 1})\theta_p( q_3^{-\ell})}
		\right)
		\biggl\{
		%			\notag \\
		%			&
		\delta\left(q_3^k\frac{w}{z}\right)f^{\vec{c}}_{i-k,j+k}(q_3^{\frac{i-j}{2}} ; p)\widetilde{T}^{\vec{c}}_{i-k}(q_3^{-k}z ; p)\widetilde{T}^{\vec{c}}_{j+k}(q_3^kw ; p)
		\notag \\ 
		&\hspace{0.3cm}- 
		\delta\left(q_3^{i-j-k}\frac{w}{z}\right)f^{\vec{c}}_{i-k,j+k}(q_3^{\frac{j-i}{2}} ; p)\widetilde{T}^{\vec{c}}_{i-k}(z ; p)\widetilde{T}^{\vec{c}}_{j+k}(w ; p)
		\biggr\}. 
\end{align}
Now, let $m \geq r$. From the induction hypothesis, we get that 
\begin{align}
	&f^{\vec{c}}_{r - 1,m}\left(
	q^{\frac{r - 1 - m}{2}}_{3}\frac{w}{z} ; p
	\right)
	\widetilde{	T	}^{\vec{c}}_{r-1}(z ; p)\widetilde{	T		}^{\vec{c}}_m(w ; p) 
	- f^{\vec{c}}_{m,r - 1}\left(
	q^{\frac{m - r + 1}{2}}_{3}\frac{z}{w} ; p
	\right)
	\widetilde{T}^{\vec{c}}_m(w ; p)\widetilde{T}^{\vec{c}}_{r-1}(z ; p)
	\notag \\
	&= 
	\frac{
		\theta_p(q_1)\theta_p(q_2)
	}{
		(p;p)^2_{\infty}\theta_p(q_3^{-1})
	}
	\sum_{k = 1}^{r-1}
	\biggl\{
	\prod_{\ell = 1}^{k-1}
	\frac{\theta_p( q_1q_3^{-\ell})\theta_p( q_2q_3^{-\ell})}{\theta_p( q_3^{-\ell - 1})\theta_p( q_3^{-\ell})}
	\times
	\notag	\\
	&\hspace{0.3cm}\times
	\bigg[
	\delta\left(q_3^k\frac{w}{z}\right)f^{\vec{c}}_{r-k - 1,m+k}(q_3^{\frac{r-m - 1}{2}} ; p)\widetilde{T}^{\vec{c}}_{r-k - 1}(q_3^{-k}z ; p)\widetilde{T}^{\vec{c}}_{m+k}(q_3^kw ; p)
	\notag \\ 
	&\hspace{1.4cm}- 
	\delta\left(q_3^{r-m-k - 1}\frac{w}{z}\right)f^{\vec{c}}_{r-k - 1,m+k}(q_3^{\frac{m-r + 1}{2}} ; p)\widetilde{T}^{\vec{c}}_{r-k - 1}(z ; p)\widetilde{T}^{\vec{c}}_{m+k}(w ; p)
	\bigg]\biggr\}. 
\label{eqn510-1505}
\end{align}
Next, we operate both sides of equation \eqref{eqn510-1505} by 
\begin{align}
\lim_{z^\prime \rightarrow q_3^{-(r-1)}z}
\bigg[
\displaystyle \theta_p\left(	q_3^{-(r-1)}\frac{z}{z^\prime}			\right)f^{\vec{c}}_{1,r-1}\left(	q_3^{\frac{2-r}{2}}		\frac{z}{z^\prime} ; p	\right)f^{\vec{c}}_{1,m}\left(		q_3^{\frac{1 - m}{2}} \frac{w}{z^\prime} ; p			\right)\widetilde{T}^{\vec{c}}_1(z^\prime ; p)
\times 
\cdots
\bigg]. 
\end{align}
We then obtain the following lemma. 

\begin{lem}
\label{lemm52}
\mbox{}
\begin{align}
	&\lim_{z^\prime \rightarrow q_3^{-(r-1)}z}
	\bigg[
	\displaystyle \theta_p\left(	q_3^{-(r-1)}\frac{z}{z^\prime}			\right)f^{\vec{c}}_{1,r-1}\left(	q_3^{\frac{2-r}{2}}		\frac{z}{z^\prime} ; p	\right)f^{\vec{c}}_{1,m}\left(		q_3^{\frac{1 - m}{2}} \frac{w}{z^\prime} ; p			\right)\widetilde{T}^{\vec{c}}_1(z^\prime ; p)
	\times 
	\operatorname{LHS \eqref{eqn510-1505}}
	\bigg]
	\notag	\\
	&=
	-
	\frac{
		\theta_p(q_1)\theta_p(q_2)
	}{
		\theta_p(q_3^{-1})
	}
	f^{\vec{c}}_{r,m}\left(q^{\frac{r  - m}{2}}_{3}\frac{w}{z} ; p\right)
	\widetilde{	T	}^{\vec{c}}_{r}(z ; p)
	\widetilde{	T		}^{\vec{c}}_m(w ; p) 
	+ 
	\frac{
		\theta_p(q_1)\theta_p(q_2)
	}{
		\theta_p(q_3^{-1})
	}
	f^{\vec{c}}_{m,r}\left(q^{\frac{m  - r}{2}}_{3}\frac{z}{w} ; p\right)
	\widetilde{	T		}^{\vec{c}}_m(w ; p)
	\widetilde{	T	}^{\vec{c}}_{r}(z ; p)
	\notag	\\
	&+ 
	\delta\left(q_3^r\frac{w}{z}\right)
	\frac{
		\theta_p(q_1)\theta_p(q_2)
	}{
		(p;p)^2_{\infty}\theta_p(q_3^{-1})
	}
	\frac{
		\theta_p\left(	q_1q_3^{1 - r}		\right)
		\theta_p\left(	q_2q_3^{1 - r}		\right)
	}{
		\theta_p\left(	q_3^{1 - r}		\right)
		\theta_p\left(	q_3^{-r}		\right)
	}
	\frac{
		\theta_p(q_1)\theta_p(q_2)
	}{
		\theta_p(q_3^{-1})
	}
	\times
	\prod_{\ell = 1}^{r - 2}
	\frac{
		\theta_p(q_1q_3^{-\ell})
		\theta_p(q_2q_3^{-\ell})
	}{
		\theta_p(q_3^{-\ell - 1})
		\theta_p(q_3^{-\ell})
	}
	\times 
	\widetilde{	T	}^{\vec{c}}_{m+r}(z ; p)
	\notag	\\
	&- 
	\delta\left(q_3^{r - m - 1} \frac{w}{z}\right)
	\frac{
		\theta_p(q_1)\theta_p(q_2)
	}{
		(p;p)^2_{\infty}\theta_p(q_3^{-1})
	}
	\frac{
		\theta_p(q_1)\theta_p(q_2)
	}{
		\theta_p(q_3^{-1})
	}
	f^{\vec{c}}_{m+1,r-1}\left(q_3^{\frac{r - m}{2}} ; p\right)
	\widetilde{T}^{\vec{c}}_{m+1}(w ; p)
	\widetilde{	T	}^{\vec{c}}_{r-1}(z ; p). 
\end{align}
\end{lem}
\begin{proof}\mbox{}
With the help of \eqref{bassss-1612}, it can be shown that 
\begin{align}
&\theta_p\left(	q_3^{-(r-1)}\frac{z}{z^\prime}			\right)f^{\vec{c}}_{1,r-1}\left(	q_3^{\frac{2-r}{2}}		\frac{z}{z^\prime} ; p	\right)f^{\vec{c}}_{1,m}\left(		q_3^{\frac{1 - m}{2}} \frac{w}{z^\prime} ; p			\right)\widetilde{T}^{\vec{c}}_1(z^\prime ; p)
\times 
\operatorname{LHS \eqref{eqn510-1505}}
\notag \\
&=
\theta_p\left(	q_3^{-(r-1)}\frac{z}{z^\prime}			\right)
f^{\vec{c}}_{1,r-1}\left(	q_3^{\frac{2-r}{2}}		\frac{z}{z^\prime} ; p	\right)
f^{\vec{c}}_{1,m}\left(		q_3^{\frac{1 - m}{2}} \frac{w}{z^\prime} ; p			\right)
f^{\vec{c}}_{r - 1,m}\left(
q^{\frac{r - 1 - m}{2}}_{3}\frac{w}{z} ; p
\right)
\notag \\
&\hspace{0.4cm}\times
\widetilde{T}^{\vec{c}}_1(z^\prime ; p) 
\widetilde{	T	}^{\vec{c}}_{r-1}(z ; p)
\widetilde{	T		}^{\vec{c}}_m(w ; p) 
\notag \\
&-
\theta_p\left(	q_3^{-(r-1)}\frac{z}{z^\prime}			\right)
f^{\vec{c}}_{1,r-1}\left(	q_3^{\frac{2-r}{2}}		\frac{z}{z^\prime} ; p	\right)
f^{\vec{c}}_{m,r-1}\left(		q_3^{\frac{m - r +1}{2}}\frac{z}{w} ; p			\right)
f^{\vec{c}}_{m,1}\left(	q_3^{\frac{m-1}{2}}\frac{z^\prime}{w} ; p				\right)
\notag \\
&\hspace{0.4cm}\times
\widetilde{	T		}^{\vec{c}}_m(w ; p) 
\widetilde{T}^{\vec{c}}_1(z^\prime ; p) 
\widetilde{	T	}^{\vec{c}}_{r-1}(z ; p)
\notag \\
&-
\theta_p\left(	q_3^{-(r-1)}\frac{z}{z^\prime}			\right)
f^{\vec{c}}_{1,r-1}\left(	q_3^{\frac{2-r}{2}}		\frac{z}{z^\prime} ; p	\right)
f^{\vec{c}}_{m,r-1}\left(		q_3^{\frac{m - r +1}{2}}\frac{z}{w} ; p			\right)
\frac{
	\theta_p(q_1)\theta_p(q_2)
}{
	(p;p)^2_{\infty}\theta_p(q_3^{-1})
}
\notag \\
&\hspace{0.4cm}\times
\delta\left(q_3\frac{w}{z^\prime}\right)
\widetilde{T}^{\vec{c}}_{m+1}(z^\prime ; p)
\widetilde{	T	}^{\vec{c}}_{r-1}(z ; p)
\notag \\
&+
\theta_p\left(	q_3^{-(r-1)}\frac{z}{z^\prime}			\right)
f^{\vec{c}}_{1,r-1}\left(	q_3^{\frac{2-r}{2}}		\frac{z}{z^\prime} ; p	\right)
f^{\vec{c}}_{m,r-1}\left(		q_3^{\frac{m - r +1}{2}}\frac{z}{w} ; p			\right)
\frac{
	\theta_p(q_1)\theta_p(q_2)
}{
	(p;p)^2_{\infty}\theta_p(q_3^{-1})
}
\notag \\
&\hspace{0.4cm}\times
\delta\left(q_3^{-m}\frac{w}{z^\prime}\right)
\widetilde{T}^{\vec{c}}_{m+1}(w ; p)
\widetilde{	T	}^{\vec{c}}_{r-1}(z ; p). 
\end{align}
The above expression can be simplified further by noting from \eqref{250-1624-16}  \eqref{251-1624-16} that 
\begin{align}
f^{\vec{c}}_{1,r-1}\left(	q_3^{\frac{2-r}{2}}		\frac{z}{z^\prime} ; p	\right)
f^{\vec{c}}_{m,r-1}\left(		q_3^{\frac{m - r +1}{2}}q_3		\frac{z}{z^\prime} ; p			\right)
= 
\frac{
\theta_p\left(	q_1^{-1}\frac{z}{z^\prime}		\right)
\theta_p\left(	q_2^{-1}\frac{z}{z^\prime}		\right)
}{
\theta_p\left(	\frac{z}{z^\prime}		\right)
\theta_p\left(	q_3\frac{z}{z^\prime}		\right)
}
f^{\vec{c}}_{m+1,r-1}\left(q_3^{\frac{m - r + 2}{2}} \frac{z}{z^\prime} ; p\right), 
\end{align}
and 
\begin{align}
f^{\vec{c}}_{1,r-1}\left(	q_3^{\frac{2-r}{2}}		\frac{z}{z^\prime} ; p	\right)
f^{\vec{c}}_{m,r-1}\left(		q_3^{\frac{m - r +1}{2}}q_3^{-m}		\frac{z}{z^\prime} ; p			\right)
= 
\frac{
\theta_p\left(	q_1^{-1}	q_3^{1-r}\frac{z}{z^\prime}	\right)
\theta_p\left(	q_2^{-1}	q_3^{1-r}\frac{z}{z^\prime}	\right)
}{
\theta_p\left(		q_3^{1-r}\frac{z}{z^\prime}	\right)
\theta_p\left(	q_3	q_3^{1-r}\frac{z}{z^\prime}	\right)
}
f^{\vec{c}}_{m+1,r-1}\left(q_3^{\frac{2 - r - m}{2}} \frac{z}{z^\prime} ; p\right). 
\end{align}
So, we obtain that 
\begin{align}
	&\theta_p\left(	q_3^{-(r-1)}\frac{z}{z^\prime}			\right)
	f^{\vec{c}}_{1,r-1}\left(	q_3^{\frac{2-r}{2}}		\frac{z}{z^\prime} ; p	\right)
	f^{\vec{c}}_{1,m}\left(		q_3^{\frac{1 - m}{2}} \frac{w}{z^\prime} ; p			\right)\widetilde{T}^{\vec{c}}_1(z^\prime ; p)
	\times 
	\operatorname{LHS \eqref{eqn510-1505}}
	\notag \\
	&=
	\theta_p\left(	q_3^{-(r-1)}\frac{z}{z^\prime}			\right)
	f^{\vec{c}}_{1,r-1}\left(	q_3^{\frac{2-r}{2}}		\frac{z}{z^\prime} ; p	\right)
	f^{\vec{c}}_{1,m}\left(		q_3^{\frac{1 - m}{2}} \frac{w}{z^\prime} ; p			\right)
	f^{\vec{c}}_{r - 1,m}\left(
	q^{\frac{r - 1 - m}{2}}_{3}\frac{w}{z} ; p
	\right)
	\notag \\
	&\hspace{0.4cm}\times
	\widetilde{T}^{\vec{c}}_1(z^\prime ; p) 
	\widetilde{	T	}^{\vec{c}}_{r-1}(z ; p)
	\widetilde{	T		}^{\vec{c}}_m(w ; p) 
	\notag \\
	&-
	\theta_p\left(	q_3^{-(r-1)}\frac{z}{z^\prime}			\right)
	f^{\vec{c}}_{1,r-1}\left(	q_3^{\frac{2-r}{2}}		\frac{z}{z^\prime} ; p	\right)
	f^{\vec{c}}_{m,r-1}\left(		q_3^{\frac{m - r +1}{2}}\frac{z}{w} ; p			\right)
	f^{\vec{c}}_{m,1}\left(	q_3^{\frac{m-1}{2}}\frac{z^\prime}{w} ; p				\right)
	\notag \\
	&\hspace{0.4cm}\times
	\widetilde{	T		}^{\vec{c}}_m(w ; p) 
	\widetilde{T}^{\vec{c}}_1(z^\prime ; p) 
	\widetilde{	T	}^{\vec{c}}_{r-1}(z ; p)
	\notag \\
	&-
	\frac{
		\theta_p(q_1)\theta_p(q_2)
	}{
		(p;p)^2_{\infty}\theta_p(q_3^{-1})
	}
	\theta_p\left(	q_3^{-(r-1)}\frac{z}{z^\prime}			\right)
	\frac{
		\theta_p\left(	q_1^{-1}\frac{z}{z^\prime}		\right)
		\theta_p\left(	q_2^{-1}\frac{z}{z^\prime}		\right)
	}{
		\theta_p\left(	\frac{z}{z^\prime}		\right)
		\theta_p\left(	q_3\frac{z}{z^\prime}		\right)
	}
	f^{\vec{c}}_{m+1,r-1}\left(
	q_3^{\frac{m - r + 2}{2}} \frac{z}{z^\prime} ; p
	\right)
	\notag \\
	&\hspace{0.4cm}\times
	\delta\left(q_3\frac{w}{z^\prime}\right)
	\widetilde{T}^{\vec{c}}_{m+1}(z^\prime ; p)
	\widetilde{	T	}^{\vec{c}}_{r-1}(z ; p)
	\notag \\
	&+
	\frac{
		\theta_p\left(	q_1^{-1}	q_3^{1-r}\frac{z}{z^\prime}	\right)
		\theta_p\left(	q_2^{-1}	q_3^{1-r}\frac{z}{z^\prime}	\right)
	}{
		\theta_p\left(	q_3	q_3^{1-r}\frac{z}{z^\prime}	\right)
	}
	f^{\vec{c}}_{m+1,r-1}\left(
	q_3^{\frac{2 - r - m}{2}} \frac{z}{z^\prime} ; p
	\right)
	\frac{
		\theta_p(q_1)\theta_p(q_2)
	}{
		(p;p)^2_{\infty}\theta_p(q_3^{-1})
	}
	\notag \\
	&\hspace{0.4cm}\times
	\delta\left(q_3^{-m}\frac{w}{z^\prime}\right)
	\widetilde{T}^{\vec{c}}_{m+1}(w ; p)
	\widetilde{	T	}^{\vec{c}}_{r-1}(z ; p). 
\label{eqn514-1743}
\end{align}
From fusion relation in \textbf{Theorem \ref{fusion2-1615}}, it is clear that 
\begin{align}
&\lim_{z^\prime \rightarrow q_3^{-(r-1)}z}
\bigg[
\theta_p\left(	q_3^{-(r-1)}\frac{z}{z^\prime}			\right)
f^{\vec{c}}_{1,r-1}\left(	q_3^{\frac{2-r}{2}}		\frac{z}{z^\prime} ; p	\right)
f^{\vec{c}}_{1,m}\left(		q_3^{\frac{1 - m}{2}} \frac{w}{z^\prime} ; p			\right)
f^{\vec{c}}_{r - 1,m}\left(q^{\frac{r - 1 - m}{2}}_{3}\frac{w}{z} ; p\right)
\notag \\
&\hspace{2cm} \times
\widetilde{T}^{\vec{c}}_1(z^\prime ; p) 
\widetilde{	T	}^{\vec{c}}_{r-1}(z ; p)
\widetilde{	T		}^{\vec{c}}_m(w ; p) 
\bigg]
\notag \\
&= 
-
\frac{
	\theta_p(q_1)\theta_p(q_2)
}{
\theta_p(q_3^{-1})
}
f^{\vec{c}}_{r - 1,m}\left(q^{\frac{r - 1 - m}{2}}_{3}\frac{w}{z} ; p\right)
f^{\vec{c}}_{1,m}\left(		q_3^{\frac{1 - m}{2}} q_3^{r - 1}\frac{w}{z} ; p			\right)
\widetilde{	T	}^{\vec{c}}_{r}(z ; p)
\widetilde{	T		}^{\vec{c}}_m(w ; p) 
\notag \\
&= 
-
\frac{
	\theta_p(q_1)\theta_p(q_2)
}{
	\theta_p(q_3^{-1})
}
f^{\vec{c}}_{r,m}\left(q^{\frac{r  - m}{2}}_{3}\frac{w}{z} ; p\right)
\widetilde{	T	}^{\vec{c}}_{r}(z ; p)
\widetilde{	T		}^{\vec{c}}_m(w ; p). 
\end{align}
In the last equality, we used \eqref{248-1624-16}. Similarly, from fusion relations in
\textbf{Theorems \ref{thm45-16-1714}} and \textbf{\ref{fusion2-1615}}, it can be shown that 
\begin{align}
&\lim_{z^\prime \rightarrow q_3^{-(r-1)}z}
\bigg[
\theta_p\left(	q_3^{-(r-1)}\frac{z}{z^\prime}			\right)
f^{\vec{c}}_{1,r-1}\left(	q_3^{\frac{2-r}{2}}		\frac{z}{z^\prime} ; p	\right)
f^{\vec{c}}_{m,r-1}\left(		q_3^{\frac{m - r +1}{2}}\frac{z}{w} ; p			\right)
f^{\vec{c}}_{m,1}\left(	q_3^{\frac{m-1}{2}}\frac{z^\prime}{w} ; p				\right)
\notag \\
&\hspace{2.1cm}
\times
\widetilde{	T		}^{\vec{c}}_m(w ; p) 
\widetilde{T}^{\vec{c}}_1(z^\prime ; p) 
\widetilde{	T	}^{\vec{c}}_{r-1}(z ; p)
\bigg]
\notag \\
&= 
-
\frac{
	\theta_p(q_1)\theta_p(q_2)
}{
	\theta_p(q_3^{-1})
}
f^{\vec{c}}_{m,r}\left(q^{\frac{m  - r}{2}}_{3}\frac{z}{w} ; p\right)
\widetilde{	T		}^{\vec{c}}_m(w ; p)
\widetilde{	T	}^{\vec{c}}_{r}(z ; p), 
\end{align}
and 
\begin{align}
&\lim_{z^\prime \rightarrow q_3^{-(r-1)}z}
\theta_p\left(	q_3^{-(r-1)}\frac{z}{z^\prime}			\right)
\frac{
	\theta_p\left(	q_1^{-1}\frac{z}{z^\prime}		\right)
	\theta_p\left(	q_2^{-1}\frac{z}{z^\prime}		\right)
}{
	\theta_p\left(	\frac{z}{z^\prime}		\right)
	\theta_p\left(	q_3\frac{z}{z^\prime}		\right)
}
f^{\vec{c}}_{m+1,r-1}\left(
q_3^{\frac{m - r + 2}{2}} \frac{z}{z^\prime} ; p
\right)
\delta\left(q_3\frac{w}{z^\prime}\right)
\widetilde{T}^{\vec{c}}_{m+1}(z^\prime ; p)
\widetilde{	T	}^{\vec{c}}_{r-1}(z ; p)
\notag \\
&= 
-
\delta\left(q_3^r\frac{w}{z}\right)
\frac{
	\theta_p\left(	q_1q_3^{1 - r}		\right)
	\theta_p\left(	q_2q_3^{1 - r}		\right)
}{
	\theta_p\left(	q_3^{1 - r}		\right)
	\theta_p\left(	q_3^{-r}		\right)
}
\frac{
	\theta_p(q_1)\theta_p(q_2)
}{
	\theta_p(q_3^{-1})
}
\times
\prod_{\ell = 1}^{r - 2}
\frac{
\theta_p(q_1q_3^{-\ell})
\theta_p(q_2q_3^{-\ell})
}{
\theta_p(q_3^{-\ell - 1})
\theta_p(q_3^{-\ell})
}
\times 
\widetilde{	T	}^{\vec{c}}_{m+r}(z ; p).
\end{align}
Note that we could not apply the fusion relations to the last term in the RHS of \eqref{eqn514-1743}. Nevertheless, the limit 
$\lim_{z^\prime \rightarrow q_3^{-(r-1)}z}$ of this term can be computed directly: 
\begin{align*}
&\lim_{z^\prime \rightarrow q_3^{-(r-1)}z}
\frac{
	\theta_p\left(	q_1^{-1}	q_3^{1-r}\frac{z}{z^\prime}	\right)
	\theta_p\left(	q_2^{-1}	q_3^{1-r}\frac{z}{z^\prime}	\right)
}{
	\theta_p\left(	q_3	q_3^{1-r}\frac{z}{z^\prime}	\right)
}
f^{\vec{c}}_{m+1,r-1}\left(
q_3^{\frac{2 - r - m}{2}} \frac{z}{z^\prime} ; p
\right)
\delta\left(q_3^{-m}\frac{w}{z^\prime}\right)
\widetilde{T}^{\vec{c}}_{m+1}(w ; p)
\widetilde{	T	}^{\vec{c}}_{r-1}(z ; p)
\\
&= 
- 
\delta\left(q_3^{r - m - 1} \frac{w}{z}\right)
\frac{
	\theta_p(q_1)\theta_p(q_2)
}{
	\theta_p(q_3^{-1})
}
f^{\vec{c}}_{m+1,r-1}\left(q_3^{\frac{r - m}{2}} ; p\right)
\widetilde{T}^{\vec{c}}_{m+1}(w ; p)
\widetilde{	T	}^{\vec{c}}_{r-1}(z ; p). 
\end{align*}
We have proved the lemma. 
\end{proof}

\begin{lem}
\label{lemm53}
\mbox{}
\begin{align}
	&\lim_{z^\prime \rightarrow q_3^{-(r-1)}z}
	\bigg[
	\theta_p\left(	q_3^{-(r-1)}\frac{z}{z^\prime}			\right)
	f^{\vec{c}}_{1,r-1}\left(	q_3^{\frac{2-r}{2}}		\frac{z}{z^\prime} ; p	\right)
	f^{\vec{c}}_{1,m}\left(		q_3^{\frac{1 - m}{2}} \frac{w}{z^\prime} ; p			\right)\widetilde{T}^{\vec{c}}_1(z^\prime ; p)
	\times 
	\operatorname{RHS \eqref{eqn510-1505}}
	\bigg]
	\notag	\\
	&=
	\frac{
		\theta_p(q_1)\theta_p(q_2)
	}{
		(p;p)^2_{\infty}\theta_p(q_3^{-1})
	}
	\frac{\theta_p(q_1^{-1})\theta_p(q_2^{-1})}{\theta_p(q_3)}
	\sum_{k = 1}^{r-2}
	\left(
	\prod_{\ell = 1}^{k-1}
	\frac{\theta_p( q_1q_3^{-\ell})\theta_p( q_2q_3^{-\ell})}{\theta_p( q_3^{-\ell - 1})\theta_p( q_3^{-\ell})}
	\right)
	\delta\left(q_3^k\frac{w}{z}\right)
	f^{\vec{c}}_{r-k - 1,m+k}(q_3^{\frac{r-m - 1}{2}} ; p)
	\times
	\notag \\
	& 
	\hspace{0.3cm}\times
	f^{\vec{c}}_{1,m+k}\left(		q_3^{\frac{1 - m - k}{2}} \frac{z}{z^\prime}		; p 		\right)
	\widetilde{T}^{\vec{c}}_{r-k}(q_3^{-k}z ; p)\widetilde{T}^{\vec{c}}_{m+k}(q_3^kw ; p)
	\notag \\
	&- 
	\frac{
		\theta_p(q_1)\theta_p(q_2)
	}{
		(p;p)^2_{\infty}\theta_p(q_3^{-1})
	}
	\frac{\theta_p(q_1^{-1})\theta_p(q_2^{-1})}{\theta_p(q_3)}
	\sum_{k = 1}^{r-2}
	\left(
	\prod_{\ell = 1}^{k-1}
	\frac{\theta_p( q_1q_3^{-\ell})\theta_p( q_2q_3^{-\ell})}{\theta_p( q_3^{-\ell - 1})\theta_p( q_3^{-\ell})}
	\right)
	\delta\left(q_3^{r-m-k - 1}\frac{w}{z}\right)
	f^{\vec{c}}_{m+k, r-k-1}(q_3^{\frac{r-m-1}{2}} ; p)
	\notag \\
	&\hspace{0.3cm}\times
	f^{\vec{c}}_{1,r-1-k}\left(		q_3^{\frac{2 - r + k}{2}} \frac{z}{z^\prime} ; p			\right)
	\widetilde{T}^{\vec{c}}_{m+k+1}(q_3^{m+k+1-r} z ; p)
	\widetilde{T}^{\vec{c}}_{r-k - 1}(z ; p)
	\notag \\
	&+
	\frac{
		\theta_p(q_1)\theta_p(q_2)
	}{
		(p;p)^2_{\infty}\theta_p(q_3^{-1})
	}
	\left(
	\prod_{\ell = 1}^{r-2}
	\frac{\theta_p( q_1q_3^{-\ell})\theta_p( q_2q_3^{-\ell})}{\theta_p( q_3^{-\ell - 1})\theta_p( q_3^{-\ell})}
	\right)
	\frac{\theta_p(q_1^{-1})\theta_p(q_2^{-1})}{\theta_p(q_3)}
	\notag \\
	&\hspace{0.3cm}\times
	\delta\left(q_3^{r-1}\frac{w}{z}\right)
	f^{\vec{c}}_{1,m+r-1}\left(	q_3^{\frac{2 - r - m}{2}}	q_3^{r-1}		; p				\right)
	\widetilde{T}^{\vec{c}}_1(q_3^{-(r-1)}z ; p)
	\widetilde{T}^{\vec{c}}_{m+r - 1}(q_3^{r-1}w ; p)
	\notag \\
	&- 
	\frac{
		\theta_p(q_1)\theta_p(q_2)
	}{
		(p;p)^2_{\infty}\theta_p(q_3^{-1})
	}
	\left(
	\prod_{\ell = 1}^{r-2}
	\frac{\theta_p( q_1q_3^{-\ell})\theta_p( q_2q_3^{-\ell})}{\theta_p( q_3^{-\ell - 1})\theta_p( q_3^{-\ell})}
	\right)
	\frac{
		\theta_p\left(	q_1^{-1}q_3^{r-1}	\right)
		\theta_p\left(	q_2^{-1}q_3^{r-1}	\right)
	}{
		\theta_p\left(	q_3^{r-1}	\right)
		\theta_p\left(q_3^{r}	\right)
	}
	\frac{\theta_p(q_1^{-1})\theta_p(q_2^{-1})}{\theta_p(q_3)}
	\notag \\
	&\hspace{0.3cm}\times
	\delta\left(q_3^{-m}\frac{w}{z}\right)
	\widetilde{T}^{\vec{c}}_{m+r}(q_3^m z ; p).
\end{align}
\end{lem}
\begin{proof}
First note that RHS of \eqref{eqn510-1505} can be written as 
\begin{align}
&\frac{
	\theta_p(q_1)\theta_p(q_2)
}{
	(p;p)^2_{\infty}\theta_p(q_3^{-1})
}
\sum_{k = 1}^{r-2}
\left(
\prod_{\ell = 1}^{k-1}
\frac{\theta_p( q_1q_3^{-\ell})\theta_p( q_2q_3^{-\ell})}{\theta_p( q_3^{-\ell - 1})\theta_p( q_3^{-\ell})}
\right)
\biggl\{
\delta\left(q_3^k\frac{w}{z}\right)
f^{\vec{c}}_{r-k - 1,m+k}(q_3^{\frac{r-m - 1}{2}} ; p)\widetilde{T}^{\vec{c}}_{r-k - 1}(q_3^{-k}z ; p)\widetilde{T}^{\vec{c}}_{m+k}(q_3^kw ; p)
\notag \\ 
&\hspace{0.3cm}- 
\delta\left(q_3^{r-m-k - 1}\frac{w}{z}\right)
f^{\vec{c}}_{r-k - 1,m+k}(q_3^{\frac{m-r + 1}{2}} ; p)\widetilde{T}^{\vec{c}}_{r-k - 1}(z ; p)\widetilde{T}^{\vec{c}}_{m+k}(w ; p)
\biggr\}
\notag \\
&+
\frac{
	\theta_p(q_1)\theta_p(q_2)
}{
	(p;p)^2_{\infty}\theta_p(q_3^{-1})
}
\left(
\prod_{\ell = 1}^{r-2}
\frac{\theta_p( q_1q_3^{-\ell})\theta_p( q_2q_3^{-\ell})}{\theta_p( q_3^{-\ell - 1})\theta_p( q_3^{-\ell})}
\right)
\biggl\{
\delta\left(q_3^{r-1}\frac{w}{z}\right)\widetilde{T}^{\vec{c}}_{m+r - 1}(q_3^{r-1}w ; p)
- 
\delta\left(q_3^{-m}\frac{w}{z}\right)\widetilde{T}^{\vec{c}}_{m+r - 1}(w ; p)
\biggr\}.
\label{eqn517-1043}
\end{align}
We distinguish the term $k = r-1$ from other terms because the term $k = r-1$ only has one current, and there is no structure function appears in this term. So, we should treat it differently from other terms. Multiplying $\theta_p\left(	q_3^{-(r-1)}\frac{z}{z^\prime}			\right)
f^{\vec{c}}_{1,r-1}\left(	q_3^{\frac{2-r}{2}}		\frac{z}{z^\prime} ; p	\right)f^{\vec{c}}_{1,m}\left(		q_3^{\frac{1 - m}{2}} \frac{w}{z^\prime} ; p		 \right)
\widetilde{T}^{\vec{c}}_1(z^\prime ; p)$ through equation \eqref{eqn517-1043}, and use the following relations, which can be deduced from \textbf{Proposition \ref{217}}, 
\begin{align}
f^{\vec{c}}_{1,r-1}\left(	q_3^{\frac{2-r}{2}}		\frac{z}{z^\prime} ; p	\right)
f^{\vec{c}}_{1,m}\left(		q_3^{\frac{1 - m}{2}} q_3^{-k}\frac{z}{z^\prime} ; p		 \right)
&= 
f^{\vec{c}}_{1,r-1-k}\left(	q_3^{\frac{2 - r - k}{2}}\frac{z}{z^\prime}		; p 	\right)
f^{\vec{c}}_{1,m+k}\left(		q_3^{\frac{1 - m - k}{2}} \frac{z}{z^\prime}		; p 		\right), 
\\
f^{\vec{c}}_{1,r-1}\left(	q_3^{\frac{2-r}{2}}		\frac{z}{z^\prime} ; p	\right)
f^{\vec{c}}_{1,m}\left(		q_3^{\frac{1 - m}{2}} q_3^{m+k+1-r}\frac{z}{z^\prime} ; p		 \right)
&=
f^{\vec{c}}_{1,r-1-k}\left(		q_3^{\frac{2 - r + k}{2}} \frac{z}{z^\prime} ; p			\right)
f^{\vec{c}}_{1,m+k}\left(		q_3^{\frac{3 + m + k - 2r}{2}}		\frac{z}{z^\prime} ; p				\right),
\\
f^{\vec{c}}_{1,r-1}\left(	q_3^{\frac{2-r}{2}}		\frac{z}{z^\prime} ; p	\right)
f^{\vec{c}}_{1,m}\left(		q_3^{\frac{1 - m}{2}} q_3^{1 - r}\frac{z}{z^\prime} ; p		 \right)
&= 
\frac{
\theta_p\left(	q_1^{-1}q_3^{1-r}\frac{z}{z^\prime}		\right)
\theta_p\left(	q_2^{-1}q_3^{1-r}\frac{z}{z^\prime}		\right)
}{
\theta_p\left(	q_3^{1-r}\frac{z}{z^\prime}		\right)
\theta_p\left(	q_3 q_3^{1-r}\frac{z}{z^\prime}		\right)
}
f^{\vec{c}}_{1,m+r-1}\left(	q_3^{\frac{2 - r - m}{2}}\frac{z}{z^\prime} ; p				\right),
\\
f^{\vec{c}}_{1,r-1}\left(	q_3^{\frac{2-r}{2}}		\frac{z}{z^\prime} ; p	\right)
f^{\vec{c}}_{1,m}\left(		q_3^{\frac{1 - m}{2}} q_3^{m}\frac{z}{z^\prime} ; p		 \right)
&= 
\frac{
\theta_p\left(	q_1^{-1}\frac{z}{z^\prime}		\right)
\theta_p\left(	q_2^{-1}\frac{z}{z^\prime}		\right)
}{
\theta_p\left(	\frac{z}{z^\prime}		\right)
\theta_p\left(	q_3\frac{z}{z^\prime}		\right)
}
f^{\vec{c}}_{1,m+r-1}\left(q_3^{\frac{2 -r + m}{2}} \frac{z}{z^\prime} ; p\right), 
\end{align}
we then obtain that 
\begin{align}
&\theta_p\left(	q_3^{-(r-1)}\frac{z}{z^\prime}			\right)
f^{\vec{c}}_{1,r-1}\left(	q_3^{\frac{2-r}{2}}		\frac{z}{z^\prime} ; p	\right)
f^{\vec{c}}_{1,m}\left(		q_3^{\frac{1 - m}{2}} \frac{w}{z^\prime} ; p			\right)\widetilde{T}^{\vec{c}}_1(z^\prime ; p)
\times 
\operatorname{RHS \eqref{eqn510-1505}}
\notag \\
&=
\frac{
	\theta_p(q_1)\theta_p(q_2)
}{
	(p;p)^2_{\infty}\theta_p(q_3^{-1})
}
\sum_{k = 1}^{r-2}
\left(
\prod_{\ell = 1}^{k-1}
\frac{\theta_p( q_1q_3^{-\ell})\theta_p( q_2q_3^{-\ell})}{\theta_p( q_3^{-\ell - 1})\theta_p( q_3^{-\ell})}
\right)
\delta\left(q_3^k\frac{w}{z}\right)
f^{\vec{c}}_{r-k - 1,m+k}(q_3^{\frac{r-m - 1}{2}} ; p)
\theta_p\left(	q_3^{-(r-1)}\frac{z}{z^\prime}			\right)
\notag \\
& 
\hspace{0.3cm}\times
f^{\vec{c}}_{1,r-1-k}\left(	q_3^{\frac{2 - r - k}{2}}\frac{z}{z^\prime}		; p 	\right)
f^{\vec{c}}_{1,m+k}\left(		q_3^{\frac{1 - m - k}{2}} \frac{z}{z^\prime}		; p 		\right)
\widetilde{T}^{\vec{c}}_1(z^\prime ; p)
\widetilde{T}^{\vec{c}}_{r-k - 1}(q_3^{-k}z ; p)\widetilde{T}^{\vec{c}}_{m+k}(q_3^kw ; p)
\notag \\
&- 
\frac{
	\theta_p(q_1)\theta_p(q_2)
}{
	(p;p)^2_{\infty}\theta_p(q_3^{-1})
}
\sum_{k = 1}^{r-2}
\left(
\prod_{\ell = 1}^{k-1}
\frac{\theta_p( q_1q_3^{-\ell})\theta_p( q_2q_3^{-\ell})}{\theta_p( q_3^{-\ell - 1})\theta_p( q_3^{-\ell})}
\right)
\delta\left(q_3^{r-m-k - 1}\frac{w}{z}\right)
f^{\vec{c}}_{r-k - 1,m+k}(q_3^{\frac{m-r + 1}{2}} ; p)
\theta_p\left(	q_3^{-(r-1)}\frac{z}{z^\prime}			\right)
\notag \\
&\hspace{0.3cm}\times
f^{\vec{c}}_{1,r-1-k}\left(		q_3^{\frac{2 - r + k}{2}} \frac{z}{z^\prime} ; p			\right)
f^{\vec{c}}_{1,m+k}\left(		q_3^{\frac{3 + m + k - 2r}{2}}		\frac{z}{z^\prime} ; p				\right)
\widetilde{T}^{\vec{c}}_1(z^\prime ; p)
\widetilde{T}^{\vec{c}}_{r-k - 1}(z ; p)\widetilde{T}^{\vec{c}}_{m+k}(w ; p)
\notag \\
&+
\frac{
	\theta_p(q_1)\theta_p(q_2)
}{
	(p;p)^2_{\infty}\theta_p(q_3^{-1})
}
\left(
\prod_{\ell = 1}^{r-2}
\frac{\theta_p( q_1q_3^{-\ell})\theta_p( q_2q_3^{-\ell})}{\theta_p( q_3^{-\ell - 1})\theta_p( q_3^{-\ell})}
\right)
\times
\frac{
	\theta_p\left(	q_1^{-1}q_3^{1-r}\frac{z}{z^\prime}		\right)
	\theta_p\left(	q_2^{-1}q_3^{1-r}\frac{z}{z^\prime}		\right)
}{
	\theta_p\left(	q_3 q_3^{1-r}\frac{z}{z^\prime}		\right)
}
\notag \\
&\hspace{0.3cm}\times
\delta\left(q_3^{r-1}\frac{w}{z}\right)
f^{\vec{c}}_{1,m+r-1}\left(	q_3^{\frac{2 - r - m}{2}}\frac{z}{z^\prime} ; p				\right)
\widetilde{T}^{\vec{c}}_1(z^\prime ; p)
\widetilde{T}^{\vec{c}}_{m+r - 1}(q_3^{r-1}w ; p)
\notag \\
&- 
\frac{
	\theta_p(q_1)\theta_p(q_2)
}{
	(p;p)^2_{\infty}\theta_p(q_3^{-1})
}
\left(
\prod_{\ell = 1}^{r-2}
\frac{\theta_p( q_1q_3^{-\ell})\theta_p( q_2q_3^{-\ell})}{\theta_p( q_3^{-\ell - 1})\theta_p( q_3^{-\ell})}
\right)
\times
\theta_p\left(	q_3^{-(r-1)}\frac{z}{z^\prime}			\right)
\frac{
	\theta_p\left(	q_1^{-1}\frac{z}{z^\prime}		\right)
	\theta_p\left(	q_2^{-1}\frac{z}{z^\prime}		\right)
}{
	\theta_p\left(	\frac{z}{z^\prime}		\right)
	\theta_p\left(	q_3\frac{z}{z^\prime}		\right)
}
\notag \\
&\hspace{0.3cm}\times
\delta\left(q_3^{-m}\frac{w}{z}\right)
f^{\vec{c}}_{1,m+r-1}\left(q_3^{\frac{2 -r + m}{2}} \frac{z}{z^\prime} ; p\right)
\widetilde{T}^{\vec{c}}_1(z^\prime ; p)
\widetilde{T}^{\vec{c}}_{m+r - 1}(w ; p). 
\label{eqn522-1201}
\end{align}
It is important to note that because of the existence of delta function, the factor
\begin{align}
f^{\vec{c}}_{r-k - 1,m+k}(q_3^{\frac{m-r + 1}{2}} ; p)
\widetilde{T}^{\vec{c}}_{r-k - 1}(z ; p)
\widetilde{T}^{\vec{c}}_{m+k}(q_3^{m+k+1-r} z ; p)
\end{align} 
appearing in the second term of RHS of \eqref{eqn522-1201} can be written as 
\begin{align}
&f^{\vec{c}}_{r-k - 1,m+k}(q_3^{\frac{m-r + 1}{2}} ; p)
\widetilde{T}^{\vec{c}}_{r-k - 1}(z ; p)
\widetilde{T}^{\vec{c}}_{m+k}(q_3^{m+k+1-r} z ; p)
\notag \\
&\hspace{0.3cm}
= 
f^{\vec{c}}_{m+k, r-k-1}(q_3^{\frac{r-m-1}{2}} ; p)
\widetilde{T}^{\vec{c}}_{m+k}(q_3^{m+k+1-r} z ; p)
\widetilde{T}^{\vec{c}}_{r-k - 1}(z ; p). 
\end{align}
Taking the limit $\lim_{z^\prime \rightarrow q_3^{-(r-1)}z}$ of \eqref{eqn522-1201} and apply fusion relations, we then obtain that 
\begin{align}
	&\lim_{z^\prime \rightarrow q_3^{-(r-1)}z}
	\bigg[
	\theta_p\left(	q_3^{-(r-1)}\frac{z}{z^\prime}			\right)
	f^{\vec{c}}_{1,r-1}\left(	q_3^{\frac{2-r}{2}}		\frac{z}{z^\prime} ; p	\right)
	f^{\vec{c}}_{1,m}\left(		q_3^{\frac{1 - m}{2}} \frac{w}{z^\prime} ; p			\right)\widetilde{T}^{\vec{c}}_1(z^\prime ; p)
	\times 
	\operatorname{RHS \eqref{eqn510-1505}}
	\bigg]
\notag	\\
&=
\frac{
	\theta_p(q_1)\theta_p(q_2)
}{
	(p;p)^2_{\infty}\theta_p(q_3^{-1})
}
\frac{\theta_p(q_1^{-1})\theta_p(q_2^{-1})}{\theta_p(q_3)}
\sum_{k = 1}^{r-2}
\left(
\prod_{\ell = 1}^{k-1}
\frac{\theta_p( q_1q_3^{-\ell})\theta_p( q_2q_3^{-\ell})}{\theta_p( q_3^{-\ell - 1})\theta_p( q_3^{-\ell})}
\right)
\delta\left(q_3^k\frac{w}{z}\right)
f^{\vec{c}}_{r-k - 1,m+k}(q_3^{\frac{r-m - 1}{2}} ; p)
\notag \\
& 
\hspace{0.3cm}\times
f^{\vec{c}}_{1,m+k}\left(		q_3^{\frac{1 - m - k}{2}} \frac{z}{z^\prime}		; p 		\right)
\widetilde{T}^{\vec{c}}_{r-k}(q_3^{-k}z ; p)\widetilde{T}^{\vec{c}}_{m+k}(q_3^kw ; p)
\notag \\
&- 
\frac{
	\theta_p(q_1)\theta_p(q_2)
}{
	(p;p)^2_{\infty}\theta_p(q_3^{-1})
}
\frac{\theta_p(q_1^{-1})\theta_p(q_2^{-1})}{\theta_p(q_3)}
\sum_{k = 1}^{r-2}
\left(
\prod_{\ell = 1}^{k-1}
\frac{\theta_p( q_1q_3^{-\ell})\theta_p( q_2q_3^{-\ell})}{\theta_p( q_3^{-\ell - 1})\theta_p( q_3^{-\ell})}
\right)
\delta\left(q_3^{r-m-k - 1}\frac{w}{z}\right)
f^{\vec{c}}_{m+k, r-k-1}(q_3^{\frac{r-m-1}{2}} ; p)
\notag \\
&\hspace{0.3cm}\times
f^{\vec{c}}_{1,r-1-k}\left(		q_3^{\frac{2 - r + k}{2}} \frac{z}{z^\prime} ; p			\right)
\widetilde{T}^{\vec{c}}_{m+k+1}(q_3^{m+k+1-r} z ; p)
\widetilde{T}^{\vec{c}}_{r-k - 1}(z ; p)
\notag \\
&+
\frac{
	\theta_p(q_1)\theta_p(q_2)
}{
	(p;p)^2_{\infty}\theta_p(q_3^{-1})
}
\left(
\prod_{\ell = 1}^{r-2}
\frac{\theta_p( q_1q_3^{-\ell})\theta_p( q_2q_3^{-\ell})}{\theta_p( q_3^{-\ell - 1})\theta_p( q_3^{-\ell})}
\right)
\frac{\theta_p(q_1^{-1})\theta_p(q_2^{-1})}{\theta_p(q_3)}
\notag \\
&\hspace{0.3cm}\times
\delta\left(q_3^{r-1}\frac{w}{z}\right)
f^{\vec{c}}_{1,m+r-1}\left(	q_3^{\frac{2 - r - m}{2}}	q_3^{r-1}		; p				\right)
\widetilde{T}^{\vec{c}}_1(q_3^{-(r-1)}z ; p)
\widetilde{T}^{\vec{c}}_{m+r - 1}(q_3^{r-1}w ; p)
\notag \\
&- 
\frac{
	\theta_p(q_1)\theta_p(q_2)
}{
	(p;p)^2_{\infty}\theta_p(q_3^{-1})
}
\left(
\prod_{\ell = 1}^{r-2}
\frac{\theta_p( q_1q_3^{-\ell})\theta_p( q_2q_3^{-\ell})}{\theta_p( q_3^{-\ell - 1})\theta_p( q_3^{-\ell})}
\right)
\frac{
	\theta_p\left(	q_1^{-1}q_3^{r-1}	\right)
	\theta_p\left(	q_2^{-1}q_3^{r-1}	\right)
}{
	\theta_p\left(	q_3^{r-1}	\right)
	\theta_p\left(q_3^{r}	\right)
}
\frac{\theta_p(q_1^{-1})\theta_p(q_2^{-1})}{\theta_p(q_3)}
\notag \\
&\hspace{0.3cm}\times
\delta\left(q_3^{-m}\frac{w}{z}\right)
\widetilde{T}^{\vec{c}}_{m+r}(q_3^m z ; p).
\end{align}
\end{proof}

According to \textbf{Lemmas \ref{lemm52}} and \textbf{\ref{lemm53}}, we can deduce that 
\begin{align}
	&f^{\vec{c}}_{r,m}\left(
	q^{\frac{r - m}{2}}_{3}\frac{w}{z} ; p
	\right)
	\widetilde{	T	}^{\vec{c}}_r(z ; p)\widetilde{	T		}^{\vec{c}}_m(w ; p) 
	- f^{\vec{c}}_{m,r}\left(
	q^{\frac{m - r}{2}}_{3}\frac{z}{w} ; p
	\right)
	\widetilde{T}^{\vec{c}}_m(w ; p)\widetilde{T}^{\vec{c}}_r(z ; p)
	\notag \\
	&= 
	\frac{
		\theta_p(q_1)\theta_p(q_2)
	}{
		(p;p)^2_{\infty}\theta_p(q_3^{-1})
	}
	\sum_{k = 1}^{r}
	\left(
	\prod_{\ell = 1}^{k-1}
	\frac{\theta_p( q_1q_3^{-\ell})\theta_p( q_2q_3^{-\ell})}{\theta_p( q_3^{-\ell - 1})\theta_p( q_3^{-\ell})}
	\right)
	\biggl\{
	\delta\left(q_3^k\frac{w}{z}\right)
	f^{\vec{c}}_{r-k,m+k}(q_3^{\frac{r-m}{2}} ; p)\widetilde{T}^{\vec{c}}_{r-k}(q_3^{-k}z ; p)\widetilde{T}^{\vec{c}}_{m+k}(q_3^kw ; p)
	\notag \\ 
	&\hspace{0.3cm}- 
	\delta\left(q_3^{r-m-k}\frac{w}{z}\right)
	f^{\vec{c}}_{r-k,m+k}(q_3^{\frac{m-r}{2}} ; p)\widetilde{T}^{\vec{c}}_{r-k}(z ; p)\widetilde{T}^{\vec{c}}_{m+k}(w ; p)
	\biggr\}.
\end{align}
This closes the induction, and hence prove the \textbf{Theorem \ref{thm51-1228}}. 

\section{Partially symmetric polynomials associated to the elliptic corner VOA}
\label{sec6-poly}

In this section, we relate the correlation functions of the currents of elliptic corner VOA to a family of partially symmetric polynomials. We begin with a brief review on the star product. 

\begin{dfn}[\cite{Feigin97}]
For each $j \in \bb{Z}^{\geq 0}$, let $S^j$ be the vector space of symmetric functions in $j$ variables. The star product $\star : S^i \tens S^j \rightarrow S^{i+j}$ is defined as follows : for any $f \in S^i, g \in S^j$
\begin{align}
(f \star g)(x_1,\dots,x_{i+j}) := 
\operatorname{Sym}
\bigg[
f(x_1,\dots,x_i)
g(x_{i+1},\dots,x_{i+j})
\prod_{
\substack{
1 \leq \alpha \leq i 
\\
i + 1 \leq \beta \leq i+j
}
}
\frac{
	\theta_p\left(q^{-1}\frac{z_\beta}{z_\alpha}\right)
	\theta_p\left(t\frac{z_\beta}{z_\alpha}\right)
	\theta_p\left(qt^{-1}\frac{z_\beta}{z_\alpha}\right)
}{
	\theta_p\left(\frac{z_\beta}{z_\alpha}\right)^3
}
\bigg].
\end{align}
\end{dfn}

\begin{prop}[\cite{Feigin97}]
Define $S := \ddsum_{k = 0}^{\infty}S^k$. Then, $S$ is a unital associative algebra with respect to the star product $\star$. 
\end{prop}

Next, for each $n \in \bb{Z}^{\geq 1}$, we define 
\begin{align}
\epsilon_n(z_1,\dots,z_n ; t ; p) := 
\prod_{
1 \leq i < j \leq n 
}
\frac{
\theta_p\left(t\frac{z_j}{z_i}\right)
\theta_p\left(t^{-1}\frac{z_j}{z_i}\right)
}{
\theta_p\left(\frac{z_j}{z_i}\right)^2
}
\in 
S^n. 
\label{eqn61-1810-16t}
\end{align}
Note that when $n = 1$, the product becomes an empty product, so $\epsilon_1(z) = 1$. The next proposition tells us that the star product among $\{\epsilon_n ~|~ n \in \bb{Z}^{\geq 1}\}$ is commutative. 

\begin{prop}[\cite{Feigin97}]
$(\epsilon_m \star \epsilon_n)(z_1,\dots,z_{m+n} ; t ; p) = (\epsilon_n \star \epsilon_m)(z_1,\dots,z_{m+n} ; t ; p) $
\label{prop63-1006}
\end{prop}

In the next theorem, we relate the correlation function of the currents of $q\widetilde{Y}^{\text{ell}}_{0,0,N}[\Psi;p]$ to symmetric functions. 

\begin{thm}[]\mbox{}
Let $\Lambda_{\bb{Q}(q,t)}(u_1,\dots,u_n)$ be the ring of symmetric polynomials in variables $u_1,\dots,u_n$ with coefficients in  $\bb{Q}(q,t)$. 
Then, 
\begin{align}
\prod_{1 \leq i < j \leq m}f^{(3^n)}_{11}\left(\frac{z_j}{z_i} ; p\right)
\times
\langle 0 |\widetilde{T}^{(3^n)}_{1}(z_1 ; p)\cdots \widetilde{T}^{(3^n)}_{1}(z_m ; p)|0\rangle
\in 
\Lambda_{\bb{Q}(q,t)}(u_1,\dots,u_n) \tens S, 
\end{align}
is symmetric in $u_1,\dots,u_n$ and $z_1,\dots,z_m$. 
\label{thm61-1624}
\end{thm}

To prove this theorem, the following lemma is helpful. 

\begin{lem}
\label{lem63-1711}
Define\footnote{
To avoid cumbersome notation, we write $\gamma^{(i,j)}_{k,\ell}(q,t ; p)$ rather than 
$\gamma^{(i,j)}(z_k, z_\ell ;q,t ; p)$. 
}
\begin{align}
	\gamma^{(i,j)}_{k,\ell}(q,t ; p)
	= 
	\begin{cases}
		\displaystyle
		\frac{
			\theta_p\left(q^{-1}\frac{z_\ell}{z_k}\right)
			\theta_p\left(t\frac{z_\ell}{z_k}\right)
			\theta_p\left(qt^{-1}\frac{z_\ell}{z_k}\right)
		}{
			\theta_p\left(\frac{z_\ell}{z_k}\right)^3
		}
		\hspace{0.2cm}
		\text{ if } i < j 
		\\
		\displaystyle
		\frac{
			\theta_p\left(t\frac{z_\ell}{z_k}\right)
			\theta_p\left(t^{-1}\frac{z_\ell}{z_k}\right)
		}{
			\theta_p\left(\frac{z_\ell}{z_k}\right)^2
		}
		\hspace{2.2cm}
		\text{ if } i = j
		\\
		\displaystyle
		\frac{
			\theta_p\left(q\frac{z_\ell}{z_k}\right)
			\theta_p\left(t^{-1}\frac{z_\ell}{z_k}\right)
			\theta_p\left(q^{-1}t\frac{z_\ell}{z_k}\right)
		}{
			\theta_p\left(\frac{z_\ell}{z_k}\right)^3
		}
		\hspace{0.2cm}
		\text{ if } i > j 
	\end{cases}
\label{6.2eqn-1042}
\end{align}
Then, 
\begin{align}
\sum_{i_1 = 1}^{n}
\cdots
\sum_{i_m = 1}^{n}
\left(
\prod_{k = 1}^{m}u_{i_k}
\tens
\prod_{1 \leq \alpha < \beta \leq m}
\gamma^{(i_\alpha,i_\beta)}_{\alpha,\beta}(q,t ; p)
\right)
\end{align}
is symmetric in $u_1,\dots,u_n$ and $z_1,\dots,z_m$. 
\end{lem}

\begin{proof}[Proof of lemma \ref{lem63-1711}]\mbox{}
\begin{align}
	&\sum_{i_1 = 1}^{n}
	\cdots
	\sum_{i_m = 1}^{n}
	\left(
	\prod_{k = 1}^{m}u_{i_k}
	\tens
	\prod_{1 \leq \alpha < \beta \leq m}
	\gamma^{(i_\alpha,i_\beta)}_{\alpha,\beta}(q,t ; p)
	\right)
	\notag	\\
	&
	=
	\underbrace{				
		\sum_{
			(a_1,\dots,a_n) \in 
			\left(
			\bb{Z}^{\geq 0}
			\right)^n
		}
	}_{
		a_1 + \dots + a_n = m
	}
	\left[
	\prod_{k = 1}^{n}u_k^{a_k}
	\tens
	\sum_{
		(I_1,\dots,I_n) \in \cals{J}(a_1,\dots,a_n)
	}
	\left(
	\prod_{k = 1}^{n}\epsilon_{a_k}(z_{I_k} ; t ; p)
	\times
	\prod_{1 \leq i < j \leq n}\prod_{\alpha \in I_i, \beta \in I_j}
	\gamma^{(i,j)}_{\alpha,\beta}(q,t ; p)
	\right)
	\right],
\end{align}
where 
\begin{align*}
\cals{J}(a_1,\dots,a_n) = 
\left\{
(I_1,\dots,I_n)
\;\middle\vert\;
\begin{array}{@{}l@{}}
	(1) \,\, I_1 \sqcup \cdots \sqcup I_n = 
	\{1,\dots,m\}
	\\
	(2) \,\, |I_k| = a_k \,\, ^\forall k = 1,\dots,n
\end{array}
\right\}. 
\end{align*}
On the other side, one can easily show by induction that 
\begin{align*}
&(\epsilon_{a_1} \star \cdots \star \epsilon_{a_n})
(z_1,\dots,z_{a_1 + \cdots + a_n} ; t ; p) 
\\
&= 
\frac{
a_1 ! \cdots a_n!
}{
(a_1 + \cdots + a_n)!
}
\sum_{
	(I_1,\dots,I_n) \in \cals{J}(a_1,\dots,a_n)
}
\left(
\prod_{k = 1}^{n}\epsilon_{a_k}(z_{I_k} ; t ; p)
\times
\prod_{1 \leq i < j \leq n}\prod_{\alpha \in I_i, \beta \in I_j}
\gamma^{(i,j)}_{\alpha,\beta}(q,t ; p)
\right). 
\end{align*}
Therefore, 
\begin{align}
	&\sum_{i_1 = 1}^{n}
	\cdots
	\sum_{i_m = 1}^{n}
	\left(
	\prod_{k = 1}^{m}u_{i_k}
	\tens
	\prod_{1 \leq \alpha < \beta \leq m}
	\gamma^{(i_\alpha,i_\beta)}_{\alpha,\beta}(q,t ; p)
	\right)
	\notag	\\
	&
	=
	\underbrace{				
		\sum_{
			(a_1,\dots,a_n) \in 
			\left(
			\bb{Z}^{\geq 0}
			\right)^n
		}
	}_{
		a_1 + \dots + a_n = m
	}
	\left[
	\prod_{k = 1}^{n}u_k^{a_k}
	\tens
	\frac{m!}{
	a_1 ! \cdots a_n!
	}
	(\epsilon_{a_1} \star \cdots \star \epsilon_{a_n})
	(z_1,\dots,z_{m} ; t ; p) 
	\right]. 
	\label{eqn61-1005}
\end{align}
Since the part of $z_1,\dots,z_m$ can be written in terms of star product, we can immediately conclude that \eqref{eqn61-1005} is symmetric in $z_1,\dots,z_m$. Moreover, by using \textbf{Proposition \ref{prop63-1006}}, we can also conclude that \eqref{eqn61-1005} is symmetric in $u_1, \dots, u_n$. 
\end{proof}

\begin{proof}[Proof of theorem \ref{thm61-1624}]\mbox{}
\begin{align}
	&\prod_{1 \leq i < j \leq m}f^{(3^n)}_{11}\left(\frac{z_j}{z_i} ; p\right)
	\times
	\langle 0 |\widetilde{T}^{(3^n)}_{1}(z_1 ; p)\cdots \widetilde{T}^{(3^n)}_{1}(z_m ; p)|0\rangle
	\notag	\\
	&= 
	\prod_{1 \leq i < j \leq m}f^{(3^n)}_{11}\left(\frac{z_j}{z_i} ; p\right)
	\times
	\sum_{i_1 = 1}^{n}\cdots \sum_{i_m = 1}^{n}
	\langle 0|
	\Lambda^{(3^n)}_{i_1}(z_1 ; p)
	\cdots \Lambda^{(3^n)}_{i_m}(z_m ; p)
	|0\rangle 
	\notag	\\
	&= 
	\sum_{i_1 = 1}^{n}\cdots \sum_{i_m = 1}^{n}
	\bigg[
	u_{i_1}\cdots u_{i_m}
	\tens
	\prod_{
		1 \leq k < \ell \leq m
	}
	C^{(i_k,i_\ell)}_{k,\ell}(q,t ; p)
	\bigg], 
\end{align}
where 
\begin{align}
	C^{(i,j)}_{k,\ell}(q,t ; p)
	:= 
	\begin{cases}
		\displaystyle
		\frac{
		\theta_p\left(q^{-1}\frac{z_\ell}{z_k}\right)
		\theta_p\left(qt^{-1}\frac{z_\ell}{z_k}\right)
		}{
		\theta_p\left(t^{-1}\frac{z_\ell}{z_k}\right)
		\theta_p\left(\frac{z_\ell}{z_k}\right)
		}
		\hspace{0.2cm}
		\text{ if } i < j 
		\\
		\displaystyle
		1
		\hspace{3.9cm}
		\text{ if } i = j
		\\
		\displaystyle
		\frac{
		\theta_p\left( q\frac{z_\ell}{z_k}\right)
		\theta_p\left( q^{-1}t\frac{z_\ell}{z_k}\right)
		}{
		\theta_p\left( t\frac{z_\ell}{z_k}\right)
		\theta_p\left( \frac{z_\ell}{z_k}\right)
		}
		\hspace{0.6cm}
		\text{ if } i > j 
	\end{cases}
\label{eqn66-1042}
\end{align}
Comparing \eqref{6.2eqn-1042} and \eqref{eqn66-1042}, we obtain that 
\begin{align}
C^{(i,j)}_{k,\ell}(q,t ; p)
= 
\frac{
\theta_p\left(\frac{z_\ell}{z_k}\right)^2
}{
\theta_p\left(t\frac{z_\ell}{z_k}\right)
\theta_p\left(t^{-1}\frac{z_\ell}{z_k}\right)
}
\gamma^{(i,j)}_{k,\ell}(q,t ; p).
\end{align}
Therefore, 
\begin{align}
&\prod_{1 \leq i < j \leq m}f^{(3^n)}_{11}\left(\frac{z_j}{z_i} ; p\right)
\times
\langle 0 |\widetilde{T}^{(3^n)}_{1}(z_1 ; p)\cdots \widetilde{T}^{(3^n)}_{1}(z_m ; p)|0\rangle
\notag \\
&= 
\left(
1 \tens 
\prod_{1 \leq k < \ell \leq m}
\frac{
	\theta_p\left(\frac{z_\ell}{z_k}\right)^2
}{
	\theta_p\left(t\frac{z_\ell}{z_k}\right)
	\theta_p\left(t^{-1}\frac{z_\ell}{z_k}\right)
}
\right)
\sum_{i_1 = 1}^{n}\cdots \sum_{i_m = 1}^{n}
\bigg[
u_{i_1}\cdots u_{i_m}
\tens
\prod_{
	1 \leq k < \ell \leq m
}
\gamma^{(i_k,i_\ell)}_{k,\ell}(q,t ; p)
\bigg]. 
\label{eqn64-1727}
\end{align}
From \textbf{Lemma \ref{lem63-1711}}, we conclude that the quantity in \eqref{eqn64-1727} is symmetric in $u_1,\dots,u_n$ and $z_1,\dots,z_m$. 
\end{proof}

\subsection{Conjecture on elliptic Macdonald polynomials}
In this subsection, we state a conjecture which relates the symmetric function in \textbf{Theorem \ref{thm61-1624}} to the elliptic Macdonald polynomials constructed recently in \cite{awata2020-2} \cite{fu20} \cite{Zen21} as a particular case of generalized Noumi-Shiraishi (GNS) polynomials \cite{awata2020}. In short, the GNS polynomials is defined for any function $\xi(z)$ that satisfies certain symmetry conditions. It turns out that these symmetry conditions are also satisfied by theta function $\theta_p(z)$. By substituting $\xi(z) = \theta_p(z)$ into the GNS polynomials, we then obtain the elliptic Macdonald polynomials. 
The precise definition of the elliptic Macdonald polynomials is given in the \textbf{Definition \ref{df66-1332}} below. 

\begin{dfn}[\cite{Zen21}]
\label{df66-1332}
For each partition $\lambda$ whose length is less than or equal to $n \in \bb{Z}^{\geq 1}$, the elliptic Macdonald polynomials $P_{\lambda}(y_1,\dots,y_n;q,t ; p)$ corresponding to the partition $\lambda$ is defined by 
\begin{align}
P_{\lambda}(y_1,\dots,y_n;q,t ; p)
:= 
\prod_{i = 1}^{n}y_i^{\lambda_i} \times
\sum_{
M \in \cals{M}_n
}
c^{\lambda}_{n}(M|q,t,p)
\prod_{
1 \leq i < j \leq n
}
\left(
\frac{y_j}{y_i}
\right)^{M_{ij}}, 
\end{align}
where
\begin{gather}
\cals{M}_n := 
\left\{
(M_{ij})_{1 \leq i , j \leq n} 
\;\middle\vert\;
\begin{array}{@{}l@{}}
(1) \,\, M_{ij} \in \bb{Z}^{\geq 0}
\\
(2) \,\, M_{k\ell} = 0 \text{ if } k \geq \ell
\end{array}
\right\}, 
\end{gather}
and 
\begin{align}
&c^{\lambda}_{n}(M|q,t,p)
\notag \\
&:= 
\prod_{k = 2}^{n}
\prod_{1 \leq i < j \leq k}
\frac{
	\Theta\left( q^{\lambda_j - \lambda_i + \sum_{a > k}(M_{ia} - M_{ja})		} t^{i - j +1}; q, p\right)_{M_{ik}}
	\Theta\left( q^{\lambda_j - \lambda_i - M_{jk} + \sum_{a > k}(M_{ia} - M_{ja})		} qt^{i - j -1}; q, p\right)_{M_{ik}}
}{
	\Theta\left( q^{\lambda_j - \lambda_i + \sum_{a > k}(M_{ia} - M_{ja})		} qt^{i - j}; q, p\right)_{M_{ik}}
	\Theta\left( q^{\lambda_j - \lambda_i - M_{jk} + \sum_{a > k}(M_{ia} - M_{ja})		} t^{i - j}; q, p\right)_{M_{ik}}
}.
\end{align}
\end{dfn}

From the definition, it is not obvious whether the elliptic Macdonald polynomials are symmetric polynomials. By using the nontrivial identity of theta function, the authors of \cite{awata2020-2} proved the next proposition. 

\begin{prop}[\cite{awata2020-2}]
The elliptic Macdonald polynomials $P_{\lambda}(y_1,\dots,y_n;q,t ; p)$ is a symmetric polynomials in $y_1,\dots,y_n$. 
\end{prop}

To state our conjecture, the map $\widetilde{\psi}^{(q,\xi)}_{\lambda}$ defined below is necesaary: for any partition $\lambda = (\lambda_1,\dots,\lambda_{\ell(\lambda)}) \in \operatorname{Par}(m)$ and for any $\xi \in \bb{Q}(q,t)$, define $\widetilde{\psi}^{(q,\xi)}_{\lambda}$ to be the map which sends a symmetric function $f(z_1,\dots,z_m)$ to
\begin{align}
	&f(y,q^{-1}y,\dots,(q^{-1})^{\lambda_1 - 1}y
	\notag	\\
	&\hspace{0.4cm} \xi y,q^{-1}\xi y,\dots,(q^{-1})^{\lambda_2 - 1}\xi y
	\notag \\
	&\hspace{2.8cm}\vdots
	\notag \\
	&\hspace{0.4cm} \xi^{\ell(\lambda) - 1} y,q^{-1}\xi^{\ell(\lambda) - 1} y,\dots,(q^{-1})^{\lambda_{\ell(\lambda)} - 1}\xi^{\ell(\lambda) - 1} y). 
\end{align}
It is clear that not all of the symmetric functions have a well-defined image under the map $\widetilde{\psi}^{(q,\xi)}_{\lambda}$. For example, the image of symmetric function 
\begin{align}
\frac{
\theta_p\left(t\frac{z_2}{z_1}\right)
\theta_p\left(t^{-1}\frac{z_2}{z_1}\right)
}{
\theta_p\left(q\frac{z_2}{z_1}\right)
\theta_p\left(q^{-1}\frac{z_2}{z_1}\right)
},
\end{align}
under $\widetilde{\psi}^{(q,\xi)}_{(2)}$ is not well-defined. However, we have the following conjecture.

\begin{conj}
\label{conj68-1258}
For each $m, n \in \bb{Z}^{\geq 1} \,\, (m \leq n)$, and  $\lambda \in \operatorname{Par}(m)$, 
there exist $\cals{N}_{\lambda}(z_1,\dots,z_m) \in \bb{Q}(q,t)(z_1,\dots,z_m)$ which make the limit\footnote{
Here we write $\widetilde{\psi}^{(q,\xi)}_{\lambda}$ rather than $1 \tens \widetilde{\psi}^{(q,\xi)}_{\lambda}$ to avoid the cumbersome notation. 
}
\begin{align}
	\lim_{\xi \rightarrow t}\,\,
	(
	\widetilde{\psi}^{(q,\xi)}_{\lambda}
	\comp 
	\bigg|_{
		\substack{
			q_1 = q, \\
			q_2 = q^{-1}t,\\
			q_3 = t^{-1} \\
		}
	}
	)
	\left(
	\cals{N}_{\lambda}(z_1,\dots,z_m ; p)
	\times
	\prod_{1 \leq i < j \leq m}f^{(3^n)}_{11}\left(\frac{z_j}{z_i} ; p\right)
	\times
	\langle 0 |\widetilde{T}^{(3^n)}_{1}(z_1 ; p)\cdots \widetilde{T}^{(3^n)}_{1}(z_m ; p)|0\rangle
	\right)
	\label{eqn61-1626}
\end{align}
exists and equal to the Macdonald polynomials $P_{\lambda}(u_1,\dots,u_n;q,t;p)$, i.e. 
\begin{align}
	&\lim_{\xi \rightarrow t}\,\,
	(\widetilde{\psi}^{(q,\xi)}_{\lambda}
	\comp 
	\bigg|_{
		\substack{
			q_1 = q, \\
			q_2 = q^{-1}t,\\
			q_3 = t^{-1} \\
		}
	}
	)
	\left(
	\cals{N}_{\lambda}(z_1,\dots,z_m ; p)
	\times
	\prod_{1 \leq i < j \leq m}f^{(3^n)}_{11}\left(\frac{z_j}{z_i} ; p\right)
	\times
	\langle 0 |\widetilde{T}^{(3^n)}_{1}(z_1 ; p)\cdots \widetilde{T}^{(3^n)}_{1}(z_m ; p)|0\rangle
	\right)
	\notag \\
	&= 
	P_{\lambda}(u_1,\dots,u_n;q,t ; p). 
	\label{eqn01-2042}
\end{align}
\end{conj}

This conjecture will be explored in our next paper. However, in this paper, we give an example as a supporting evidence of the conjecture. 
\begin{exa}
	Consider $m = 3$ and partition $\lambda = (2,1)$. We can show that 
	\begin{align*}
		&\prod_{1 \leq i < j \leq 3}f^{(3^n)}_{11}\left(\frac{z_j}{z_i} ; p\right)
		\times
		\langle 0 |\widetilde{T}^{(3^n)}_{1}(z_1 ; p)\widetilde{T}^{(3^n)}_{1}(z_2 ; p) \widetilde{T}^{(3^n)}_{1}(z_3 ; p)|0\rangle
		\\
		&= 
		\underbrace{
			\sum_{i_1 = 1}^{n}
			\sum_{i_2 = 1}^{n}
			\sum_{i_3 = 1}^{n}
		}_{
			i_1 = i_2 < i_3
		}
		u_{i_1}u_{i_2}u_{i_3}
		\frac{
			\theta_p(q^{-1}z_3/z_1)
			\theta_p(qt^{-1}z_3/z_1)
		}{
			\theta_p(t^{-1}z_3/z_1)
			\theta_p(z_3/z_1)
		}
		\frac{
			\theta_p(q^{-1}z_3/z_2)
			\theta_p(qt^{-1}z_3/z_2)
		}{
			\theta_p(t^{-1}z_3/z_2)
			\theta_p(z_3/z_2)
		}
		\\
		&+
		\underbrace{
			\sum_{i_1 = 1}^{n}
			\sum_{i_2 = 1}^{n}
			\sum_{i_3 = 1}^{n}
		}_{
			i_1 < i_2 = i_3
		}
		u_{i_1}u_{i_2}u_{i_3}
		\frac{
		\theta_p(q^{-1}z_2/z_1)
		\theta_p(qt^{-1}z_2/z_1)
		}{
		\theta_p(t^{-1}z_2/z_1)
		\theta_p(z_2/z_1)
		}
		\frac{
		\theta_p(q^{-1}z_3/z_1)
		\theta_p(qt^{-1}z_3/z_1)
		}{
		\theta_p(t^{-1}z_3/z_1)
		\theta_p(z_3/z_1)
		}
		\\
		&+
		\underbrace{
			\sum_{i_1 = 1}^{n}
			\sum_{i_2 = 1}^{n}
			\sum_{i_3 = 1}^{n}
		}_{
			i_1 < i_2 < i_3
		}
		u_{i_1}u_{i_2}u_{i_3}
		\frac{
			\theta_p(q^{-1}z_2/z_1)
			\theta_p(qt^{-1}z_2/z_1)
		}{
			\theta_p(t^{-1}z_2/z_1)
			\theta_p(z_2/z_1)
		}
		\frac{
			\theta_p(q^{-1}z_3/z_1)
			\theta_p(qt^{-1}z_3/z_1)
		}{
			\theta_p(t^{-1}z_3/z_1)
			\theta_p(z_3/z_1)
		}
		\frac{
			\theta_p(q^{-1}z_3/z_2)
			\theta_p(qt^{-1}z_3/z_2)
		}{
			\theta_p(t^{-1}z_3/z_2)
			\theta_p(z_3/z_2)
		}
		\\
		&+
		\underbrace{
			\sum_{i_1 = 1}^{n}
			\sum_{i_2 = 1}^{n}
			\sum_{i_3 = 1}^{n}
		}_{
			i_1 < i_3 < i_2
		}
		u_{i_1}u_{i_2}u_{i_3}
		\frac{
			\theta_p(q^{-1}z_2/z_1)
			\theta_p(qt^{-1}z_2/z_1)
		}{
			\theta_p(t^{-1}z_2/z_1)
			\theta_p(z_2/z_1)
		}
		\frac{
			\theta_p(q^{-1}z_3/z_1)
			\theta_p(qt^{-1}z_3/z_1)
		}{
			\theta_p(t^{-1}z_3/z_1)
			\theta_p(z_3/z_1)
		}
		\frac{
			\theta_p(qz_3/z_2)
			\theta_p(q^{-1}tz_3/z_2)
		}{
			\theta_p(tz_3/z_2)
			\theta_p(z_3/z_2)
		}
		\\
		&+ \text{ other terms }
	\end{align*}
	Choose $
	\displaystyle
	\cals{N}_{(2,1)}(z_1,z_2,z_3 ; p) = 
	\frac{
		\theta_p(t^{-1}z_3/z_1)
		\theta_p(z_3/z_1)}{
		\theta_p(q^{-1}z_3/z_1)
		\theta_p(qt^{-1}z_3/z_1)}
	\frac{
		\theta_p(t^{-1}z_3/z_2)
		\theta_p(z_3/z_2)}{
		\theta_p(q^{-1}z_3/z_2)
		\theta_p(qt^{-1}z_3/z_2)
	}
	$. Then, 
	\begin{align*}
		&\lim_{\xi \rightarrow t} \,\,
		(\widetilde{\psi}^{(\xi)}_{(2,1)}			
		\comp 
		\bigg|_{
			\substack{
				q_1 = q, \\
				q_2 = q^{-1}t,\\
				q_3 = t^{-1} \\
			}
		}
		)
		\left(
		\cals{N}_{(2,1)}(z_1,z_2,z_3 ; p) 
		\times
		\prod_{1 \leq i < j \leq 3}f^{(3^n)}_{11}\left(\frac{z_j}{z_i} ; p\right)
		\times
		\langle 0 |\widetilde{T}^{(3^n)}_{1}(z_1 ; p)\widetilde{T}^{(3^n)}_{1}(z_2 ; p) \widetilde{T}^{(3^n)}_{1}(z_3 ; p)|0\rangle
		\right)
		\\
		&= 
		\underbrace{
			\sum_{i_1 = 1}^{n}
			\sum_{i_2 = 1}^{n}
			\sum_{i_3 = 1}^{n}
		}_{
			i_1 = i_2 < i_3
		}
		u_{i_1}u_{i_2}u_{i_3}
		+
		\underbrace{
			\sum_{i_1 = 1}^{n}
			\sum_{i_2 = 1}^{n}
			\sum_{i_3 = 1}^{n}
		}_{
			i_1 < i_2 = i_3
		}
		u_{i_1}u_{i_2}u_{i_3}
		+
		\underbrace{
			\sum_{i_1 = 1}^{n}
			\sum_{i_2 = 1}^{n}
			\sum_{i_3 = 1}^{n}
		}_{
			i_1 < i_2 < i_3
		}
		u_{i_1}u_{i_2}u_{i_3}
		\frac{
		\theta_p( q^{-2})
		\theta_p( t^{-1})
		}{
		\theta_p(q^{-1}t^{-1})
		\theta_p(q^{-1})
		}
		\\
		&+
		\underbrace{
			\sum_{i_1 = 1}^{n}
			\sum_{i_2 = 1}^{n}
			\sum_{i_3 = 1}^{n}
		}_{
			i_1 < i_3 < i_2
		}
		u_{i_1}u_{i_2}u_{i_3}
		\frac{
		\theta_p(q^{-2})
		\theta_p(t^{-1})
		}{
		\theta_p(q^{-1}t^{-1})
		\theta_p(q^{-1})
		}
		\frac{
			\theta_p(q^2t)
			\theta_p(t^2)
		}{
		\theta_p(qt^2)
		\theta_p(qt)
		}
		\frac{
			\theta_p(q)
			\theta_p(qt)
		}{
			\theta_p(t)
			\theta_p(q^2)
		}
		\\
		&
		= 
		P_{(2,1)}(u_1,\dots,u_n;q,t;p)
	\end{align*}
\end{exa}

\subsection{Conjecture on partially symmetric polynomials}

In the previous subsection, we have stated a conjecture relating the elliptic Macdonald polynomials to the correlation function of the currents of elliptic corner VOA whose $\vec{c} = (3^n)$. In this subsection, we state a conjecture related to the elliptic corner VOA in more general $\vec{c}$. To write the statement of the conjecture, the notion of partially symmetric polynomials defined below is necessary. 

\begin{dfn}
Let $g(x_1,\dots,x_n) := \sum_{(i_1,\dots,i_n)}c_{i_1,\dots,i_n}x_1^{i_1}\cdots x_n^{i_n}$ be a formal power series, and let $I_1, \dots, I_\ell$ be subsets of $\{1,\dots,n\}$ such that $I_1 \sqcup \cdots \sqcup I_\ell = \{1,\dots,n\}$. For each $k \in \{1,\dots,n\}$, let $\widetilde{I}_k = \{x_a \,| \, a \in I_k\}$. We say that the formal power series 
$g(x_1,\dots,x_n)$ is partially symmetric with respect to the set of variables $\widetilde{I}_1,\dots,\widetilde{I}_k$ if for any $\sigma \in S_n$ which satisfies the condition that $\sigma(I_k) = I_k$ for all 
$k \in \{1,\dots,n\}$, we have $g(x_{\sigma(1)},\dots,x_{\sigma(n)}) = g(x_1,\dots,x_n)$. 
\end{dfn}

Now we are ready to state our conjecture. 

\begin{conj}
\label{conj64-1623}
Let $\vec{c} = (c_1,\dots,c_n) = (1^L2^M3^N)$ and suppose that $q_{\vec{c}} \neq 1$. Then, 
\begin{align}
	\prod_{1 \leq i < j \leq m}f^{\vec{c}}_{11}\left(\frac{z_j}{z_i} ; p\right)
	\times
	\langle 0 |\widetilde{T}^{\vec{c}}_{1}(z_1 ; p)\cdots \widetilde{T}^{\vec{c}}_{1}(z_m ; p)|0\rangle
\end{align}
is symmetric in $z_1,\dots,z_m$ and partially symmetric with respect to the following set of variables 
\begin{align*}
	\{	u_{1},\dots,u_{L}	\}, \hspace{0.3cm}
	\{	u_{L+1},\dots,u_{L+M}	\}, \hspace{0.3cm}
	\{	u_{L+M+1},\dots,u_{L+M+N}	\}
\end{align*}
\end{conj}

One can try to prove this conjecture by proceeding as in the proof of \textbf{Theorem \ref{thm61-1624}}, but this time $\epsilon_n$ defined in \eqref{eqn61-1810-16t} have to be changed to the \say{coloured} one $\epsilon^{(i)}_n$. In general, the star product among the coloured $\epsilon^{(i)}_n$ is not commutative. 
However, if both $\epsilon^{(i)}_n, \epsilon^{(j)}_m$ have the same colour, i.e. $c_i = c_j$, then the commutativity still holds. This is the rough idea of the proof of this conjecture. 
The precise proof and the study of the combinatorial properties of this partially symmetric polynomials will be reported in our 
future work.

\appendix

\section{Proof of lemma \ref{thm44-1334}}
\label{appA-1323}

In this appendix, we provide a detailed proof of the \textbf{Lemma \ref{thm44-1334}}. This prove is conducted by the method of mathematical induction on $r$. 

\subsection{Basis step}

We begin with the case where $r = 1$. In this case, we have to show that 
\begin{align}
	\vec{P}^{(n)}_{m+1} 
	= 
	\bigsqcup_{p = 1}^{n}
	\underbrace{			\bigsqcup_{j \in \bb{Z}^{\geq 0}}				}_{
		1 \leq j \leq p 
	}
	\varphi^{(1)}_{j}\big(		\bb{I}^{(n)}_{m,p}			\big). 
	\label{A1-eqn}
\end{align}
It is clear that 
\begin{align}
\vec{P}^{(n)}_{m+1} 
\supseteq
\bigcup_{p = 1}^{n}
\underbrace{			\bigcup_{j \in \bb{Z}^{\geq 0}}				}_{
	1 \leq j \leq p 
}
\varphi^{(1)}_{j}\big(		\bb{I}^{(n)}_{m,p}			\big)
\label{eqna2-2034}
\end{align}
Next, let $(m_1,\dots,m_n) \in \vec{P}^{(n)}_{m+1}$, and let $k$ be the smallest positive integer such that $m_k \neq 0$. 

In the case $m_k \geq 2$, we then obtain that 
\begin{align}
(m_1,\dots,m_n) = \varphi^{(1)}_{k}(m_1,\dots,m_{k} - 1, \dots, m_n) \in \varphi^{(1)}_{k}(\bb{I}^{(n)}_{m,k}).
\label{eqna3-2034}
\end{align}
Equation \eqref{eqna3-2034} together with \eqref{eqna2-2034} allows us to conclude that 
\begin{align}
	\vec{P}^{(n)}_{m+1} 
	= 
	\bigcup_{p = 1}^{n}
	\underbrace{			\bigcup_{j \in \bb{Z}^{\geq 0}}				}_{
		1 \leq j \leq p 
	}
	\varphi^{(1)}_{j}\big(		\bb{I}^{(n)}_{m,p}			\big). 
\end{align}

In the case $m_k = 1$, let $k^\prime$ be the smallest positive which is larger than $k$ such that $m_{k^\prime} \neq 0$. We then get that 
\begin{align}
	(m_1,\dots,m_n) = \varphi^{(1)}_{k}(m_1,\dots,m_{k} - 1, \dots, m_n) \in \varphi^{(1)}_{k}(\bb{I}^{(n)}_{m,k^\prime}). 
	\label{eqna5-2042}
\end{align}
Similarly, equation \eqref{eqna5-2042} together with \eqref{eqna2-2034} allows us to conclude that 
\begin{align}
	\vec{P}^{(n)}_{m+1} 
	= 
	\bigcup_{p = 1}^{n}
	\underbrace{			\bigcup_{j \in \bb{Z}^{\geq 0}}				}_{
		1 \leq j \leq p 
	}
	\varphi^{(1)}_{j}\big(		\bb{I}^{(n)}_{m,p}			\big). 
\end{align}
So we have shown that in any case $\vec{P}^{(n)}_{m+1}$ can be written as a union of $\varphi^{(1)}_{j}\big(		\bb{I}^{(n)}_{m,p}			\big)$. 

Next we show that this union is the disjoint union. So, we have to show that for any two different choices $(j,p) \neq (j^\prime,p^\prime)$, $\varphi^{(1)}_{j}\big(		\bb{I}^{(n)}_{m,p}			\big) \cap 
\varphi^{(1)}_{j^\prime}\big(		\bb{I}^{(n)}_{m,p^\prime}			\big) = \emptyset. $
Suppose for contradiction that $\varphi^{(1)}_{j}\big(		\bb{I}^{(n)}_{m,p}			\big) \cap 
\varphi^{(1)}_{j^\prime}\big(		\bb{I}^{(n)}_{m,p^\prime}			\big) \neq  \emptyset$, and pick $(m_1,\dots,m_n) \in \varphi^{(1)}_{j}\big(		\bb{I}^{(n)}_{m,p}			\big) \cap 
\varphi^{(1)}_{j^\prime}\big(		\bb{I}^{(n)}_{m,p^\prime}			\big)$. 
Let $k$ be the smallest position such that $m_k \neq 0$. 

In the case $m_k \geq 2$, we get that $j = p = k$ and $j^\prime = p^\prime = k$. In particular, we get that $(j,p) = (j^\prime,p^\prime)$. This leads to \textbf{contradiction} since we assumed that $(j,p) \neq (j^\prime,p^\prime)$. 

In the case $m_k = 1$, we get that $j = j^\prime = k$. 
From the assumption that $(j,p) \neq (j^\prime,p^\prime)$, we get that $p \neq p^\prime$. 
Since $(m_1,\dots,m_n) \in \varphi^{(1)}_{k}\big(		\bb{I}^{(n)}_{m,p}			\big) \cap 
\varphi^{(1)}_{k}\big(		\bb{I}^{(n)}_{m,p^\prime}			\big)$, there exist 
$(\widetilde{m}_1,\dots,\widetilde{m}_n) \in \bb{I}^{(n)}_{m,p}	$ and 
$(m^\prime_1,\dots,m^\prime_n) \in 	\bb{I}^{(n)}_{m,p^\prime}	$ such that 
\begin{align*}
	\varphi^{(1)}_{k}(\widetilde{m}_1,\dots,\widetilde{m}_n) =
	(m_1,\dots,m_n)
	=
	\varphi^{(1)}_{k}(m^\prime_1,\dots,m^\prime_n).
\end{align*}
Since $\varphi^{(1)}_{k}$ is injective, we then obtain that $(\widetilde{m}_1,\dots,\widetilde{m}_n) = (m^\prime_1,\dots,m^\prime_n)$. In particular, $\bb{I}^{(n)}_{m,p} \cap \bb{I}^{(n)}_{m,p^\prime} \neq \emptyset$. 
On the other hand, we know that if $p_1 \neq p_2$, then $\bb{I}^{(n)}_{m,p_1} \cap \bb{I}^{(n)}_{m,p_2} = \emptyset$. So, $\bb{I}^{(n)}_{m,p} \cap \bb{I}^{(n)}_{m,p^\prime} \neq \emptyset$ implies that $p = p^\prime$. This leads to \textbf{contradiction}.

Thus, we have proved the statement when $r = 1$. 

\subsection{Inductive step}

Now assume that the statement in \textbf{Theorem \ref{thm44-1334}}  is true for $1,2,\dots,r-1$ where $r \geq 2$ and for any $m \in \bb{Z}^{\geq 1}$. We will show that 
\begin{align}
	\vec{P}^{(n)}_{m+r} = 
	\bigsqcup_{k = 1}^{r}
	\bigsqcup_{p = 1}^{n}
	\underbrace{	\bigsqcup_{j_1,j_2,\dots,j_k	\in \bb{Z}^{\geq 1} }			}_{
		\substack{	
			\\
			1 \leq j_1 < j_2 < \cdots < j_k \leq p
		}
	}
	\underbrace{		\bigsqcup_{a_{j_1}, a_{j_2}, \dots, a_{j_k} \in \bb{Z}^{\geq 1}}		}_{
		\substack{	
		a_{j_1} + a_{j_2} + \cdots + a_{j_k} = r
		}
	}
	\varphi^{(	a_{j_1}, a_{j_2}, \dots, a_{j_k}		)}_{j_1,j_2,\dots,j_k}
	(\bb{I}^{(n)}_{m,p}). 
	\label{A5-1245}
\end{align}
As in the basis step, we have to show that $\vec{P}^{(n)}_{m+r}$ can be written as the unions in \eqref{A5-1245}, and show that these unions are disjoint unions. Let us first show the union part. 

It is clear that 
\begin{align}
	\vec{P}^{(n)}_{m+r} 
	\supseteq 
	\bigcup_{k = 1}^{r}
	\bigcup_{p = 1}^{n}
	\underbrace{	\bigcup_{j_1,j_2,\dots,j_k	\in \bb{Z}^{\geq 1} }			}_{
		\substack{	
			\\
			1 \leq j_1 < j_2 < \cdots < j_k \leq p
		}
	}
	\underbrace{		\bigcup_{a_{j_1}, a_{j_2}, \dots, a_{j_k} \in \bb{Z}^{\geq 1}}		}_{
		\substack{	
a_{j_1} + a_{j_2} + \cdots + a_{j_k} = r
		}
	}
	\varphi^{(	a_{j_1}, a_{j_2}, \dots, a_{j_k}		)}_{j_1,j_2,\dots,j_k}
	(\bb{I}^{(n)}_{m,p}). 
	\label{A6-1455}
\end{align}
Now pick $(m_1,\dots,m_n) \in \vec{P}^{(n)}_{m+r}$, and let $\ell_1 < \dots < \ell_k < \alpha$ be the positive integers which satisfies the following conditions : 
\begin{itemize}
	\item $m_{\ell_1} \neq 0, m_{\ell_2} \neq 0, \dots, m_{\ell_k} \neq 0, m_{\alpha} \neq 0$
	\item $m_i = 0$ for $i \in \{	1, 2, \dots, \alpha			\}\backslash \{\ell_1,\ell_2,\dots,\alpha\}$
	\item $m_{\ell_1} + \cdots + m_{\ell_k} \geq r$ but $m_{\ell_1} + \cdots + m_{\ell_{k-1}} < r$
\end{itemize}
It is not hard to see that this sequence of positive integers $\ell_1 < \dots < \ell_k < \alpha$ exists. 

In the case $m_{\ell_1} + \cdots + m_{\ell_k} > r$, then we get that 
\begin{align*}
	(m_1,\dots,m_n) = \varphi^{(	m_{\ell_1}, m_{\ell_2}, \dots, c_{\ell_k}		)}_{\ell_1,\ell_2,\dots,\ell_k}(0,\dots,0,m_{\ell_k} - c_{\ell_k},\dots,m_n) 
	\in 
	\varphi^{(	m_{\ell_1}, m_{\ell_2}, \dots, c_{\ell_k}		)}_{\ell_1,\ell_2,\dots,\ell_k}\big(
	\bb{I}^{(n)}_{m,\ell_k}
	\big),
\end{align*}
where $c_{\ell_k} = r - (m_{\ell_1} + \cdots + m_{\ell_{k-1}})$. 
On the other hand, if $m_{\ell_1} + \cdots + m_{\ell_k} = r$, then 
\begin{align*}
	(m_1,\dots,m_n) = \varphi^{(	m_{\ell_1}, m_{\ell_2}, \dots, m_{\ell_k}		)}_{\ell_1,\ell_2,\dots,\ell_k}(0,\dots,0,m_{\alpha},\dots,m_n) 
	\in 
	\varphi^{(	m_{\ell_1}, m_{\ell_2}, \dots, m_{\ell_k}		)}_{\ell_1,\ell_2,\dots,\ell_k}\big(	\bb{I}^{(n)}_{m,\alpha}		\big).
\end{align*}
So we have shown that 
\begin{align}
	\vec{P}^{(n)}_{m+r} 
	= 
	\bigcup_{k = 1}^{r}
	\bigcup_{p = 1}^{n}
	\underbrace{	\bigcup_{j_1,j_2,\dots,j_k	\in \bb{Z}^{\geq 1} }			}_{
		\substack{	
			\\
			1 \leq j_1 < j_2 < \cdots < j_k \leq p
		}
	}
	\underbrace{		\bigcup_{a_{j_1}, a_{j_2}, \dots, a_{j_k} \in \bb{Z}^{\geq 1}}		}_{
		\substack{	
a_{j_1} + a_{j_2} + \cdots + a_{j_k} = r
		}
	}
	\varphi^{(	a_{j_1}, a_{j_2}, \dots, a_{j_k}		)}_{j_1,j_2,\dots,j_k}
	(\bb{I}^{(n)}_{m,p}). 
\label{a9-eqn-2103}
\end{align}

Next, we show the unions in \eqref{a9-eqn-2103} are disjoint unions. To show this, we have to show the statements (2.\ref{2.1.1521}) - (2.\ref{2.1.1523}) below. 
\begin{enumerate}[(2.1)]
	\item \label{2.1.1521}
	For any two different choices $(k,p) \neq (k^\prime,p^\prime)$, we have 
	\begin{align*}
		\varphi^{(a_{j_1},\dots,a_{j_k})}_{j_1,\dots,j_k}(\bb{I}^{(n)}_{m,p}) \,\, \cap \,\, 
		\varphi^{(b_{\ell_1},\dots,b_{\ell_{k^\prime}})}_{\ell_1,\dots,\ell_{k^\prime}}(\bb{I}^{(n)}_{m,p^\prime}) = \emptyset. 
	\end{align*}
	for any $(a_{j_1},\dots,a_{j_k}), (b_{\ell_1},\dots,b_{\ell_{k^\prime}})$ satisfying the following conditions true:
	\begin{itemize}
		\item $a_{j_1}, \dots, a_{j_k} \neq 0$ and $a_{j_1} + \cdots + a_{j_k} = r$
		\item $b_{\ell_1}, \dots, b_{\ell_k^{\prime}} \neq 0$ and $b_{\ell_1} + \cdots + b_{\ell_{k^\prime}} = r$
	\end{itemize}
	\item \label{2.1.1522}
	For any fixed $(k,p)$, if $1 \leq j_1 < j_2 < \cdots < j_k \leq p$ and $1 \leq \ell_1 < \ell_2 < \cdots < \ell_k \leq p$ are different, then we have 
	\begin{align*}
		\varphi^{(a_{j_1},\dots,a_{j_k})}_{j_1,\dots,j_k}(\bb{I}^{(n)}_{m,p}) \,\, \cap \,\, 
		\varphi^{(b_{\ell_1},\dots,b_{\ell_{k}})}_{\ell_1,\dots,\ell_{k}}(\bb{I}^{(n)}_{m,p}) = \emptyset. 
	\end{align*}
	for any $(a_{j_1},\dots,a_{j_k}), (b_{\ell_1},\dots,b_{\ell_{k}})$ satisfying the following conditions true:
	\begin{itemize}
		\item $a_{j_1}, \dots, a_{j_k} \neq 0$ and $a_{j_1} + \cdots + a_{j_k} = r$
		\item $b_{\ell_1}, \dots, b_{\ell_k} \neq 0$ and $b_{\ell_1} + \cdots + b_{\ell_{k}} = r$
	\end{itemize}
	\item \label{2.1.1523}
	For any $(k,p)$ and $1 \leq j_1 < j_2 < \cdots < j_k \leq p$, if 
	$(a_{j_1},\dots,a_{j_k}) \neq (b_{j_1},\dots,b_{j_k})$, then 
	\begin{align*}
		\varphi^{(a_{j_1},\dots,a_{j_k})}_{j_1,\dots,j_k}(\bb{I}^{(n)}_{m,p}) 
		\,\, \cap \,\, 
		\varphi^{(b_{j_1},\dots,b_{j_k})}_{j_1,\dots,j_k}(\bb{I}^{(n)}_{m,p}) 
		= 
		\emptyset. 
	\end{align*}
\end{enumerate}

\underline{\textit{Proof of (2.1)} :}
Suppose for contradiction that $\varphi^{(a_{j_1},\dots,a_{j_k})}_{j_1,\dots,j_k}(\bb{I}^{(n)}_{m,p}) \,\, \cap \,\, 
\varphi^{(b_{\ell_1},\dots,b_{\ell_{k^\prime}})}_{\ell_1,\dots,\ell_{k^\prime}}(\bb{I}^{(n)}_{m,p^\prime}) \neq \emptyset$, and pick 
$(m_1,\dots,m_n) \in \varphi^{(a_{j_1},\dots,a_{j_k})}_{j_1,\dots,j_k}(\bb{I}^{(n)}_{m,p}) \,\, \cap \,\, 
\varphi^{(b_{\ell_1},\dots,b_{\ell_{k^\prime}})}_{\ell_1,\dots,\ell_{k^\prime}}(\bb{I}^{(n)}_{m,p^\prime})$. 
In particular, $(m_1,\dots,m_n) \in \varphi^{(a_{j_1},\dots,a_{j_k})}_{j_1,\dots,j_k}(\bb{I}^{(n)}_{m,p})$. This, in turn, tells us that $(m_1,\dots,m_n)$ have to lie in one of the following cases:

\textbf{Case $1$ : $k = 1$}

In this case, we get that $a_{j_1} = r$. Recall that now we consider the inductive step in which $r > 1$. So, we can regard 
$(m_1,\dots,m_n) \in \varphi^{(r-1)}_{j_1}(		\bb{I}^{(n)}_{m+1,j_1}			)$.

\textbf{Case $2$ : $k > 1$ and $a_{j_k} \geq 2$}

In this case, we can regard $(m_1,\dots,m_n) \in \varphi^{(a_{j_1},\dots,a_{j_k} - 1)}_{j_1,\dots,j_k}(\bb{I}^{(n)}_{m+1,j_k})$.

\textbf{Case $3$ : $k > 1$ and $a_{j_k} = 1$}

In this case, we can regard $(m_1,\dots,m_n) \in \varphi^{(a_{j_1},\dots,a_{j_{k-1}})}_{j_1,\dots,j_{k-1}}(\bb{I}^{(n)}_{m+1,j_k})$. 

\vspace{0.3cm}

By using the similar argument, we can see that $(m_1,\dots,m_n) \in \varphi^{(b_{\ell_1},\dots,b_{\ell_{k^\prime}})}_{\ell_1,\dots,\ell_{k^\prime}}(\bb{I}^{(n)}_{m,p^\prime})$ implies that 
$(m_1,\dots,m_n)$ have to lie in one of the following cases:

\textbf{Case $A$ : $k^\prime = 1$}

In this case, we can regard $(m_1,\dots,m_n) \in \varphi^{(r-1)}_{\ell_1}(		\bb{I}^{(n)}_{m+1,\ell_1}			)$.

\textbf{Case $B$ : $k^\prime > 1$ and $b_{\ell_{k^\prime}} \geq 2$}

In this case, we can regard $(m_1,\dots,m_n) \in \varphi^{(b_{\ell_1},\dots,b_{\ell_{k^\prime}} - 1)}_{\ell_1,\dots,\ell_{k^\prime}}(\bb{I}^{(n)}_{m+1,\ell_{k^\prime}})$. 

\textbf{Case $C$ : $k^\prime > 1$ and $b_{\ell_{k^\prime}} = 1$}

In this case, we can regard $(m_1,\dots,m_n) \in \varphi^{(b_{\ell_1},\dots,b_{\ell_{k^\prime - 1}} )}_{\ell_1,\dots,\ell_{k^\prime - 1}}(\bb{I}^{(n)}_{m+1,\ell_{k^\prime}})$. 

\vspace{0.3cm}
We have to show that for any $i \in \{1,2,3\}$ and $j \in \{A,B,C\}$, \textbf{Case $i$} $\cap$ \textbf{Case $j$} always leads to contradiction. Since the idea in the proof of various cases are similar, here we show only the cases $(i,j) = (2,A), (2,B), (2,C)$. 

In the case $(i,j) = (2,A)$, we see that $(m_1,\dots,m_n) \in \varphi^{(a_{j_1},\dots,a_{j_k} - 1)}_{j_1,\dots,j_k}(\bb{I}^{(n)}_{m+1,j_k})$ ($k \geq 1$ and $a_{j_k} \geq 2$), and $(m_1,\dots,m_n) \in \varphi^{(r-1)}_{\ell_1}(		\bb{I}^{(n)}_{m+1,\ell_1}			)$. From induction hypothesis, we obtain that 
\begin{align*}
	\varphi^{(a_{j_1},\dots,a_{j_k} - 1)}_{j_1,\dots,j_k}(\bb{I}^{(n)}_{m+1,j_k})
	\,\, \cap \,\, 
	\varphi^{(r-1)}_{\ell_1}(		\bb{I}^{(n)}_{m+1,\ell_1}			)
	= 
	\emptyset
\end{align*}
This leads to \textbf{contradiction}.

In the case $(i,j) = (2,B)$, we can regard $(m_1,\dots,m_n) \in \varphi^{(a_{j_1},\dots,a_{j_k})}_{j_1,\dots,j_k}(\bb{I}^{(n)}_{m,p}) \,\, \cap \,\, 
\varphi^{(b_{\ell_1},\dots,b_{\ell_{k^\prime}})}_{\ell_1,\dots,\ell_{k^\prime}}(\bb{I}^{(n)}_{m,p^\prime})$
as 
\begin{align}
	(m_1,\dots,m_n) \in 
	\varphi^{(a_{j_1},\dots,a_{j_k} - 1)}_{j_1,\dots,j_k}(\bb{I}^{(n)}_{m+1,j_k}) 
	\,\, \cap \,\, 
	\varphi^{(b_{\ell_1},\dots,b_{\ell_{k^\prime}} - 1)}_{\ell_1,\dots,\ell_{k^\prime}}(\bb{I}^{(n)}_{m+1,\ell_{k^\prime}})
\end{align}
where $k > 1, a_{j_k} \geq 2$ and $k^\prime > 1, b_{\ell_{k^\prime}} \geq 2$. If $k \neq k^\prime$, we can immediately conclude from the induction hypothesis that 
\begin{align}
	\varphi^{(a_{j_1},\dots,a_{j_k} - 1)}_{j_1,\dots,j_k}(\bb{I}^{(n)}_{m+1,j_k}) 
	\,\, \cap \,\, 
	\varphi^{(b_{\ell_1},\dots,b_{\ell_{k^\prime}} - 1)}_{\ell_1,\dots,\ell_{k^\prime}}(\bb{I}^{(n)}_{m+1,\ell_{k^\prime}})
	= \emptyset,
\end{align}
and this leads to contradiction. So, we get that $k = k^\prime$. Since $(k,p) \neq (k^\prime,p^\prime)$, we obtain that $p \neq p^\prime$. Without loss of generality, assume that $p < p^\prime$. We know that 
\begin{align}
	\varphi^{(a_{j_1},\dots,a_{j_k})}_{j_1,\dots,j_k}(\bb{I}^{(n)}_{m,p}) \,\, \cap \,\, 
	\varphi^{(b_{\ell_1},\dots,b_{\ell_{k}})}_{\ell_1,\dots,\ell_{k}}(\bb{I}^{(n)}_{m,p^\prime}) \neq \emptyset
	\label{a12-2139}
\end{align}
because it contains $(m_1,\dots,m_n)$ as an element. 
Since $p < p^\prime$, we get that $j_k = \ell_{k} = p$. Otherwise, the intersection \eqref{a12-2139} is empty. 
From the fact that $j_k = \ell_{k} = p$, we can conclude that for any $i \in \{1,\dots,k-1\}$, $j_i = \ell_{i} \text{ and } a_{j_i} = b_{\ell_{i}}$. 
Since $a_{j_1} + \cdots + a_{j_k} = b_{\ell_1} + \cdots + b_{\ell_{k}}$, we conclude that $a_{j_k} = b_{\ell_k}$, i.e. $a_p = b_p$. 

On the other hand, from the fact that $(m_1,\dots,m_n) \in \varphi^{(a_{j_1},\dots,a_{j_k})}_{j_1,\dots,j_k}(\bb{I}^{(n)}_{m,p}) \,\, \cap \,\, 
\varphi^{(b_{\ell_1},\dots,b_{\ell_{k}})}_{\ell_1,\dots,\ell_{k}}(\bb{I}^{(n)}_{m,p^\prime})$, there exists $m^\prime_p \in \bb{Z}^{\geq 1}$ such that $m_p = m^\prime_p + a_p = b_p$. So, $a_p \neq b_p$, and this leads to \textbf{contradiction}. 

In the case $(i,j) = (2,C)$, we see that $(m_1,\dots,m_n) \in \varphi^{(a_{j_1},\dots,a_{j_k} - 1)}_{j_1,\dots,j_k}(\bb{I}^{(n)}_{m+1,j_k}) \,\, \cap \,\, \varphi^{(b_{\ell_1},\dots,b_{\ell_{k^\prime - 1}} )}_{\ell_1,\dots,\ell_{k^\prime - 1}}(\bb{I}^{(n)}_{m+1,\ell_{k^\prime}})$. If $j_k \neq \ell_{k^\prime}$, we can immediately conclude from induction hypothesis that 
\begin{align*}
	\varphi^{(a_{j_1},\dots,a_{j_k} - 1)}_{j_1,\dots,j_k}(\bb{I}^{(n)}_{m+1,j_k}) \,\, \cap \,\, \varphi^{(b_{\ell_1},\dots,b_{\ell_{k^\prime - 1}} )}_{\ell_1,\dots,\ell_{k^\prime - 1}}(\bb{I}^{(n)}_{m+1,\ell_{k^\prime}})
	= 
	\emptyset.
\end{align*}
Thus, $j_k = \ell_{k^\prime}$. However, since $\ell_{k^\prime - 1} < \ell_{k^\prime} = j_k$, the tuples $(j_1,\dots,j_k)$ and $(\ell_1,\dots,\ell_{k^\prime - 1})$ are different. Thus, from induction hypothesis, we get that the intersection must be empty. This leads to \textbf{contradiction}.

\underline{\textit{Proof of (2.2)} :}
Since the sequences $1 \leq j_1 < j_2 < \cdots < j_k \leq p$ and $1 \leq \ell_1 < \ell_2 < \cdots < \ell_k \leq p$ are different, there exists $t \in \{1,\dots,k\}$ such that 
\begin{itemize}
	\item $j_t \neq \ell_t$
	\item $j_1 = \ell_1, \dots, j_{t-1} = \ell_{t-1}$
\end{itemize}
Suppose for contradiction that $\varphi^{(a_{j_1},\dots,a_{j_k})}_{j_1,\dots,j_k}(\bb{I}^{(n)}_{m,p}) \,\, \cap \,\, 
\varphi^{(b_{\ell_1},\dots,b_{\ell_{k}})}_{\ell_1,\dots,\ell_{k}}(\bb{I}^{(n)}_{m,p}) \neq \emptyset$, and pick 
\begin{align*}
	(m_1,\dots,m_n) \in \varphi^{(a_{j_1},\dots,a_{j_k})}_{j_1,\dots,j_k}(\bb{I}^{(n)}_{m,p}) \,\, \cap \,\, 
	\varphi^{(b_{\ell_1},\dots,b_{\ell_{k}})}_{\ell_1,\dots,\ell_{k}}(\bb{I}^{(n)}_{m,p}). 
\end{align*}

If $t \neq k$, then we get that 
\begin{itemize}
	\item $j_1 = \ell_1 < p \,\, , \dots, \,\, j_{t - 1} = \ell_{t - 1} < p$,
	\item $j_t \neq \ell_t$ and $j_t < p, \ell_t < p$. 
\end{itemize}
Without loss of generality, we assume that $j_t < \ell_t$. Since $(m_1,\dots,m_n) \in \varphi^{(a_{j_1},\dots,a_{j_k})}_{j_1,\dots,j_k}(\bb{I}^{(n)}_{m,p})$, we get that $m_{j_t} = a_{j_t}$. On the other hand, since 
$(m_1,\dots,m_n) \in \varphi^{(b_{\ell_1},\dots,b_{\ell_{k}})}_{\ell_1,\dots,\ell_{k}}(\bb{I}^{(n)}_{m,p})$, we obtain that $m_{j_t} = 0$. Thus, $a_{j_t} = 0$. Nevertheless, this \textbf{contradicts} to the fact that $a_{j_t} \neq 0$. 

If $t = k$, then we get that 
\begin{itemize}
	\item $j_1 = \ell_1 < p \,\, , \dots, \,\, j_{k - 1} = \ell_{k - 1} < p$,
	\item $j_{k} \neq \ell_k$.
\end{itemize}
Without loss of generality, we assume that $j_k < \ell_k$. Since $(m_1,\dots,m_n) \in \varphi^{(a_{j_1},\dots,a_{j_k})}_{j_1,\dots,j_k}(\bb{I}^{(n)}_{m,p})$, we get that $m_{j_k} \neq 0$. On the other hand, since $(m_1,\dots,m_n) \in \varphi^{(b_{\ell_1},\dots,b_{\ell_{k}})}_{\ell_1,\dots,\ell_{k}}(\bb{I}^{(n)}_{m,p})$, we obtain that $m_{j_k} = 0$. Therefore, $\varphi^{(a_{j_1},\dots,a_{j_k})}_{j_1,\dots,j_k}(\bb{I}^{(n)}_{m,p}) \,\, \cap \,\, 
\varphi^{(b_{\ell_1},\dots,b_{\ell_{k}})}_{\ell_1,\dots,\ell_{k}}(\bb{I}^{(n)}_{m,p}) = \emptyset$. 
This leads to \textbf{contradiction}. 

\underline{\textit{Proof of (2.3)} :}
Suppose for contradiction that $\varphi^{(a_{j_1},\dots,a_{j_k})}_{j_1,\dots,j_k}(\bb{I}^{(n)}_{m,p}) \,\, \cap \,\, 
\varphi^{(b_{j_1},\dots,b_{j_{k}})}_{j_1,\dots,j_{k}}(\bb{I}^{(n)}_{m,p}) \neq \emptyset$, and choose 
\begin{align*}
	(m_1,\dots,m_n) \in \varphi^{(a_{j_1},\dots,a_{j_k})}_{j_1,\dots,j_k}(\bb{I}^{(n)}_{m,p}) \,\, \cap \,\, 
	\varphi^{(b_{j_1},\dots,b_{j_{k}})}_{j_1,\dots,j_{k}}(\bb{I}^{(n)}_{m,p}).
\end{align*}
So, there exist $(m^\prime_1,\dots,m^\prime_n) \in \bb{I}^{(n)}_{m,p}$ and $(m^{\prime\prime}_1,\dots,m^{\prime\prime}_n) \in \bb{I}^{(n)}_{m,p}$ such that 
\begin{align*}
	&m_{j_1} = m^\prime_{j_1} + a_{j_1} = m^{\prime\prime}_{j_1} + b_{j_1} 
	\\
	&m_{j_2} = m^\prime_{j_2} + a_{j_2} = m^{\prime\prime}_{j_2} + b_{j_2} 
	\\
	&\hspace{2.5cm}\vdots
\end{align*}
If $j_1 < \cdots < j_k < p$, we then get that $(a_{j_1},\dots,a_{j_k}) = (b_{j_1},\dots,b_{j_k})$. This leads to contradiction since we assume in the case (2.3) that $(a_{j_1},\dots,a_{j_k}) \neq (b_{j_1},\dots,b_{j_k})$. Therefore, we get that
\begin{align*}
	j_1 < \cdots < j_{k-1} < p \text{ and } j_k = p.
\end{align*}
Thus, $a_{j_1} = b_{j_1}, \dots, a_{j_{k-1}} = b_{j_{k-1}}$. 
Consequently, $a_{j_k} = b_{j_k}$, and we get that $(a_{j_1},\dots,a_{j_k}) = (b_{j_1},\dots,b_{j_k})$. This leads to \textbf{contradiction}. 

\vspace{0.3cm}

So, we have shown the statement \eqref{A5-1245}. This closes the induction.

	%%%%%%%%%%%%%%%%%%%%%%%%%%%%%%%%%%%%%%%%%%%%%%%%%%%%%%%%%%%%%%%%%%%%%%%%%%
	%\newpage

\end{document}